\documentclass[oneside, a4paper]{amsart}

\usepackage[]{geometry}
\usepackage{pdflscape}
\usepackage{lmodern}
\DeclareFixedFont{\ttb}{T1}{txtt}{bx}{n}{12} 
\DeclareFixedFont{\ttm}{T1}{txtt}{m}{n}{12}  

\usepackage{changepage}
\usepackage{natbib}

\usepackage{float}
\usepackage{xargs}
\usepackage[pdftex]{xcolor}
\definecolor{deepblue}{rgb}{0,0,0.5}
\definecolor{deepred}{rgb}{0.6,0,0}
\definecolor{deepgreen}{rgb}{0,0.5,0}
\definecolor{shadecolor}{rgb}{0.95,0.95,0.95} 
\usepackage{mdframed}
\usepackage{listings}
\usepackage{graphicx}
\usepackage{amscd, amsmath, amsfonts, amssymb, bbm}
\usepackage{setspace}
\usepackage{enumerate}   
\usepackage{url}
\usepackage{amsthm}
\usepackage{bm}
\usepackage{xy}
\usepackage{enumitem}
\usepackage{multicol, multirow, makecell}
\usepackage{hhline}
\usepackage[]{hyperref}
\newcommand{\changeurlcolor}[1]{\hypersetup{urlcolor=#1}} 

\usepackage{cleveref}
\usepackage[foot]{amsaddr}
\usepackage{placeins}

\usepackage{subcaption}
\usepackage{booktabs}
\hypersetup{
    colorlinks,
    linkcolor={blue!100!black},
    citecolor={blue!50!black},
    urlcolor={blue!80!black}
}

\theoremstyle{plain}
\newtheorem{Theorem}{Theorem}[section]
\newtheorem{Algorithm}[Theorem]{Algorithm}
\newtheorem{Notation}[Theorem]{Notation}
\newtheorem{Definition}[Theorem]{Definition}
\newtheorem{Corollary}[Theorem]{Corollary}

\theoremstyle{remark}
\newtheorem{Remark}[Theorem]{Remark}

\numberwithin{equation}{section}

\newcommand{\ind}[1]{\mathbbm{1}_{ \left(#1\right) }}

\numberwithin{equation}{section}

\renewcommand{\rho}{\varrho}

\DeclareMathOperator*{\argmin}{argmin}

\newcommand{\EW}[1]{\mathbb{E} \left[ #1 \right]}

\newcommand{\Var}[1]{\textrm{\normalfont {Var}} \left( #1 \right)}

\begin{document}
\title[Calibration of LSV models with generative adversarial
networks]{A Generative Adversarial Network Approach to Calibration of
Local Stochastic Volatility Models}

\begin{abstract}
  We propose a fully data-driven approach to calibrate local
  stochastic volatility (LSV) models, circumventing in particular the
  ad hoc interpolation of the volatility surface. To achieve this, we~parametrize the leverage function by a family of feed-forward neural
  networks and learn their parameters directly from the available
  market option prices. This should be seen in the context of neural SDEs 
  and (causal) generative adversarial networks: we \emph{generate} volatility surfaces by specific neural SDEs, whose quality is assessed by quantifying, possibly in an \emph{adversarial} manner, distances to market prices. The minimization
  of the calibration functional relies strongly on a variance
  reduction technique based on hedging and deep hedging, which is
  interesting in its own right: it allows the calculation of model prices
  and model implied volatilities in an accurate way using only small sets of
  sample paths. For numerical illustration we implement a SABR-type 
  LSV model and conduct a thorough statistical performance analysis on
  many samples of implied volatility smiles, showing the accuracy and
  stability of the method.
\end{abstract}

\keywords{LSV calibration, neural SDEs, generative
    adversarial networks, deep hedging, variance reduction, stochastic optimization\vspace{2mm}}

\thanks{Christa Cuchiero gratefully
  acknowledges financial support by the Vienna Science and Technology
  Fund (WWTF) under grant MA16-021.\vspace{2mm}\\
  Wahid Khosrawi and Josef Teichmann gratefully
  acknowledge support by the ETH Foundation and the SNF
  Project 179114.\vspace{2mm}\\
Christa Cuchiero\\
  University of Vienna, Department of Statistics and Operations Research, Data Science @ Uni Vienna, Oskar-Morgenstern-Platz 1, 1090 Wien, Austria\\
  E-mail: \texttt{christa.cuchiero@univie.ac.at} \vspace{2mm}\\
  Wahid Khosrawi\\
  ETH Z\"urich, D-MATH, R\"amistrasse 101, CH-8092 Z\"urich, Switzerland\\
  E-mail: \texttt{wahid.khosrawi@math.ethz.ch}\vspace{2mm}\\
  Josef Teichmann\\
  ETH Z\"urich, D-MATH, R\"amistrasse 101, CH-8092 Z\"urich, Switzerland\\
  E-mail: \texttt{josef.teichmann@math.ethz.ch}\vspace{2mm}\\
  Published version: \url{https://www.mdpi.com/2227-9091/8/4/101}\\
  github repository: \url{https://github.com/wahido/neural_locVol}
}

\author{Christa Cuchiero, Wahid Khosrawi and Josef Teichmann}
\vspace*{-0.5cm}
\maketitle

\vspace*{-0.5cm}

\section{Introduction}

Each day a crucial task is performed in financial institutions all
over the world: the calibration of stochastic models to current market
or historical data. So far the model choice was not only driven by the
capacity of capturing empirically observed market features well, but
also by the computational tractability of the calibration process.
This is now undergoing a big change since machine-learning
technologies offer new perspectives on model calibration.

Calibration is the choice of \emph{one} model from a \emph{pool of
  models}, given current market and historical data. Depending on the nature of data this 
is considered to be an inverse problem or a problem of statistical
inference. We consider here current market data, in particular volatility
surfaces, therefore we rather emphasize the inverse problem point of
view. We however stress that it is the ultimate goal of calibration to
include both data sources simultaneously.
In this respect machine learning might help considerably.

We can distinguish three kinds of machine learning-inspired approaches
for calibration to current market prices: First, having solved the
inverse problem already several times, one can learn from this experience (i.e.,~training data) the calibration map from market data
to model parameters directly. Let us here mention one of the pioneering papers
by \cite{H:16} that applied neural networks to learn this
calibration map in the context of interest rate models. This
was taken up in \mbox{\cite{CMMMST:18}} for calibrating more complex mixture models. Second, one can learn the map from model
parameters to model prices (compare e.g. \cite{liu2019neural,liu2019pricing})
and then invert this map possibly with machine
learning technology. In the context of rough volatility modeling,
see~\mbox{\cite{GJR:18}}, such approaches turned out to be very successful:
we refer here to \mbox{\cite{BHMST:19}} and the references therein.
Third, the calibration problem is
considered to be the search for a model which \emph{generates} given market
prices and where additionally technology from generative adversarial networks, first~introduced by \cite{GPMXWOCB:14},  can be
used. This means parameterizing the model pool in a way which is accessible for machine
learning techniques and interpreting the inverse problem as a training
task of a generative network, whose quality is assessed by an \emph{adversary.}
We pursue this approach in the present article and use as generative models so-called neural stochastic differential equations (SDE), which  just means to parameterize the  drift and volatility of an It\^o-SDE by neural networks.

\subsection{Local Stochastic Volatility Models as Neural SDEs}
\label{sec:LSVintro}

We focus here on calibration of \emph{local
  stochastic volatility (LSV) models}, which are in view of existence and uniqueness still an intricate model class.
LSV models, going back to \cite{Jex:99, L:02, RMQ:07}, combine classical
stochastic volatility with local volatility to achieve both a good
fit to time series data and in principle a perfect calibration to the
implied volatility smiles and skews. In~these models, the discounted
price process $(S_t)_{t \geq 0}$ of an asset satisfies
\begin{align}\label{eq:LSV}
  dS_t = S_t L(t,S_t) \alpha_t dW_t,
\end{align}
where $(\alpha_t)_{t \geq 0}$ is some stochastic process taking values
in $\mathbb{R}$, and (a sufficiently regular function) $L(t,s)$ the
so-called \emph{leverage function} depending on time and the current 
value of the asset and $W$ a one-dimensional Brownian motion. Note
that the stochastic volatility process $\alpha$ can be very general
and could for instance  be chosen as rough volatility model. By slight
abuse of terminology we call $\alpha$ stochastic volatility even
though it is strictly speaking not the volatility of the log price of
$S$.

For notational simplicity we consider here the one-dimensional case,
but the setup easily translates to a multivariate situation with
several assets and a matrix valued analog of $\alpha$ as well as a
matrix valued leverage function.

The leverage function $L$ is the crucial part in this model. It allows
in principle to perfectly calibrate the implied volatility surface
seen on the market.  To achieve this goal $L$ must satisfy
\begin{equation}\label{eq: impl dup link}
  L^2(t,s)=\frac{\sigma^2_{\text{Dup}}(t,s)}{\mathbb{E}[ \alpha_t^2 | S_t=s]},
\end{equation}
where $\sigma_{\text{Dup}}$ denotes Dupire's local volatility function
(see \citeauthor{D:94} (\citeyear{D:94}, \citeyear{D:96})). For the derivation of~\eqref{eq: impl dup link}, we
refer to \cite{GH:13}. Please note that \eqref{eq: impl dup link} is an
implicit equation for $L$ as it is needed for the computation of
$\mathbb{E}[ \alpha_t^2 | S_t=s]$.  This in turn means that the SDE for the price process
$(S_t)_{t \geq 0}$ is actually a McKean–Vlasov SDE, since the law of $(S_t, \alpha_t)$ enters in the characteristics of the equation.  Existence and
uniqueness results for this equation are not at all obvious, since the
coefficients do not satisfy any kind of standard conditions like for
instance Lipschitz continuity in the Wasserstein space. Existence of a
short-time solution of the associated nonlinear Fokker-Planck
equation for the density of $(S_t)_{t \geq 0}$ was shown in
\cite{AT:10} under certain regularity assumptions on the initial
distribution. As stated in \cite{GH:11}  a very
challenging and still open problem is to derive the set of stochastic
volatility parameters for which LSV models exist uniquely for a given
market implied volatility surface. We refer to \cite{JZ:16} and  \mbox{\cite{Lacker:19}}, where recent progress in solving this problem has been made.

Despite these intriguing existence issues, LSV models have attracted---due to their appealing feature of a potentially perfect smile
calibration and their econometric properties---a lot of attention from the
calibration and implementation point of view. We refer to \cite{GH:11,
  GH:13, CMR:17} for Monte Carlo (MC) methods (see also \citeauthor{G:14} (\citeyear{G:14}, \citeyear{G:16}) for the multivariate case), to \cite{RMQ:07, TZLKH:15}
for PDE methods based on nonlinear Fokker-Planck equations and to
\cite{SYZ:17} for inverse problem techniques. Within these approaches
the particle approximation method for the McKean–Vlasov SDE proposed
in \cite{GH:11, GH:13} works impressively well, as~very few
paths must be used to achieve very accurate calibration
results.

In the current paper we propose an alternative, fully data-driven
approach circumventing in particular the interpolation of the
volatility surface, being necessary in several other approaches in
order to compute Dupire's local volatility.
This means that we only take the available discrete data into account and do not generate a continuous surface interpolating between the given market option prices. Indeed,  we just
\emph{learn} or \emph{train} the leverage function $L$ to
generate the available market option prices accurately. Although in principle the method allows for calibration to any traded options, we work here with vanilla derivatives.

Setting $T_0=0$ and denoting by
$ T_1 <T_2\cdots < T_n$ the maturities of the available options, we~parametrize the leverage function $L(t,s)$ via a family of
neural networks $F^i: \mathbb{R} \to \mathbb{R}$ with weights
$\theta_i \in \Theta_i$, i.e.
\[
  L(t,s, \theta) = 1+ F^i(s, \theta_i), \quad t\in [T_{i-1}, T_{i}), \quad i \in\{1,
  \ldots, n\}. 
\]

We here consider for simplicity only the univariate case. The
multivariate situation just means that $1$ is replaced by the identity
matrix and the neural networks $F^i(\cdot, \theta_i)$ are maps from
$\mathbb{R}^d \to \mathbb{R}^{d \times d}$.

This then leads to the generative model class of neural SDEs (see \cite{GSSS:20} for related work), which in the case of time-inhomogeneous It\^o-SDEs, just means to parametrize the drift $\mu(\cdot ,\cdot, \theta)$ and volatility $\sigma(\cdot , \cdot ,\theta)$ by neural networks with parameters $\theta$, i.e.,
\begin{align}\label{eq:NNSDE}
  dX_t(\theta)=\mu(X_t(\theta), t ,\theta)dt+\sigma(X_t(\theta),t , \theta)d W_t, \quad X_0(\theta)=x.
\end{align}

In our case, there is no drift and the volatility  (for the price) reads as
\begin{align}\label{eq:NNvol}
  \sigma(S_t(\theta), t, \theta)=S_t(\theta)\left( 1 + \sum_{i=1}^n F^i(S_t(\theta),  \theta_i) 1_{[T_{i-1}, T_{i})}(t) \right)\alpha_t .
\end{align}

Progressively for each maturity, the parameters of the neural networks
are learned by optimizing the following calibration criterion
\begin{align}\label{eq:calicrit}
  \inf_{\theta}\sup_\gamma \sum_{j=1}^{J} w^\gamma_{j} \ell^\gamma (\pi_{j}^{\text{mod}}(\theta) -\pi_{j}^{\text{mkt}}) \, ,
\end{align}
where $J$ is the number of considered options and where $\pi^{\text{mod}}_{j}(\theta)$ and $\pi^{\text{mkt}}_{j}$ stand for
the respective model and market  prices.

Moreover, for every fixed $ \gamma$, $\ell^\gamma$ is a
nonlinear non-negative convex function with $\ell^{\gamma}(0)=0$ and
$\ell^{\gamma}(x) >0 $ for $x \neq 0$, measuring the distance between model and
market prices. The terms
$w^\gamma_j$, for fixed $\gamma$, denote some weights,
in our case of vega type (compare \cite{cont2004recovering}). Using such vega type weights  allows
to match implied volatility data, our actual goal, very well. The parameters $\gamma$ take here the role of the adversarial part. Indeed, by considering a family of weights and loss functions parameterized by $\gamma$  enables us to take into account the uncertainty of the loss function. In this sense this constitutes a discriminative model as it modifies the distribution of the target, i.e.,~the different market prices, given $\theta$ and thus the model prices. This can for instance mean that the adversary chooses the weights $w^{\gamma}$ in such a way to put most mass on those options where the fit is worst or that it modifies $\ell^{\gamma}$ by choosing $\pi^{\text{mkt}}$ within the bid-ask spread with the largest possible distance to the current model prices.
In a concrete implementation in Section 4.3, we build a family of loss functions like this, using different market implied volatilities lying in the bid-ask spread.  
We can therefore  also  solve such kind of robust calibration problems.

The precise algorithms are outlined in Section \ref{sec:LSVcali} and
Section \ref{sec:numerics}, where we also conduct a thorough
statistical performance analysis.
Notice that as is somehow quite typical for financial applications, we need
to guarantee a very high accuracy, whence a variance
reduction technique to compute the model prices via Monte Carlo is crucial for this learning task. 
This relies on hedging and deep hedging, which
allows the computation of accurate model prices $\pi^{\text{mod}}(\theta)$ for
training purposes with only up to $5\times 10^4$ trajectories. Let us  
remark that we do not aim to compete with existing algorithms, as
e.g.~the particle method by \cite{GH:11, GH:13}, in terms of speed but
rather provide a generic data-driven algorithm that is universally
applicable for all kind of options, also in multivariate situations,
without resorting to Dupire type volatilities. This general applicability comes at the expense of a higher computation time compared to \cite{GH:11, GH:13}. In terms of accuracy, we achieve an average calibration error of about 5 to 10 basis points, whence
our method is comparable or in some situations even better than \cite{GH:11} (compare Section 6 and the results in \cite{GH:11}). Moreover, we also observe good  extrapolation and generalization properties of the calibrated leverage~function.

\subsection{Generative Adversarial Approaches in Finance}\label{sec:gen}

The above introduced generative adversarial approach
might seem at first sight unexpected as generative adversarial models or networks are rather applied in areas such as photorealistic image generation. From an abstract point of view, however,
a \emph{generative network} is nothing else than a neural network
(mostly of recurrent type) $ G^\theta $ depending on parameters $ \theta $, which transports a
standard input law $\mathbb{P}_I$ to a target output law $ \mathbb{P}_O$. In our case, $\mathbb{P}_I$ corresponds to the law of the Brownian motion $ W $ and the stochastic
volatility process $\alpha$. The target law $\mathbb{P}_O$ is given in terms of
certain functionals,  namely the set of market prices, and is thus not fully specified.

Denoting the push-forward of $\mathbb{P}_I$ under the transport map $G^{\theta}$ by $ G^\theta_* \mathbb{P}_I $, the goal is to find parameters $\theta$ such that $G^\theta_* \mathbb{P}_I \approx \mathbb{P}_O$. For this purpose, appropriate distance functions must be used. Standard examples include entropies, integral distances,
Wasserstein or Radon distances, etc. The \emph{adversarial character} appears when the chosen distance is represented as
supremum over certain classes of functions, which can themselves be
parameterized via neural networks of a certain type. This leads to a
game, often of zero-sum type, between the generator and the adversary. As mentioned above, one example, well known from industry, is calibrating a model to generate prices within a bid-ask spread: in this case there is more than one loss function, each of them representing a distance to a possible price structure, and the supremum over these loss functions is the actual distance between generated price structure and the target. In other words: the distance to the worst possible price structure within the bid-ask spread should be as small as possible (see Section \ref{sec:robust}).

In our case the solution measure of the neural SDE as specified in \eqref{eq:NNSDE} and \eqref{eq:NNvol} corresponds to the transport $G^{\theta}_{*} \mathbb{P}_I$
and we measure the distance by \eqref{eq:calicrit}, which can be rewritten as
\begin{align}\label{eq:calgen}
   & \inf_{\theta } \sup_{\gamma}
  \sum_{j=1}^J w^{\gamma}_{j} \ell^{\gamma}(\underbrace{\mathbb{E}_{G^{\theta}_{*}\mathbb{P}_I}[C_j]}_{\text{model price}}-\underbrace{\mathbb{E}_{\mathbb{P}_O}[C_j]}_{\text{market price}}),
\end{align}
where $C_j$ are the corresponding option payoffs.

In general, one could consider distance functions $d^{\gamma}$ such that
the game between generator and adversary appears as
\begin{align*}
  \inf_\theta \sup_\gamma d^\gamma(G^\theta_* \mathbb{P}_I,\mathbb{P}_O) \, .
\end{align*}

The advantage of this point of view is two-fold:
\begin{enumerate}
  \item we have access to the unreasonable effectiveness of modeling by
        neural networks,
        due to their good generalization
        and regularization properties;
  \item the game theoretic view disentangles realistic price generation from discriminating with different loss functions, parameterized by $\gamma$. This reflects the fact that it is not necessarily clear which loss function one should use. 
        Notice
        that \eqref{eq:calgen}  is not the usual form of generative adversarial network (GAN) problems, since the adversary 
        distance $ \ell^\gamma $ is nonlinear in $ \mathbb{P}_I $ and
        $ \mathbb{P}_O $, but we believe that it is worth taking this
        abstract point of view.
\end{enumerate}

There is no reason these generative models, if sufficient
computing power is available, should~not take market price data as
inputs, too. This would correspond, from the point of view of
generative adversarial networks, to actually learn a map
$ G^{\theta,\text{market prices}}$, such that for any price
configuration of market prices one has instantaneously a generative
model given, which produces those prices. This~requires just a rich
data source of typical market prices (and computing power!).

Even though it is usually not considered like that, one can also view
the generative model as an engine producing a likelihood function on probability measures on path
space: given historic data, $\mathbb{P}_O$ is  then just the empirical measure of the one observed
trajectory that is inserted in the likelihood function.
This would allow, with precisely the same
technology, a maximum likelihood approach, where one searches for those parameters of the generative network that maximize the likelihood of the historic trajectory.
This then falls in the realm of generative approaches that appear in the literature under the name ``market generators''. Here the goal is to precisely mimic the behavior and features of historical market trajectories. This line of research has been recently pursued in e.g. {\cite{KS:19}}; {\cite{WBWB:19}} \cite{H:19, BHPLW:20, AX:20}.

From a bird's eye perspective this machine-learning approach to calibration might just
look like a standard inverse problem  with another
parameterized family of functions. We, however, insist on one important
difference, namely implicit regularizations (see e.g. \cite{HTW:19}), which always appear
in machine-learning applications and which are cumbersome to mimic  in
classical inverse~problems.

Finally, let us comment more generally on machine-learning approaches in  mathematical finance, which become more and more prolific. Concrete applications include
hedging {\cite{BGTW:19}}, portfolio~selection \cite{GTX:19}, stochastic
portfolio theory {{\cite{SV:16}}; {\cite{CST:20}}}, optimal stopping {\cite{BCJ:19}}, optimal
transport and robust finance {\cite{EK:19}}, stochastic games and
control problems {\cite{HPBL:18}} as well as high-dimensional nonlinear
partial differential equations (PDEs) {\cite{HJE:17}}; {\cite{HPW:19}}.  Machine learning also allows for new insights into structural properties of financial markets as investigated in {\cite{SC:19}}. For an
exhaustive overview of machine-learning applications in mathematical
finance, in particular for option pricing and hedging we refer to
{\cite{RW:19}}.

The remainder of the article is organized as follows. Section
\ref{sec:VR} introduces the variance reduction technique based on
hedge control variates, which is crucial in our optimization tasks. In
Section \ref{sec:LSVcali} we explain our calibration method, in
particular how to optimize \eqref{eq:calicrit}. The details of the
numerical implementation and the results of the statistical
performance analysis are then given in Section \ref{sec:numerics} as
well as Section~\ref{chap: plots and histograms}. In Appendix \ref{sec:SDE}
we state stability theorems for stochastic differential equations depending on parameters. This is applied to neural SDEs when calculating derivatives with respect to the parameters of the neural
networks. In Appendix \ref{sec:NN} we recall preliminaries on deep learning by giving a
brief overview of universal approximation properties of artificial
neural networks and briefly explaining stochastic gradient
descent. Finally, Appendix \ref{alt} contains alternative optimization
approaches to \eqref{eq:calicrit}.

\section{Variance Reduction for Pricing and Calibration Via Hedging and Deep Hedging}\label{sec:VR}

This section is dedicated to introducing a generic variance reduction
technique for Monte Carlo pricing and calibration by using hedging
portfolios as control variates. This method will be crucial in our LSV
calibration presented in Section \ref{sec:LSVcali}. For similar considerations we refer to \cite{VSSS:18, potters:2001}.

Consider on a finite time horizon $T >0$, a financial market in
discounted terms with $r$ traded instruments $(Z_t)_{t \in [0,T]}$
following an $\mathbb{R}^r $-valued stochastic process on some
filtered probability space
$(\Omega, (\mathcal{F}_t)_{t \in [0,T]} , \mathcal{F},
  \mathbb{Q})$. Here, $ \mathbb{Q}$ is a risk neutral measure and
$(\mathcal{F}_t)_{t \in [0,T]}$ is supposed to be right continuous. In
particular, we suppose that $(Z_t)_{t \in [0,T]}$ is an
$r$-dimensional square integrable martingale with c\`adl\`ag
paths.

Let $C$ be an $\mathcal{F}_T$-measurable random variable describing
the payoff of some European option at maturity $T >0$.  Then the usual
Monte Carlo estimator for the price of this option is given by
\begin{equation}\label{eq: MC estimator}
  \pi= \frac{1}{N} \sum_{n=1}^N C_n,
\end{equation}
where $(C_1,\ldots,C_N)$ are i.i.d with the same distribution as $C$ 
and $N \in\mathbb{N}$.  This estimator can easily be modified by
adding a stochastic integral with respect to $Z$. Indeed, consider a
strategy $(h_t)_{t \in [0,T]} \in L^2(Z)$ and some constant
$c$. Denote the stochastic integral with respect to $Z$ by
$I=(h \bullet Z)_T$ and consider the following estimator   
\begin{align}\label{eq:est_red}
  \widehat{\pi}= \frac{1}{N}\sum_{n=1}^N (C_n - c I_n),
\end{align}
where $(I_1,\ldots,I_N)$ are i.i.d with the same distribution as $I$.
Then, for any $(h_t)_{t \in [0,T]} \in L^2(Z)$ and $c$, this estimator
is still an unbiased estimator for the price of the option with payoff
$C$ since the expected value of the stochastic integral vanishes.  If
we denote by
\[
  H=\frac{1}{N} \sum_{n=1}^N I_n,
\]
then the variance of $\widehat{\pi}$ is given by
\[
  \text{Var}(\widehat{\pi})=\text{Var}(\pi)+ c^2 \text{Var}(H) - 2 c
  \text{Cov}(\pi, H).
\]

This becomes minimal by choosing
\[
  c=\frac{\text{Cov}(\pi, H)}{\Var{H}}.
\]

With this choice, we have
\[
  \text{Var}(\widehat{\pi})=(1-\text{Corr}^2(\pi,H))\text{Var}(\pi).
\]

In particular, in the case of a perfect pathwise hedge, where $\pi=H$ a.s., we have 
$\text{Corr}(\pi, H)=1$ and $\text{Var}(\widehat{\pi})=0$,
since in this case
\[
  \text{Var}(\pi)=\text{Var}(H)=\text{Cov}(\pi, H).
\]

Therefore, it is crucial to find a good approximate hedging portfolio
such that $\text{Corr}^2(\pi,H)$ becomes large. This is subject of
Sections \ref{sec:hedgeBS} and \ref{sec:NNH} below.

\subsection{Black–Scholes Delta Hedge}\label{sec:hedgeBS}
In many cases, of local stochastic volatility models as of form
\eqref{eq:LSV} and options depending only on the terminal value of the
price process, a Delta hedge of the Black–Scholes model works
well. Indeed, let $C=g(S_T)$ and let
$\pi^{g}_{\text{BS}}(t,T, s, \sigma)$ be the price at time $t$ of this
claim in the Black–Scholes model. Here, $s$ stands for the price
variable and $\sigma$ for the volatility parameter in the Black– 
Scholes model. Moreover, we indicate the dependency on the maturity
$T$ as well. Then choosing as hedging instrument only the price $S$
itself and as approximate hedging strategy
\begin{equation}\label{eq: BS hedge formula}
  h_t = \partial_s \pi^g_{\text{BS}}(t,T,S_t,L(t,S_t)\alpha_t)
\end{equation}
usually already yields a considerable variance reduction. In fact, it
is even sufficient to consider $\alpha_t$ alone to achieve satisfying
results, i.e., one has
\begin{equation}\label{eq: BS hedge formula1}
  h_t = \partial_s \pi^g_{\text{BS}}(t,T, S_t,\alpha_t),
\end{equation}

This reduces the computational costs for the evaluation of the hedging
strategies even further.

\subsection{Hedging Strategies as Neural Networks---Deep Hedging}\label{sec:NNH}
Alternatively, in particular when the number of hedging instruments
becomes higher, one can learn the hedging strategy by parameterizing it
via neural networks. For a brief overview of neural networks and
relevant notation used below, we refer to Appendix~\ref{sec:NN}.

Let the payoff be again a function of the terminal values of the
hedging instruments, \mbox{i.e.,~$C=g(Z_T)$}. Then in Markov models it
makes sense to specify the hedging strategy via a function
\[
  h: \mathbb{R}_+ \times \mathbb{R}^r \to \mathbb{R}^r, \; h_t= h(t,z),
\]
which in turn will correspond to an artificial neural network
$(t,z) \mapsto h(t,z,\delta) \in \mathcal{NN}_{r+1, r}$ with weights
denoted by $\delta$ in some parameter space $\Delta$ (see
Notation\footnote{We here use $\delta$ to denote the parameters of the
  hedging neural networks, as $\theta$ shall be used for the networks
  of the leverage~function.} \mbox{\ref{not}}).  Following the approach in
\mbox{\cite[Remark 3]{BGTW:19}}, an optimal hedge for the claim $C$ with
given market price $\pi^{\text{mkt}}$ can be computed via
\[
  \inf_{\delta \in \Delta} \mathbb{E}\left[u\left(-C+
    \pi^{\text{mkt}}+(h(\cdot, Z_{\cdot-}, \delta) \bullet
    Z_{\cdot})_T\right)\right]
\]
for some convex loss function $u:\mathbb{R} \to \mathbb{R}_+ $. Recall that $(h \bullet Z)_T$ denotes the stochastic integral with respect to $Z$. If
$u(x)=x^2$, which is often used in practice, this then corresponds to
a quadratic hedging criterion.

To tackle this optimization problem, we can apply stochastic
gradient descent, because we fall in the realm of problem
\eqref{eq:optimization}. Indeed, the stochastic objective function
$Q(\delta)(\omega)$ is given by
\[
  Q(\delta)(\omega) =u(-C(\omega)+\pi^{\text{mkt}}+ (h(\cdot,
  Z_{\cdot-}, \delta)(\omega) \bullet Z_{\cdot}(\omega))_T ) .
\]
The optimal hedging strategy $h(\cdot,\cdot , \delta^*)$ for an optimizer
$\delta^*$ can then be used to define
\[
  (h(\cdot, Z_{\cdot-},\delta^*)\bullet Z_{\cdot})_T
\]
which is in turn
used in \eqref{eq:est_red}.

As always in this article we shall assume that activation functions
$ \phi $ of the neural network as well as the convex loss function $ u $ are smooth, hence
we can calculate derivatives with respect to $ \delta $ in a straight
forward way. This is important to apply stochastic gradient descent, see
Appendix \ref{sec:SDG}.
We shall show that the gradient of
$ Q(\delta)$ is given by
\[
  \nabla_{\delta} Q(\delta) (\omega) =u'(-C(\omega)+\pi^{\text{mkt}}+ (h(\cdot,
  Z_{\cdot-}, \delta)(\omega) \bullet Z_{\cdot}(\omega))_T ) (\nabla_{\delta} h(\cdot,
  Z_{\cdot-},\delta)(\omega) \bullet Z_{\cdot}(\omega))_T ,
\]
i.e., we are allowed to move the gradient inside the stochastic integral, and that approximations with simple processes, as we shall do in practice, converge to the correct quantities.  To ensure this
property, we shall apply the following theorem, which
follows from results in Section \ref{sec:SDE}.

\begin{Theorem}\label{thm:variation_hedges}
  For $\varepsilon \geq 0$,  let $ Z^\varepsilon $ be a solution of a stochastic differential equation
  as described in Theorem \ref{sta:thm} with drivers $Y=(Y^1,\ldots, Y^d)$,
  functionally
  Lipschitz operators $ F^{\varepsilon,i}_j $, $i=1, \ldots, r$, $j=1, \ldots,d$ and
  a process $(J^{\varepsilon,1}, \ldots J^{\varepsilon,r})$, which is here for all $\varepsilon \geq 0 $ simply $J1_{\{t=0\}}(t)$ for some constant vector $J \in \mathbb{R}^r=J$, i.e.
  \[
    Z^{\varepsilon,i}_t = J^{i} + \sum_{j=1}^d \int_0^t
    F^{\varepsilon,i}_j(Z^{\varepsilon})_{s-} dY^j_s, \quad t \geq 0.
  \]
  
  Let
  $ (\varepsilon,t, z) \mapsto f^\varepsilon(t,z) $ be a map, such that the
  bounded  c\`agl\`ad process $ f^\varepsilon := f^\varepsilon(.-,Z^0_{.-}) $
  converges ucp to $f^0:=f^0(.-,Z^0_{.-})$, then
  \[
    \lim_{\varepsilon \to 0} (f^\varepsilon \bullet Z^\varepsilon) = \big (f^0
    \bullet Z^0 \big)
  \]
  holds true.
\end{Theorem}
\begin{proof}
  Consider the extended system
  \[
    d (f^\varepsilon \bullet Z^\varepsilon) = \sum_{j=1}^d
    f^\varepsilon(t-,Z_{t-}^\varepsilon) \,
    F^{\varepsilon,i}_j(Z^\varepsilon)_{t-} dY^j_t
  \]
  and
  \[
    d Z^{\varepsilon,i}_t = \sum_{j=1}^d
    F^{\varepsilon,i}_j(Z^\varepsilon)_{t-} dY^j_t, \,
  \]
  where we obtain existence, uniqueness and stability for the second
  equation by Theorem \ref{sta:thm}, and~from where we obtain ucp
  convergence of the integrand of the first equation: since stochastic
  integration is continuous with respect to the ucp topology we obtain
  the result.
\end{proof}

The following corollary implies the announced properties, namely that we can move the gradient inside the stochastic integral and that the derivatives of a
discretized integral with a discretized version of $Z$ and approximations of the hedging strategies are actually close to the derivatives of the
limit~object.

\begin{Corollary}
  Let, for $\varepsilon >0$,  $Z^{\varepsilon}$ denote a discretization of the process of hedging instruments $Z\equiv Z^0$ such that the conditions of Theorem \ref{thm:variation_hedges} are satisfied.  Denote, for $\varepsilon \geq 0$, the corresponding hedging strategies by $(t,z,\delta) \mapsto h^{\varepsilon}(t,z ,\delta)$ given by neural networks $\mathcal{NN}_{r+1,r}$, whose activation functions are bounded and $C^1$, with~bounded~derivatives.

  \begin{enumerate}
    \item Then the derivative $\nabla_{\delta} (h (\cdot, Z_{\cdot-},\delta) \bullet Z)$ in direction $\delta$ at $\delta_0$ satisfies
          \[
            \nabla_{\delta} (h(\cdot,Z_{\cdot-},\delta_0)  \bullet Z)= (\nabla_{\delta} h(\cdot,Z_{\cdot-},\delta_0)  \bullet Z).
          \]
    \item If additionally the derivative in direction $\delta$ at $\delta_0$ of  $  \nabla_{\delta} h^{\varepsilon}(\cdot,Z_{\cdot-},\delta_0) $ converges  ucp to $ \nabla_{\delta} h(\cdot,Z_{\cdot-},\delta_0)$ as $\varepsilon \to 0$, then the directional derivative of the discretized integral, i.e.\\ $ \nabla_{\delta}(h^{\varepsilon} (\cdot,Z^\varepsilon_{\cdot-},\delta_0)\bullet Z^{\varepsilon})$ or equivalently $ ( \nabla_{\delta} h^{\varepsilon}(\cdot,Z^\varepsilon_{\cdot-},\delta_0) \bullet Z^{\varepsilon})$,
          converges, as the discretization mesh $\varepsilon \to 0$, to
          \[
            \lim_{\varepsilon \to 0}( \nabla_{\delta} h^{\varepsilon}(\cdot,Z^{\varepsilon}_{\cdot-},\delta_0) \bullet Z^{\varepsilon})= (\nabla_{\delta} h(\cdot,Z_{\cdot-},\delta_0)  \bullet Z).
          \]
  \end{enumerate}
\end{Corollary}

\begin{proof}
  To prove (1), we apply Theorem \ref{thm:variation_hedges} with
  \[
    f^{\varepsilon}(\cdot,Z_{\cdot-}) =\frac{h(\cdot,Z_{\cdot-},\delta_0 +\varepsilon \delta)- h(\cdot,Z_{\cdot-},\delta_0) }{\varepsilon},
  \]
  which converges ucp to $f^0=\nabla_{\delta} h(\cdot,Z_{\cdot-},\delta_0)$.
  Indeed, by the neural network assumptions, we have (with the $\sup$  over some compact set)
  \[
    \lim_{\varepsilon \to 0}\sup_{(t,z)}\left\|\frac{h(t,z,\delta_0 +\varepsilon \delta)- h(t,z,\delta_0) }{\varepsilon} - \nabla_{\delta} h(t,z,\delta_0)\right\|=0,
  \]
  by equicontinuity of
  $\{(t,z) \mapsto \nabla_{\delta} h(t,z ,\delta_0+ \varepsilon \delta) \,|\, \varepsilon \in [0,1] \}$.

  Concerning (2) we apply again Theorem \ref{thm:variation_hedges}, this time with
  \[
    f^{\varepsilon}(\cdot,Z_{\cdot-}) = \nabla_{\delta} h^{\varepsilon}(\cdot,Z_{\cdot-},\delta_0),
  \]
  which converges by assumption ucp to $f^0= \nabla_{\delta} h(\cdot,Z_{\cdot-},\delta_0)$.

\end{proof}

\section{Calibration of LSV Models}\label{sec:LSVcali}

Consider an LSV model as of \eqref{eq:LSV} defined on some filtered
probability space
$(\Omega, (\mathcal{F}_t)_{t \in [0,T]} , \mathcal{F}, \mathbb{Q})$,
where $ \mathbb{Q}$ is a risk neutral measure. We assume the
stochastic process $\alpha$ to be fixed. This can for instance be
achieved by first approximately calibrating the pure stochastic volatility model
with $L \equiv 1$, so to capture only the right order of magnitude of the parameters and then fixing them.

Our main goal is to determine the leverage function $L$ in perfect
accordance with market data. We~here consider only European call
options, but our approach allows in principle to take all kind of
other options into account.

Due to the universal approximation properties outlined in Appendix
\ref{sec:NN} (Theorem \ref{th:universalapprox}) and in spirit of neural SDEs, we choose to
parameterize $L$ via neural networks. More precisely, set $T_0=0$ and
let $0 < T_1 \cdots < T_n=T$ denote the maturities of the available
European call options to which we aim to calibrate the LSV model.  We
then specify the leverage function $L(t,s)$ via a family of neural
networks,~i.e.,
\begin{align}\label{eq:leverageNN}
  L(t,s, \theta)=\left( 1 + \sum_{i=1}^n F^i(s, \theta_i) 1_{[T_{i-1}, T_{i})}(t) \right),
\end{align}
where  $F^i \in \mathcal{N N}_{1,1}$ for $i= 1,\ldots,n$ (see Notation~\ref{not}).
For notational simplicity we shall often omit the dependence
on  $\theta_i \in \Theta_i$. However, when needed we write for instance $S_t(\theta)$, where $\theta$ then stands for all parameters $\theta_i$ used up to time $t$.

For purposes of training, similarly as in Section \ref{sec:NNH}, we shall need
to calculate derivatives of the LSV process with respect to
$ \theta $. The following result can be understood as the chain rule applied to $\nabla_{\theta}S(\theta)$, which we prove here rigorously by applying the results of Appendix \ref{sec:SDE}.

\begin{Theorem}\label{thm:variation_leverage}
  Let $ (t,s, \theta) \mapsto L(t, s, \theta) $ be of form \eqref{eq:leverageNN} where the neural networks $(s, \theta_i) \mapsto F^i(s, \theta_i)$ are bounded and $C^1$, with bounded and Lipschitz continuous derivatives\footnote{This just means that the activation function is bounded and $C^1$, with bounded and Lipschitz continuous derivatives.},  for all  $i=1, \ldots, n$.
  Then the directional derivative in direction $\theta$ at $\widehat{\theta}$ satisfies the following
  equation
  \begin{equation}
    \begin{split} \label{eq:derivative}
      d \left ( \nabla_{\theta} S_t(\widehat{\theta}) \right) &= \Big (
      \nabla_{\theta} S_t(\widehat{\theta})
      L(t,S_t(\widehat{\theta}), \widehat{\theta})
      + S_t(\widehat{\theta})  \partial_s
      L(t,S_t(\widehat{\theta}), \widehat{\theta})  \nabla_{\theta} S_t(\widehat{\theta}) \\
      &\quad + S_t(\widehat{\theta})  \nabla_{\theta} L(t,S_t(\widehat{\theta}) ,\widehat{\theta})\Big
      ) \alpha_t dW_t \, ,
    \end{split}
  \end{equation}
  with initial value $0$. This can be solved by variation of constants, i.e.
  \begin{align}\label{eq:grad}
    \nabla_{\theta} S_t(\widehat{\theta})=\int_0^t P_{t-s} S_s(\widehat{\theta}) \nabla_{\theta} L(s, S_s(\widehat{\theta}), \widehat{\theta})\alpha_s dW_s,
  \end{align}
  where
  \[
    P_t=\mathcal{E}\left(\int_0^t \left(L(s, S_s(\widehat{\theta}), \widehat{\theta})+S_s(\widehat{\theta})\nabla_s L(s, S_s(\widehat{\theta}), \widehat{\theta})\right) \alpha_s dW_s \right)
  \]
  with $\mathcal{E}$ denoting the stochastic exponential.

\end{Theorem}

\begin{proof}
  First note that Theorem \ref{e_u:thm} implies the existence and uniqueness
  of
  \[
    d S_t(\theta) = S_t(\theta)L(t,S_t(\theta),\theta) \alpha_t
    dW_t \, ,
  \]
  for every $\theta$.  Here, the driving process is one-dimensional and given by $Y=\int_0^{\cdot}\alpha_s dW_s$.
  Indeed, according to Remark \ref{rem:appl},
  if $(t,s) \mapsto L(t,s , \theta)$ is bounded, c\`adl\`ag in $t$  and Lipschitz in $s$ with a Lipschitz constant independent of $t$, $S_{\cdot}\mapsto S_{\cdot}(\theta)L(\cdot, S_{\cdot}(\theta),\theta)$ is functionally Lipschitz and Theorem \ref{e_u:thm} implies the assertion. These conditions are implied by the form of $L(t,s,\theta) $ and the conditions on the neural networks $F^i$.

  To prove the form of the derivative process we apply Theorem \ref{sta:thm} to the following system:~consider
  \[
    d S_t(\widehat{\theta}) = S_t(\widehat{\theta})L(t,S_t(\widehat{\theta}),\widehat{\theta}) \alpha_t
    dW_t \, ,
  \]
  together with
  \[
    d S_t(\widehat{\theta}+\varepsilon \theta) =
    S_t(\widehat{\theta}+\varepsilon \theta)L(t,S_t(\widehat{\theta}+\varepsilon \theta),\widehat{\theta}+\varepsilon \theta) \alpha_t
    dW_t \, ,
  \]
  as well as
  \begin{align*}
    d \,\frac{S_t(\widehat{\theta}+\varepsilon \theta)- S_t(\widehat{\theta})}{\varepsilon} & =  \frac{S_t(\widehat{\theta}+\varepsilon \theta)L(t,S_t(\widehat{\theta}+\varepsilon \theta),\widehat{\theta}+\varepsilon \theta)-
      S_t(\widehat{\theta})L(t,S_t(\widehat{\theta}),\widehat{\theta}) }  {\varepsilon}
    \alpha_t dW_t                                                                                                                                                                                                                                                                               \\
                                                                                            & =\Big(\frac{S_t(\widehat{\theta}+\varepsilon \theta)- S_t(\widehat{\theta})}{\varepsilon} L(t,S_t(\widehat{\theta}+\varepsilon \theta),\widehat{\theta}+\varepsilon \theta)                       \\
                                                                                            & \quad+ S_t(\widehat{\theta})\frac{L(t,S_t(\widehat{\theta}+\varepsilon \theta),\widehat{\theta}+\varepsilon \theta)-L(t,S_t(\widehat{\theta}),\widehat{\theta})}{\varepsilon}\Big)\alpha_t dW_t .
  \end{align*}
  
  In the terminology of Theorem \ref{sta:thm}, $Z^{\varepsilon, 1}=S(\widehat{\theta})$, $Z^{\varepsilon, 2}=S(\widehat{\theta}+ \varepsilon \theta)$ and  $Z^{\varepsilon, 3}=\frac{S_t(\widehat{\theta}+\varepsilon \theta)- S_t(\widehat{\theta})}{\varepsilon}$. Moreover, $F^{\varepsilon, 3}$ is given by
  \begin{equation}\label{eq:funcLip}
    \begin{split}
      F^{\varepsilon, 3}(Z^0_t)&= Z_t^{0,3} L(t,Z_t^{0,2},\widehat{\theta}+\varepsilon \theta)+Z^{0,1}_t\partial_s L(t, Z^{0,1}_t ,\widehat{\theta})Z_t^{0,3}+\mathcal{O}(\varepsilon)\\
      &\quad +Z^{0,1}_t\frac{ L(t,Z_t^{0,1} ,\widehat{\theta}+\varepsilon \theta)- L(t, Z^{0,1}_t,\widehat{\theta})}{\varepsilon},
    \end{split}
  \end{equation}
  which converges ucp to
  \begin{align*}
    F^{0,3}(Z^0_t)= Z_t^{0,3} L(t,Z_t^{0,2},\widehat{\theta})+Z^{0,1}_t\partial_s L(t, Z^{0,1}_t ,\widehat{\theta})Z_t^{0,3}+Z^{0,1}_t\nabla_{\theta} L(t, Z^{0,1}_t ,\widehat{\theta}).
  \end{align*}
  
  Indeed, for every fixed $t $, the family  $\{s \mapsto L(t,s, \widehat{\theta} +\varepsilon \theta), \,|\, \varepsilon \in [0,1] \}$ is due to the form of the neural networks equicontinuous. Hence pointwise convergence implies uniform convergence in $s$. This~together with $ L(t,s, \theta)$ being piecewise constant in $t$ yields
  \[
    \lim_{\varepsilon \to 0} \sup_{(t,s)} | L(t,s , \widehat{\theta}+\varepsilon \theta) - L(t,s , \widehat{\theta})|=0,
  \]
  whence ucp convergence of the first term in \eqref{eq:funcLip}.
  The convergence of term two is clear. The  one of term three follows again from the fact that the family  $\{s \mapsto \nabla_{\theta} L(t,s ,  \widehat{\theta}+ \varepsilon \theta) \,|\, \varepsilon \in [0,1] \}$ is  equicontinuous, which is again a consequence of the form of the neural networks.

  By the assumptions on the derivatives, $F^{0, 3}$  is functionally Lipschitz.
  Hence Theorem \ref{e_u:thm} yields the existence of a unique solution to \eqref{eq:derivative} and Theorem \ref{sta:thm} implies
  convergence.
\end{proof}

\begin{Remark}\text{~~~}
  \begin{enumerate}
    \item
          For the pure existence and uniqueness of
          \[
            d S_t(\theta) = S_t(\theta)L(t,S_t(\theta),\theta) \alpha_t
            dW_t \, ,
          \]
          with $L(t, s, \theta) $  of form \eqref{eq:leverageNN}, it suffices that
          the neural networks $s \mapsto F^i(s , \theta_i)$ are bounded and Lipschitz, for all  $i=1, \ldots, n$ (see also Remark \ref{rem:appl}).
    \item Formula \eqref{eq:grad} can be used for well-known backward propagation schemes.
  \end{enumerate}
\end{Remark}

Theorem \ref{thm:variation_leverage}
guarantees the existence and uniqueness of the derivative process. This thus allows the setting up of gradient-based search
algorithms for training.

In view of this let us now come to the precise optimization task as already outlined in Section \ref{sec:LSVintro}. To ease the notation, we shall here omit the dependence of the weights $w$ and
the loss function $\ell$ on the parameter $ \gamma $.  For each maturity $T_i$, we assume to have $J_i$
options with strikes $K_{ij}$, $j \in \{1, \ldots, J_i\}$. The~calibration functional for the $i$-th maturity is then of the form
\begin{align}\label{eq:minimization}
  \argmin_{\theta_i \in \Theta_i}  \sum_{j=1}^{J_i} w_{ij} \ell(\pi_{ij}^{\text{mod}}(\theta_i) -\pi_{ij}^{\text{mkt}}), \quad i \in \{1, \ldots, n\}.
\end{align}
Recall from the introduction that $\pi^{\text{mod}}_{ij}(\theta_i)$ ($\pi^{\text{mkt}}_{ij}$
respectively) denotes the model (market resp.) price of an option with
maturity $T_i$ and strike $K_{ij}$. Moreover, 
$\ell: \mathbb{R} \to \mathbb{R}_+$ is some non-negative, nonlinear,
convex loss function (e.g.~square or absolute value) with $\ell(0)=0$
and $\ell(x) >0$ for $x \neq 0$, measuring the distance between market
and model prices. Finally, $w_{ij}$ denote some weights, e.g.~of vega
type (compare \cite{cont2004recovering}), which we use to match
implied volatility data rather than pure prices. Notice that we here omit for notational convenience the dependence of $w_{ij}$ and $\ell$ on parameters $\gamma$ which describe the adversarial part.

We solve the minimization problems~\eqref{eq:minimization}
iteratively: we start with maturity $T_1$ and fix $\theta_1$. This~then enters in the computation of $\pi^{\text{mod}}_{2j}(\theta_2)$
and thus in \eqref{eq:minimization} for maturity $T_2$, etc. To
simplify the notation in the sequel, we shall therefore leave the
index $i$ away so that for a generic maturity $T >0$,
\eqref{eq:minimization} becomes
\begin{align*}
  \argmin_{\theta \in \Theta}  \sum_{j=1}^{J} w_{j} \ell(\pi_{j}^{\text{mod}}(\theta) -\pi_{j}^{\text{mkt}}).
\end{align*}

Since the model prices are given by
\begin{align}\label{eq:modprice}
  \pi_j^{\text{mod}}(\theta)=\mathbb{E}[ (S_{T}(\theta)-K_{j})^+],
\end{align}
we have
$\pi_{j}^{\text{mod}}(\theta)
  -\pi_{j}^{\text{mkt}}=\EW{Q_{j}(\theta)}$ where
\begin{align}\label{eq:Q}
  Q_{j}(\theta)(\omega):= (S_{T}(\theta)(\omega)-K_{j})^+-\pi^{\text{mkt}}_{j}.
\end{align}

The
calibration task then amounts to finding a minimum of
\begin{align}\label{eq:calibration}
  f(\theta):= \sum_{j=1}^{J} w_{j} \ell(\EW{Q_{j}(\theta)}).
\end{align}

As $\ell$ is a nonlinear function, this is not of the expected value
form of problem \eqref{eq:optimization}. Hence standard stochastic
gradient descent, as outlined in Appendix \ref{sec:SDG}, cannot be
applied in a straightforward manner.

We shall tackle this problem via hedge control variates as introduced
in Section \ref{sec:VR}. In the following we explain this in more
detail.

\subsection{Minimizing the Calibration Functional}

Consider the standard Monte Carlo estimator for
$\mathbb{E}[Q_{j}(\theta)]$ so that \eqref{eq:calibration} is
estimated by
\begin{align}\label{eq:Qest}
  f^{\text{MC}}(\theta):= \sum_{j=1}^J w_j \ell\left(\frac{1}{N}\sum_{n=1}^N Q_{j}(\theta)(\omega_n)\right),
\end{align}
for i.i.d samples $\{\omega_1, \ldots, \omega_N\} \in \Omega$.  Since
the Monte Carlo error decreases as $\frac{1}{\sqrt{N}}$, the number of
simulations $N$ must be chosen large ($\approx 10^8$) to
approximate well the true model prices in \eqref{eq:modprice}. Note
that implied volatility to which we actually aim to calibrate is even
more sensitive. As~stochastic gradient descent is not directly
applicable due to the nonlinearity of $\ell$, it seems necessary at
first sight to compute the gradient of the whole function
$\widehat{f}(\theta)$ to minimize \eqref{eq:Qest}.  As~$N\approx 10^8$, this is however computationally very expensive and leads
to numerical instabilities as we must compute the gradient of a sum that contains $10^8$ terms. Hence with this method an (approximative) minimum in the
high-dimensional parameter space $\Theta$ cannot be found in a reasonable amount of
time.

One very expedient remedy is to apply hedge control variates as
introduced in Section \ref{sec:VR} as variance reduction
technique. This allows the reduction of the number of samples $N$ in the
Monte Carlo estimator considerably to only up to $ 5 \times 10^4$ 
sample paths.

Assume that we have $r$ hedging instruments (including the price
process $S$) denoted by $(Z_t)_{t \in [0,T]}$ which are square
integrable martingales under $\mathbb{Q}$ and take values in
$\mathbb{R}^r$.  Consider, for $j=1, \ldots, J$, strategies
$h_j: [0,T]\times \mathbb{R}^r \to \mathbb{R}^r$ such that
$h(\cdot, Z_{\cdot}) \in L^2(Z)$ and some constant $c$. Define
\begin{align}\label{eq:X}
  X_{j}(\theta)(\omega):= Q_j(\theta)(\omega)  - c (h_j(\cdot, Z_{\cdot-}(\theta)(\omega)) \bullet Z_{\cdot}(\theta)(\omega))_T
\end{align}

The calibration functionals
\eqref{eq:calibration} and \eqref{eq:Qest}, can then simply be defined
by replacing $ Q_j(\theta)(\omega)$ by $X_j(\theta)(\omega)$ so that
we end up minimizing
\begin{align}\label{eq:calfin}
  \widehat{f}(\theta)(\omega_1, \ldots, \omega_N)=\sum_{j=1}^J w_j \ell\left(\frac{1}{N}\sum_{n=1}^N X_{j}(\theta)(\omega_n)\right).
\end{align}

To tackle this task, we apply the following variant of gradient
descent: starting with an initial guess $\theta^{(0)}$, we iteratively
compute
\begin{align}\label{eq:updateimpl}
  \theta^{(k+1)}=\theta^{(k)}- \eta_k \; G(\theta^{(k)})(\omega^{(k)}_1, \ldots, \omega^{(k)}_N),
\end{align}
for some learning rate $\eta_k$, i.i.d samples
$(\omega^{(k)}_1, \ldots, \omega^{(k)}_N) $, where the values
\[
  G(\theta^{(k)})(\omega^{(k)}_1, \ldots, \omega^{(k)}_N)
\]
are gradient-based quantities that remain to be specified. These samples can either
be chosen to be the same in each iteration or to be newly sampled in
each update step. The difference between these two approaches is negligible, since $N$ is chosen so as to yield a small Monte Carlo
error, whence the gradient is nearly deterministic. In our numerical
experiments we newly sample in each update step.

In the simplest form, one could simply set
\begin{equation}\label{eq: standard grad}
  G(\theta^{(k)})(\omega^{(k)}_1, \ldots, \omega^{(k)}_N) = \nabla \widehat{f}(\theta)(\omega_1^{(k)},\ldots, \omega_N^{(k)}).
\end{equation}

Note however that the derivative of the stochastic integral term in \eqref{eq:X} is in general quite expensive. We thus implement the following modification.

We set
\begin{align*}
  \omega^N                & = (\omega_1, \ldots, \omega_N),                                 \\
  Q^N_j(\theta)(\omega^N) & = \frac{1}{N} \sum_{n=1}^N Q_j(\theta)(\omega_n),               \\
  Q^N (\theta)(\omega^N)  & = (Q_1^N(\theta)(\omega^N),\ldots, Q_J^N(\theta)(\omega^N)   ),
\end{align*}
and define $\tilde f : \mathbb{R}^J\rightarrow \mathbb{R}$ via
\begin{equation*}
  \tilde f (x) = \sum_{j=1}^J w_j \ell(x_j).
\end{equation*}

We then set
\begin{equation*}
  G(\theta)(\omega^N) = D_x(\tilde f)(X^N(\theta)(\omega^N) )  D_\theta(Q^N)(\theta)(\omega^N).
\end{equation*}

Please note that this quantity is actually easy to compute in terms of backpropagation. Moreover, leaving~the stochastic integral away in the inner derivative is justified by its vanishing expectation. During the forward pass, the stochastic integral terms are included in the computation; however the contribution to the gradient (during the backward pass) is partly neglected, which can e.g.~ be implemented via the tensorflow \texttt{stop\_gradient} function.

Concerning the choice of the hedging strategies, we can parameterize
them as in Section \ref{sec:NNH} via neural networks and find the
optimal weights $\delta$ by computing
\begin{align}\label{eq:deltaopt}
  \argmin_{\delta \in \Delta} \frac{1}{N}\sum_{n=1}^N  u(-X_{j}(\theta, \delta)(\omega_n)).
\end{align}
for i.i.d samples $\{\omega_1, \ldots, \omega_N\} \in \Omega$ and some
loss function $u$ when $\theta$ is fixed.  Here,
\[
  X_{j}(\theta, \delta)(\omega)=(S_T(\theta)(\omega)-K_j)^+
  -(h_j(\cdot, Z_{\cdot-}(\theta)(\omega),\delta) \bullet
  Z_{\cdot}(\theta)(\omega))_T - \pi_j^{\text{mkt}}.
\]

This means to iterate the two optimization procedures,
i.e.,~minimizing \eqref{eq:calfin} for $\theta$ (with fixed $\delta)$
and~\eqref{eq:deltaopt} for $\delta$ (with fixed $\theta$).  Clearly
the Black–Scholes hedge ansatz as of Section \ref{sec:hedgeBS} works
as well, in this case without additional optimization with respect to
the hedging strategies.

For alternative approaches how to minimize~\eqref{eq:calibration}, we
refer to Appendix~\ref{alt}.

\section{Numerical Implementation} \label{sec:numerics}

In this section, we discuss the numerical implementation of the proposed calibration method.
We implement our approach via tensorflow, taking advantage of
GPU-accelerated computing. All~computations are performed on a
single-gpu Nvidia GeForce ® GTX 1080 Ti machine. For~the implied volatility 
computations, we rely on the python py\_vollib library.\footnote{See \url{http://vollib.org/}.}

Recall that a LSV model is given on some filtered probability space
$(\Omega, (\mathcal{F}_t)_{t \in [0,T]} , \mathcal{F}, \mathbb{Q})$ by
\[ dS_t = S_t \alpha_t L(t,S_t) dW_t,\quad S_0>0, \] for some stochastic
process $\alpha$. When calibrating to data, it is, therefore, necessary
to make further specifications. We calibrate the following SABR-type
LSV model.

\begin{Definition}\label{eq: SABR LSV}
  The SABR-LSV model is specified via the SDE,
  \begin{equation*}
    \begin{split}
      dS_t &= S_t  L(t,S_t) \alpha_t dW_t,\\
      d\alpha_t &= \nu \alpha_t dB_t, \\
      d\langle W, B \rangle_t &= \rho dt,
    \end{split}
  \end{equation*}
  with parameters $\nu \in \mathbb{R}$, $\rho \in [-1,1]$ and initial
  values \(\alpha_0 >0,\, S_0 >0\). Here, $B$ and $W$ are two correlated
  Brownian motions.
\end{Definition}

\begin{Remark}

  We shall often work in log-price coordinates for $S$. In particular,
  we can then consider $L$ as a function of $X:=\log S$ rather then
  $S$. By denoting this parametrization again with $L$, we therefore
  have $L(t,X)$ instead of $L(t,S)$ and the model dynamics read as
  \begin{equation*}
    \begin{split}
      dX_t &= \alpha_t  L(t,X_t ) dW_t - \frac{1}{2} \alpha_t^2 L^2(t,X_t) dt,\\
      d\alpha_t &= \nu \alpha_t dB_t, \\
      d\langle W, B \rangle_t &= \rho dt.
    \end{split}
  \end{equation*}
  
  Please note that $\alpha$ is a geometric Brownian motion, in particular, 
  the closed form solution for $\alpha$ is available and given by
  \begin{equation*}
    \alpha_t = \alpha_0 \exp \left( - \frac{\nu^2}{2} t + \nu B_t \right).
  \end{equation*}
\end{Remark}

For the rest of the paper we shall set $S_0=1$.

\subsection{Implementation of the Calibration Method}

We now present a proper numerical test and demonstrate the
effectiveness of our approach on a family of typical market smiles
(instead of just one calibration example). We consider as \emph{ground truth} a situation where market smiles are produced by a parametric
family. By randomly sampling smiles from this family we then show that
they can be calibrated up to small errors, which we analyze~statistically.

\subsubsection{Ground Truth Assumption}
We start by specifying the ground truth assumption.  It is known that
a discrete set of prices can be exactly calibrated by a local
volatility model using Dupire's volatility function, if an appropriate
interpolation method is chosen. Hence, any market observed smile data
can be reproduced by the following model (we assume zero riskless rate
and define $X=\log(S)$),
\[
  dS_t = \sigma_{\text{Dup}}(t, X_t ) S_t dW_t,
\]
or equivalently
\begin{equation}\label{eq:LocVolSDE}
  dX_t = -\frac{1}{2} \sigma_{\text{Dup}}^2(t,X_t) dt +   \sigma_{\text{Dup}}(t, X_t ) dW_t,
\end{equation}
where $\sigma_{\text{Dup}}$ denotes Dupire's local volatility
function~\cite{D:96}. Our ground truth assumption consists of supposing
that the function $\sigma_{\text{Dup}}$ (or to be more precise
$\sigma_{\text{Dup}}^2$) can be chosen from a parametric family. Such
parametric families for local volatility models have been discussed in
the literature, consider e.g.~ \cite{carmona2009local} or
\cite{carmonainbook}. In the latter, the authors introduce a family of
local volatility functions $\widetilde{a}_{\xi}$ indexed by parameters
\[
  \xi=(p_1,p_2,\sigma_0,\sigma_1,\sigma_2)
\]
and $p_0 = 1 - (p_1+p_2)$ satisfying the constraints
\[
  \sigma_0,\sigma_1,\sigma_2, p_1,p_2 >0 \text{ and } p_1+p_2 \le 1 .
\]

Setting
$k(t,x,\sigma)= \exp\left(-x^2/(2 t \sigma^2) - t \sigma^2/8 \right)$,
$\widetilde{a}_{\xi}$ is then defined as
\[
  \widetilde{a}^2_\xi(t,x) = \frac{\sum_{i=0}^2 p_i \sigma_i
    k(t,x,\sigma_i) }{\sum_{i=0}^2 (p_i / \sigma_i) k(t,x,\sigma_i) }.
\]

In Figure~\ref{fig:iV parametric family}a we show plots of implied
volatilities for different slices (maturities) for a realistic choice of
parameters. As one can see, the produced smiles seem to be
unrealistically flat. Hence we modify the local volatility function
$\widetilde{a}_{\xi}$ to produce more pronounced and more realistic
smiles.  To be precise, we~define a new family of local volatility
functions $a_{\xi}$ indexed by the set of parameters $\xi$ as
\begin{equation}\label{eq: defi of L_theta}
  a_{\xi}^2(t,x) = \frac{1}{4} \times \min \Big( 2, \left|  \frac{ \left( \sum_{i=0}^2 p_i \sigma_i k(t,x,\sigma_i) + \Lambda (t,x)  \right) \left( 1-0.6 \times \ind{t>0.1} \right) }
    {\sum_{i=0}^2 (p_i / \sigma_i) k(t,x,\sigma_i) +0.01 } \right|   \Big),
\end{equation}
with
\[
  \Lambda(t,x) := \left(\frac{\ind{t\le 0.1} }{1+0.1 t
  }\right)^{\lambda_2}\min\left\{ \left(\gamma_1
  \left(x-\beta_1\right)_+ + \gamma_2 \left(-x-\beta_2\right)_+
  \right)^\kappa ,\, \lambda_1 \right\} .
\]

We fix the choice of the parameters
$\gamma_i,\beta_i, \lambda_i, \kappa $ as given in Table~\ref{table:paraLocvol}. By taking absolute values above, we can drop the requirement $p_0>0$ which is what we do in the sequel. Please note that $a^2_\xi$ is not defined at
$t=0$. When doing a Monte Carlo simulation, we simply replace
$a^2_{\xi}(0,x)$ with $a^2_\xi(\Delta_t,x)$, where $\Delta_t$ is the
time increment of the Monte Carlo simulation.

What is left to be specified are the parameters
\[\xi=(p_1,p_2,\sigma_0,\sigma_1,\sigma_2)\]
with $p_0=1-p_1-p_2$.  This motivates our statistical test for the
performance evaluation of our method. To be precise, our ground truth
assumption is that all observable market prices are explained by a
variation of the parameters $\xi$. For illustration, we plot implied
volatilities for this modified local volatility function in
Figure~\ref{fig:iV parametric family}b for a specific parameter set
$\xi$.

Our ground truth model is now specified as in \eqref{eq:LocVolSDE}
with $\sigma_{\text{Dup}}$ replaced by $a_{\xi}$, i.e.,~
\begin{equation}\label{eq:LocVolSDE1}
  dX_t = -\frac{1}{2} a_{\xi}^2(t,X_t) dt +  a_{\xi}(t, X_t ) dW_t.
\end{equation}

\begin{table}[H]
  \centering
   \caption{Fixed Parameters for the ground truth assumption
    $a^2_\xi$.}
  \label{table:paraLocvol}
  \begin{tabular}{c c c c c c c }
    \toprule[0.6pt]
    \boldmath{$\gamma_{1}$} & \boldmath{$\gamma_{2}$} & \boldmath{$\lambda_{1}$} & \boldmath{$\lambda_{2}$} & \boldmath{$\beta_{1}$} & \boldmath{$\beta_{2}$} & \boldmath{$\kappa$} \\
    \midrule[0.5pt]
    $1.1$        & $20$         & $10$          & $10$          & $0.005$     & $0.001$     & $0.5$    \\
    \bottomrule[0.6pt]
  \end{tabular}

\end{table}\unskip\vspace{-12pt}

\begin{figure}[H]\centering
  \begin{subfigure}{0.49\textwidth}\centering
    \includegraphics[trim=.0cm .0cm 0cm .0cm, clip,
      width=0.98\textwidth]{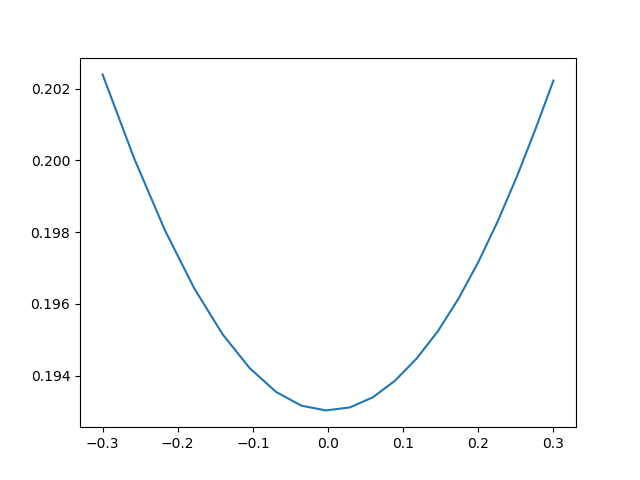}
    \caption{}
    \label{fig:iV parametric family (a)}
  \end{subfigure}
  \begin{subfigure}{0.49\textwidth}\centering
    \includegraphics[trim=.0cm .0cm 0cm .0cm, clip,
      width=0.98\textwidth]{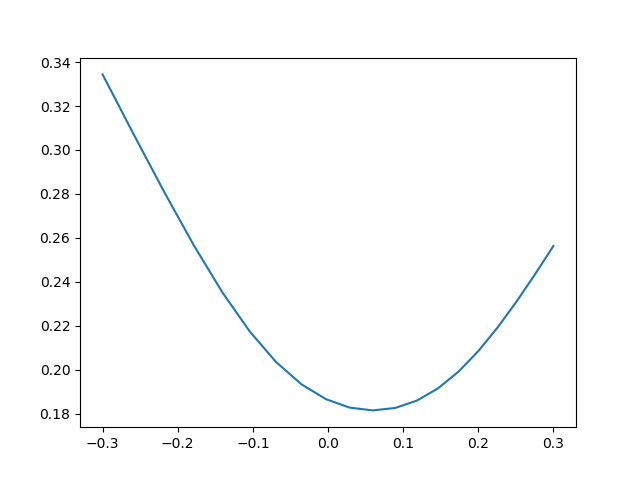}
    \caption{} 
    \label{fig:iV parametric family (b)}
  \end{subfigure}
  \caption{Implied volatility of the original parametric family
    $\widetilde{a}_{\xi}$ (\textbf{a}) versus our modification $a_{\xi}$ (\textbf{b}) for maturity $T=0.5$, the $x$-axis is given on log-moneyness $\ln(K/S_0)$.
  }
  \label{fig:iV parametric family}
\end{figure}

\subsubsection{Performance Test}\label{sec:performance test} We now come to the evaluation of our
proposed method.  We want to calibrate the SABR-LSV model to synthetic market
prices generated by the previously formulated ground truth
assumption. This~corresponds to randomly sampling the parameter $\xi$
of the local volatility function $a_{\xi}$ and to compute prices
according to \eqref{eq:LocVolSDE1}. Calibrating the SABR-LSV model,
i.e.,~finding the parameters \(\nu,\rho\), the initial volatility
$\alpha_0$ and the unknown leverage function $L$, to these prices and
repeating this multiple times then allows for a statistical analysis
of the errors.

As explained in Section \ref{sec:LSVcali}, we consider European call
options with maturities $T_1 < \cdots <T_{n}$ and denote the strikes
for a given maturity $T_i$ by $K_{ij}$, $j \in \{1, \ldots, J_i\}$. To
compute the ground truth prices for these European calls we use a
Euler-discretization of \eqref{eq:LocVolSDE1} with time step
$\Delta_t = 1/100$. Prices are then obtained by a variance reduced
Monte Carlo estimator using $10^7$ Brownian paths and a Black–Scholes
delta hedge variance reduction as described previously. For a given 
parameter set $\xi$, we use the same Brownian paths for all strikes
and maturities.

Overall, in this test, we consider $n=4$ maturities with $J_i=20$
strike prices for all $i=1,\ldots,4$. The~values for $T_i$ are given
in Figure~\ref{fig:matsAndStrikes}a. For the choice of the strikes
$K_i$, we choose evenly spaced points, i.e.,
\[
  K_{i,j+1}-K_{i,j} = \frac{K_{i,20}-K_{i,1}}{19}.
\]

For the smallest and largest strikes per maturity we choose
\[
  K_{i,1}= \exp\left(-k_i\right), \; K_{i,20}= \exp\left(k_i\right),
\]
with the values of $k_i$ given in Figure~\ref{fig:matsAndStrikes}b.

\begin{figure}[H]
  \centering
  \begin{subfigure}[b]{0.3\textwidth}\centering
    \begin{tabular}{c c c c }
      \toprule[1pt]
      $T_1$ & $T_2$ & $T_3$ & $T_4$ \\
      \midrule[0.5pt]
      0.15  & 0.25  & 0.5   & 1.0   \\
      \bottomrule[1pt]
    \end{tabular}
    \subcaption{}\label{fig:matsAndStrikes-a}
  \end{subfigure}
  \begin{subfigure}[b]{0.3\textwidth} \centering
    \begin{tabular}{c c c c }
      \toprule[1pt]
      $k_{1}$ & $k_{2}$ & $k_{3}$ & $k_{4}$ \\
      \midrule[0.5pt]
      $0.1$   & $0.2$   & $0.3$   & $0.5$   \\
      \bottomrule[1pt]
    \end{tabular}
    \subcaption{}\label{fig:matsAndStrikes-b}
  \end{subfigure}
  \caption{ Parameters for the synthetic prices to which we calibrate: (\textbf{a}) maturities;  (\textbf{b}) parameters that define the strikes for the call options per maturity.}
    \label{fig:matsAndStrikes} 
\end{figure}

We now specify a distribution under which we draw the parameters
\[
  \xi = (p_1, p_2, \sigma_0,\sigma_1,\sigma_2, )
\]
for our test. The components are all drawn independently from each
other under the uniform distribution on the respective intervals given
below.
\begin{multicols}{3}
  \begin{itemize}
    \item[-] $I_{p_1}= [0.4,0.5]$
    \item[-] $I_{p_2}= [0.4,0.7]$
    \item[-] $I_{\sigma_0} = [0.5,1.7]$
    \item[-] $I_{\sigma_1}= [0.2,0.4]$
    \item[-] $I_{\sigma_2}= [0.5,1.7]$
  \end{itemize}
\end{multicols}

We can now generate data by the following scheme.

\begin{itemize}
  \item For $m=1,\ldots, 200$ simulate parameters $\xi_m$ under the law
        described above.
  \item For each $m$, compute prices of European calls for maturities
        $T_i$ and strikes $K_{ij}$ for $i=1,\ldots,n=4$ and
        $j=1, \ldots, 20$ according to \eqref{eq:LocVolSDE1} using $10^7$
        Brownian trajectories (for each $m$ we use new trajectories).
  \item Store these prices.
\end{itemize}

\begin{Remark}\label{remark: MC error fail}
  In very few cases, the simulated parameters were such that the implied volatility computation for model prices  failed at least for one maturity due to the remaining Monte Carlo error. In those cases, we simply skip that sample and continue with the next, meaning that we will perform the statistical test only on the samples for which these implied volatility computations were successful.
\end{Remark}

The second part consists of calibrating each of these surfaces and
storing pertinent values for which we conduct a statistical
analysis. In the following we describe the procedure in detail:

Recall that we specify the leverage function $L(t,x)$ via a family of
neural networks, i.e.,
\begin{align*}
  L(t,x) = 1+ F^i(x)  \quad t\in [T_{i-1}, T_{i}), \quad i \in\{1, \ldots, n=4\},
\end{align*}
where $F^i \in \mathcal{N N}_{1,1}$ (see Notation~\ref{not}).  Each
$F^i$ is specified as a $4$-hidden layer feed-forward network where
the dimension of each of the hidden layers is $64$. As activation
function we choose leaky-ReLU\footnote{Recall that $\phi:\mathbb R\rightarrow\mathbb R$ is the leaky-ReLu activation function with parameter $\alpha\in\mathbb R$ if $\phi(x) = \alpha x \mathbbm 1_{(x<0)} + x \mathbbm 1_{(x\ge0)}$. In particular, classical ReLu is is retrieved by setting $\alpha=0$.}

 with parameter $0.2$ for the first three hidden layers and  $\phi=\tanh$ for the last hidden layer. This choice means of course a considerable overparameterization, where we deal with  much more parameters than data points. As is well known from the theory of machine learning, this however allows a profit to be made from implicit regularizations for the leverage function, meaning that the variations of higher derivatives are~small.

\begin{Remark}
  In our experiments, we tested different network architectures. Initially, we used networks with three to five hidden layers with layer dimensions between $50$ and $100$ and activation function $\tanh$ in all layers. Although the training was successful, we observed that training was significantly slower  with significant lower calibration accuracy compared to the final architecture.
  We also tried classical ReLU, but observed that the training sometimes got stuck due to flat gradients. In case of pure leaky-ReLU activation functions, we observed numerical instabilities. By adding a final $\tanh$ activation, this computation was regularized leading to the results we present here.
\end{Remark}

Since closed form pricing formulas are not available for such an LSV model, let
us  briefly specify our pricing method. For the variance reduced
Monte Carlo estimator as of \eqref{eq:calfin} we always use a standard
Euler-SDE discretization with step size $\Delta_t=1/100$. As variance
reduction method, we~implement the running Black–Scholes Delta hedge
with instantaneous running volatility of the price process,
i.e.,~$L(t,X_t)\alpha_t$ is plugged in the formula for the
Black–Scholes Delta as in \eqref{eq: BS hedge formula}.  The~only
parameter that remains to be specified, is the number of trajectories
used for the Monte Carlo estimator which is done in Algorithm
\ref{alg1} and Specification \ref{spec: SGD} below.

As a first calibration step, we calibrate the SABR model
(i.e.,~\eqref{eq: SABR LSV} with $L\equiv 1$) to the synthetic market prices of the first maturity and
fix the calibrated SABR parameters $\nu,\rho$ and $\alpha_0$.
This calibration is not done by the SABR formula, but rather in the same way the LSV model calibration is implemented: we use a Monte Carlo
simulation based engine where gradients are computed via backpropagation. The calibration objective function is analog to
\eqref{eq:calfin} and we compute the full gradient as specified in \eqref{eq: standard grad}. We only use a maximum of 2000 trajectories and the running Black–Scholes hedge for variance reduction
per gradient computation, as we are only interested in an approximate fit. In fact, when compared to a better initial SABR fit achieved by the SABR formula, we observed that the calibration fails more often due to local minima becoming an issue.

For training the parameters $\theta_i$, $i=1, \ldots, 4$, of the neural networks we apply Algorithm~\ref{alg1} in the Appendix~\ref{app:algo}.

\subsection{Numerical Results for the Calibration Test}\label{sec:test}

We now discuss the results of our test. We start by pointing out that from the $200$ synthetic market smiles generated,
four smiles caused difficulties, in the sense that our implied volatility computation
failed due to the remaining Monte Carlo error in the model price computation, compare Remark~\ref{remark: MC error fail}. By increasing the training
parameters slightly (in particular the number of trajectories used in the training), this issue can be mitigated but the resulting
calibrated implied volatility errors stay large out of the money where
the smiles are extreme, and the training will take more time. Hence, we opt to remove those four samples from
the following statistical analysis as they represented unrealistic market smiles.

In Figure~\ref{fig:plot1} we show calibration results for a typical example
of randomly generated synthetic market data. From this it is already visible that the  worst-case calibration error (which occurs out of the money) ranges typically between 5 and 15 basis points. The corresponding calibration
result for the square of the leverage function $L^2$ is given in Figure~\ref{fig: L2 plots}.

Let us note that our method achieves a very high calibration accuracy for the considered range of strikes across all considered maturities. 
This can be seen in the results of a worst-case analysis of calibration errors in Figure~\ref{fig: boxplots}. There we show the mean as well as
different quantiles of the data. Please note that the mean always lies below 10 basis point across all strikes and maturities.

Regarding calibration times, we can report that from the 196 samples, 191 finished within 26 to 27~min. In all these cases, the abort criterion was active on the first time it was checked, i.e.,~after~5000~iterations. The other five samples are examples of smiles comparable to the four where implied volatility computation itself failed. In those cases, more iteration steps where needed resulting in times between 46 and 72 minutes. These samples also correspond to the less successful calibration results.

To perform an out of sample analysis, we
check for extra- and interpolation properties of the learned leverage function. This means that we compute implied volatilities on an extended range and compare to
the implied volatility of the ground truth assumption. The strikes of
these ranges are again computed by taking 20 equally spaced points as
before, but with parameters $k_i$ as of  table Figure~\ref{fig:matsAndStrikes}b
multiplied with 1.5. This has
also the effect that the strikes inside the original range do not
correspond to the strikes considered during training, which  allows for an
additional analysis of the interpolation properties. These results are illustrated in Figure \ref{fig:extra}, from which we see that extrapolation is very close to the local volatility model.

\subsection{Robust Calibration---An Instance of the Adversarial Approach}\label{sec:robust}

Let us now describe a robust version of our calibration methodology realized in an adversarial manner.
We start by assuming that there are multiple ``true''
option prices which correspond to the bid-ask spreads observed  on the market. The way we realize this in our experiment is to use several
local volatility functions that generate equally plausible market implied volatilities.
Recall that the local volatility functions in our statistical test above are functions of the parameters $(p_0,p_1,\sigma_0,\sigma_1,\sigma_2)$. We fix these parameters and generate 4 smiles from local volatility functions with slightly perturbed~parameters
\[  (p_0+u_{i1},p_1+u_{i2},\sigma_0+u_{i3},\sigma_1+u_{i4},\sigma_2+u_{i5})   \text{ for } i=1,\ldots,4, \]
where  $u_{ij}$ are i.i.d. uniformly distributed random variables, i.e., $u_{ij} \sim \mathcal U_{[-u,u]} $ with $u=0.01$. The loss function for maturity $T_i$ in the training part now changes to
\begin{equation}
  \inf_{\theta} \sum_{j=1}^{J_i} w_j \sup_{m=1,\ldots,4}\ell \left( \frac{1}{N}\sum_{n=1}^N X_{j,m}(\theta)(\omega_n)\right),
\end{equation}
with $X_{j,m}$ defined as $X_j$ in \eqref{eq:X} (see also \eqref{eq:Q}) but with synthetic market prices $m=1,\ldots,4$ generated by the $m$-th local volatility function. We are thus in an adversarial situation as described in the introduction: we~have several possibilities for the loss function corresponding to the different market prices and we take the supremum over these (individually for each strike). In our toy example we can simply compute the gradient of this supremum function with respect to $\theta$. In a more realistic situation, where~we do not only have $4$ smiles but a continuum we would iterate the inf and sup computation, meaning that we would also perform a gradient step with respect to  $m$. This corresponds exactly to the adversary part. For a given parameter set $\theta$, the adversary tries to find the worst loss function.

In Figure \ref{fig:plotrob}, we illustrate the result of this robust calibration, where find that the calibrated model lies between the four different smiles over which we take the supremum.

\begin{figure}[H]
  \centering
  \includegraphics[trim=.cm .6cm 0.cm .0cm,
    clip,width=0.48\textwidth]{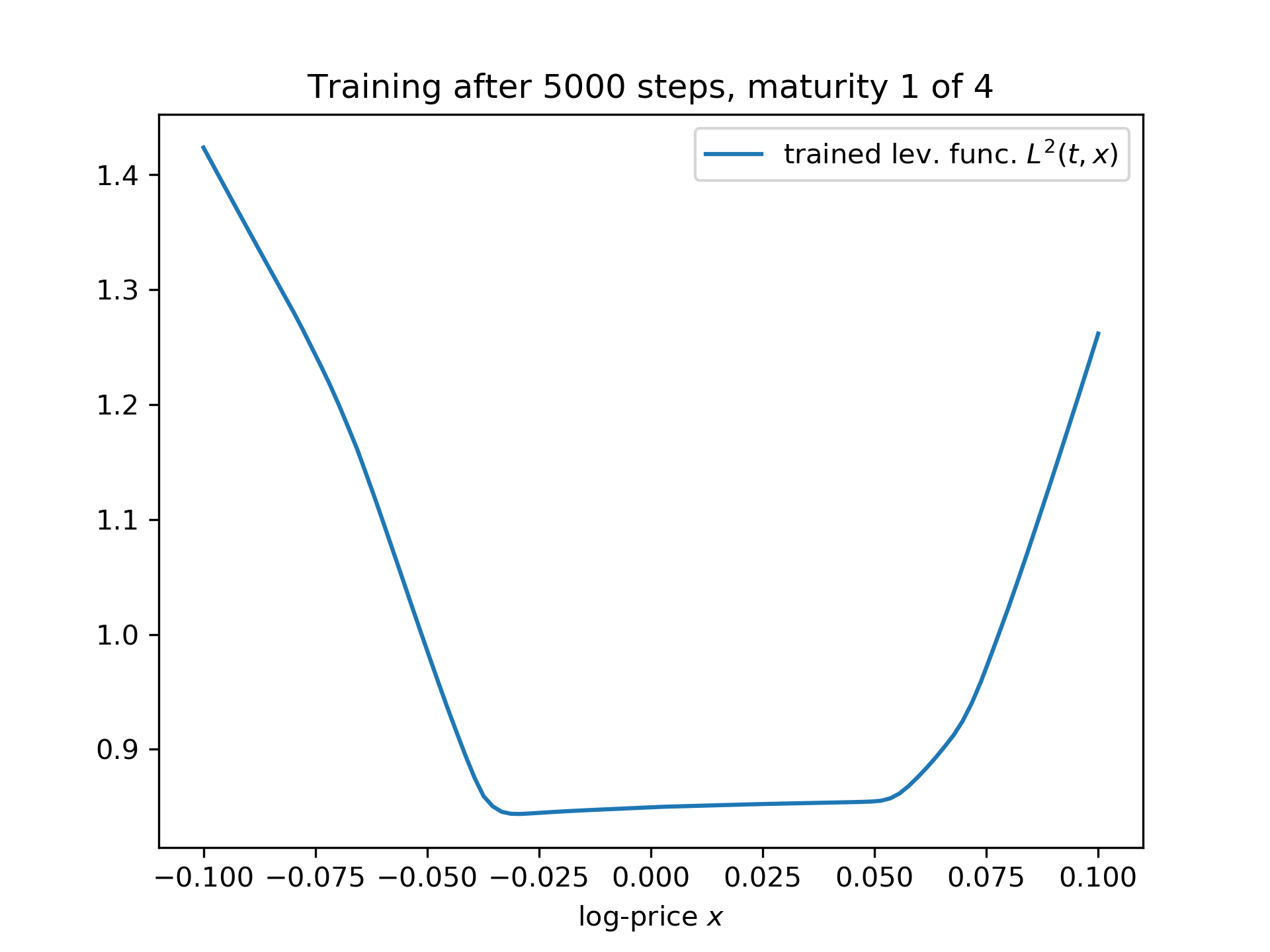}
  \includegraphics[trim=.cm .6cm 0.cm .0cm,
    clip,width=0.48\textwidth]{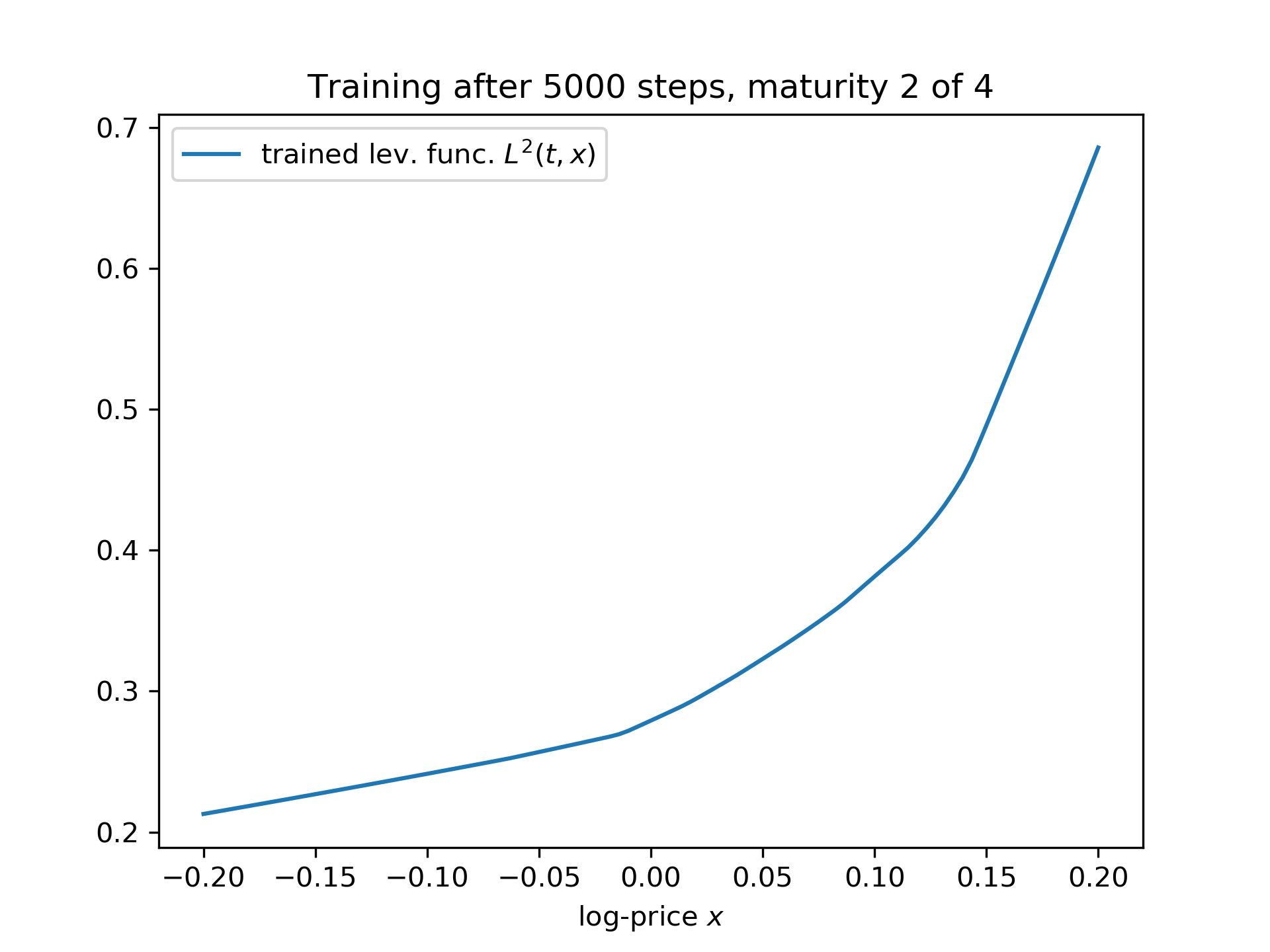}
  \includegraphics[trim=.cm .6cm 0.cm .0cm,
    clip,width=0.48\textwidth]{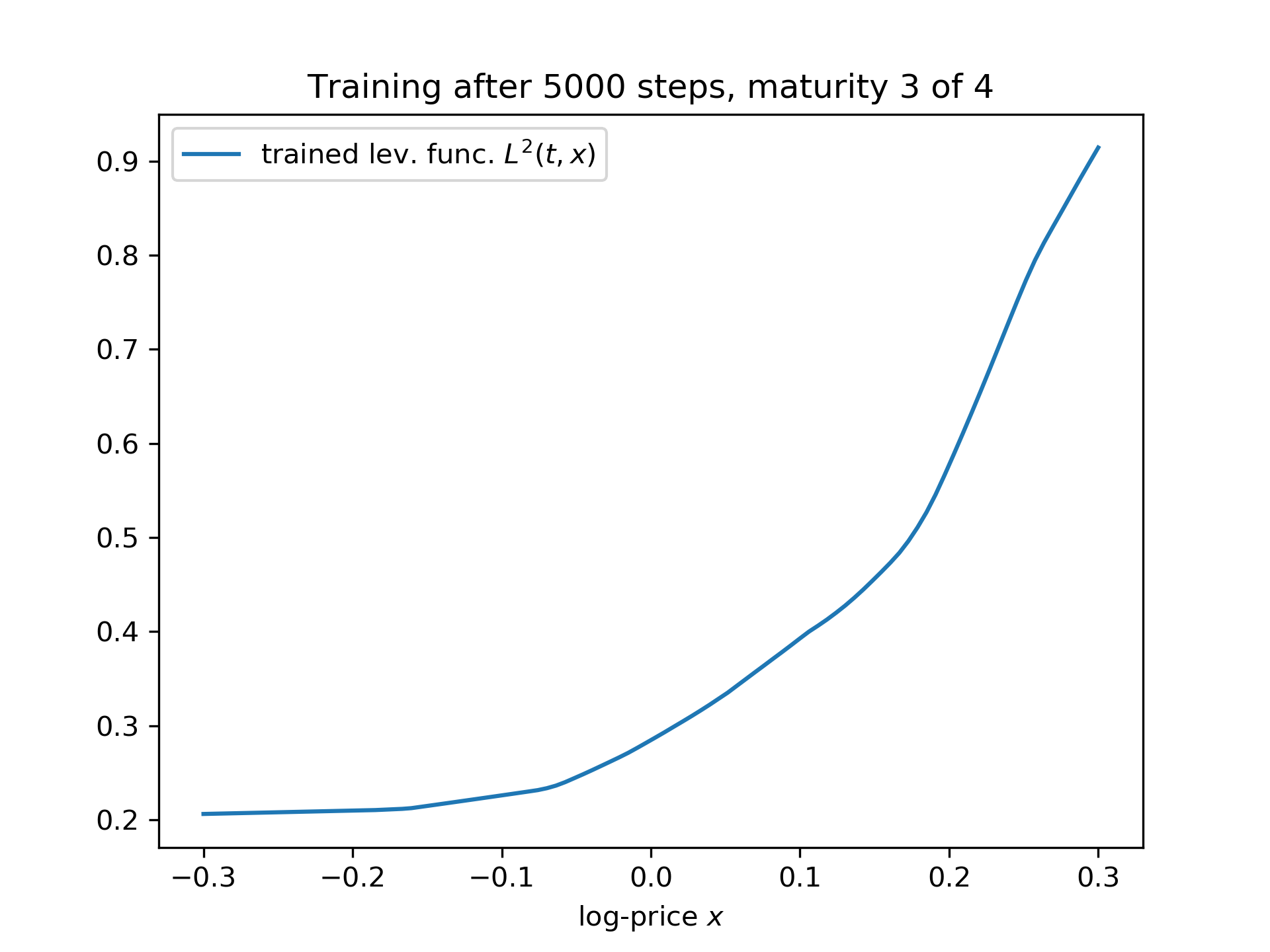}
  \includegraphics[trim=.cm .6cm 0.cm .0cm,
    clip,width=0.48\textwidth]{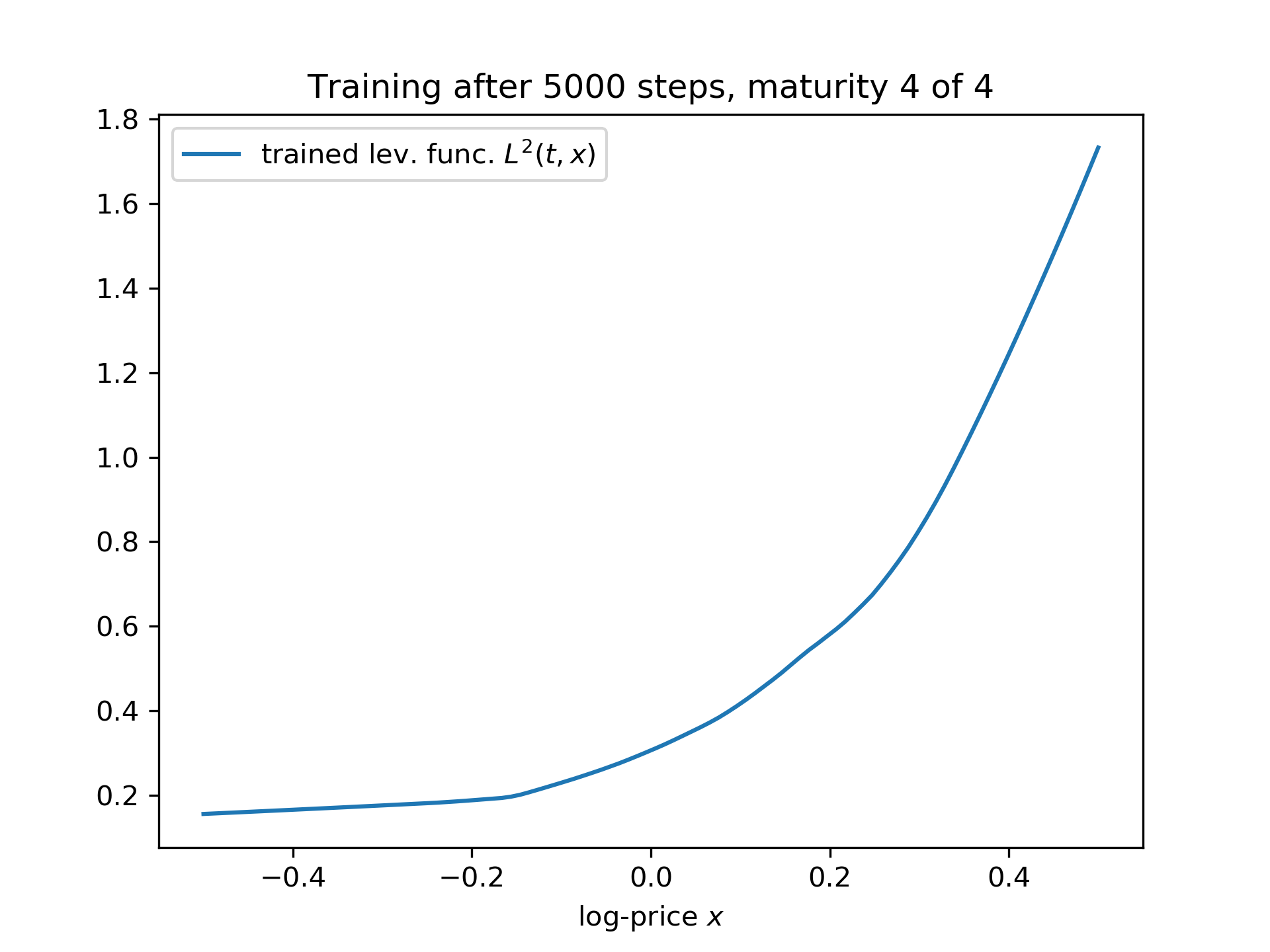}
  \caption{Plot of the calibrated leverage function $x \mapsto L^2(t,x)$ at $t\in \{0,T_1,T_2,T_3\}$ in the example shown in Figure~\ref{fig:plot1}. The $x$-axis is given in log-moneyness $\ln(K/S_0)$.}
  \label{fig: L2 plots}
\end{figure}\unskip

\begin{figure}[H]
  \centering \includegraphics[trim=.8cm 0.4cm 1.6cm 0.7cm,
    clip,width=0.44\textwidth]{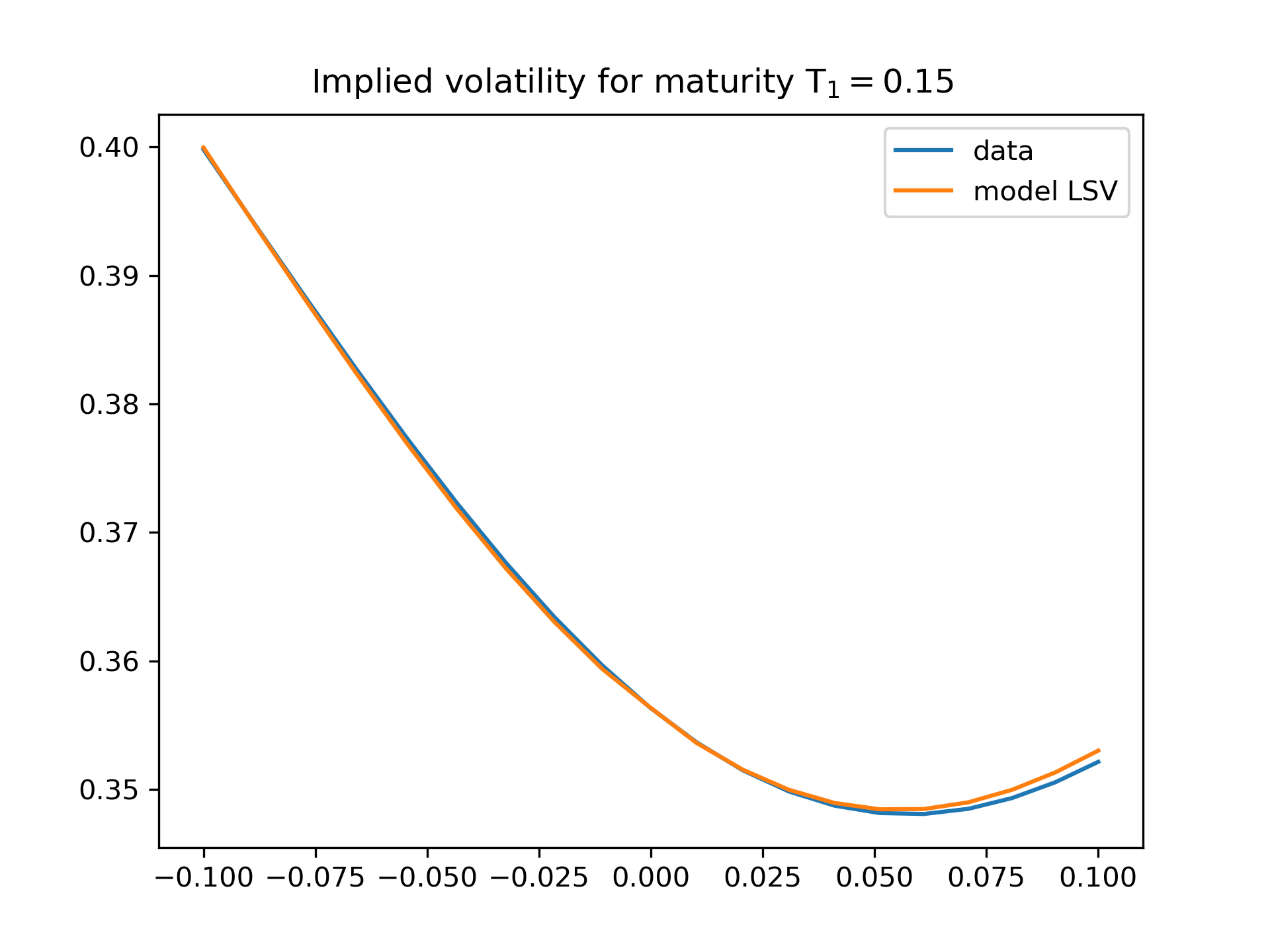}
  \includegraphics[trim=0.3cm 0.4cm 1.6cm 0.7cm,
    clip,width=0.44\textwidth]{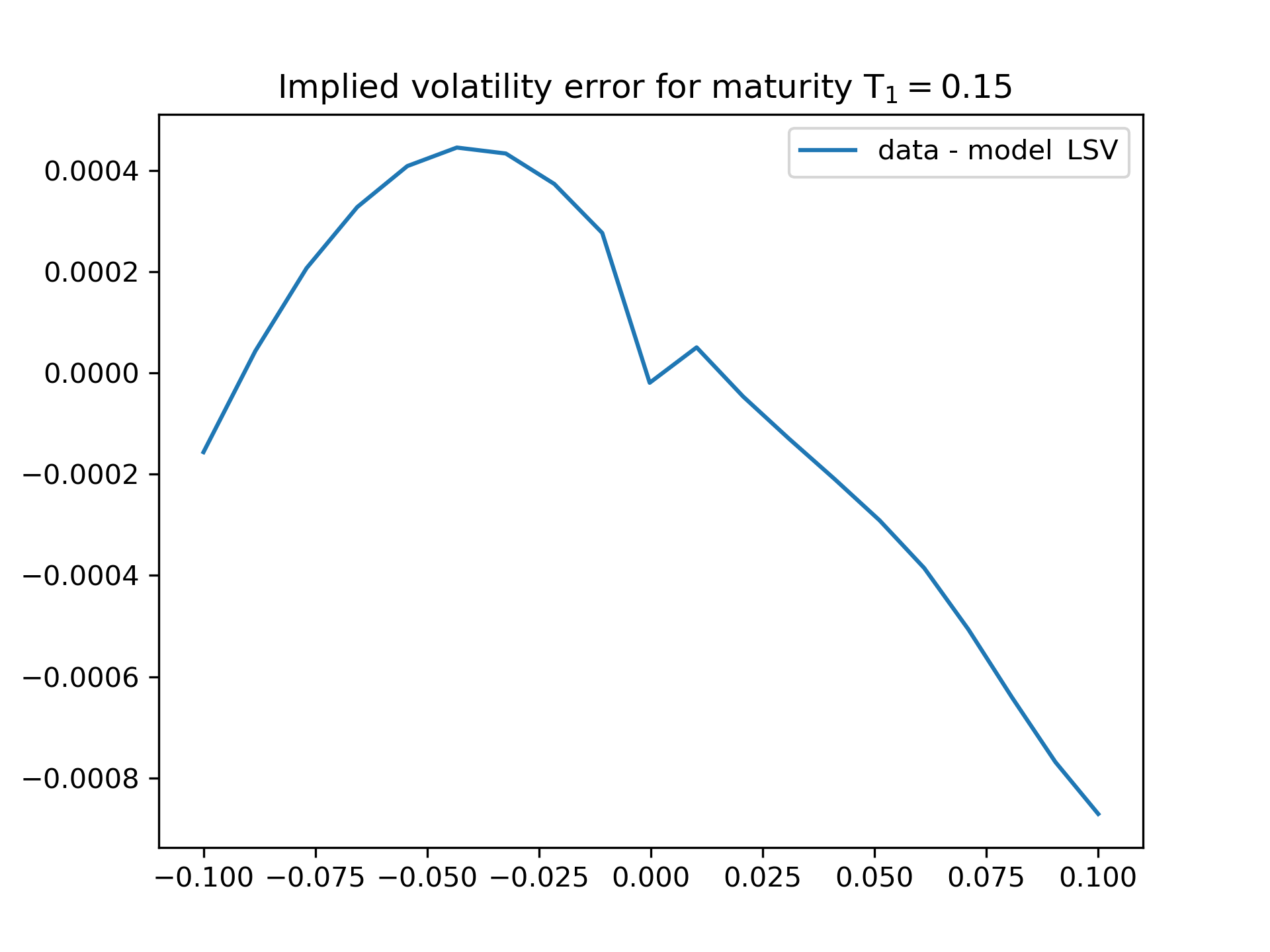}
  \includegraphics[trim=.8cm 0.4cm 1.6cm 0.7cm,
    clip,width=0.44\textwidth]{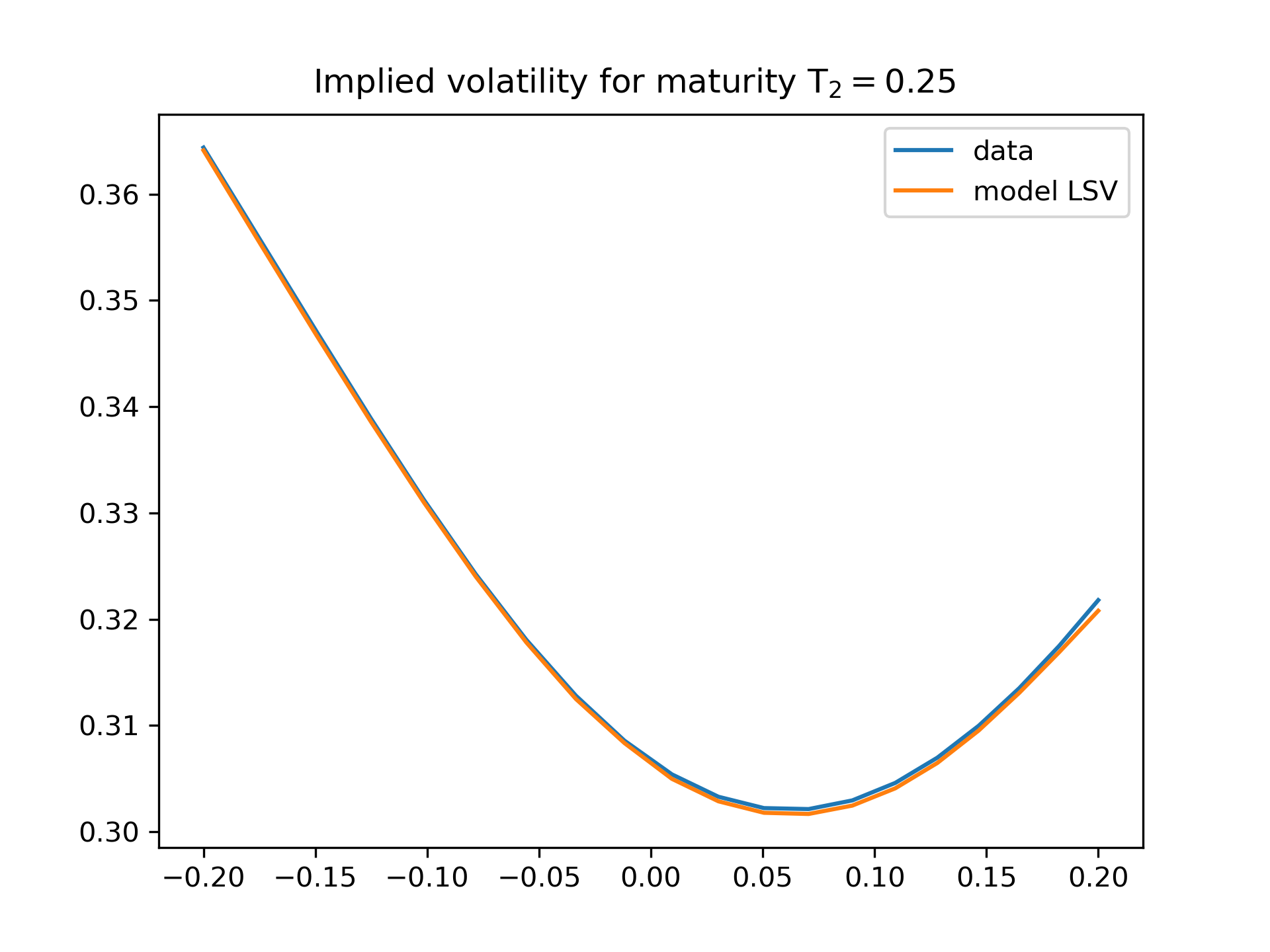} 
  \includegraphics[trim=0.3cm 0.4cm 1.6cm 0.7cm,
    clip,width=0.44\textwidth]{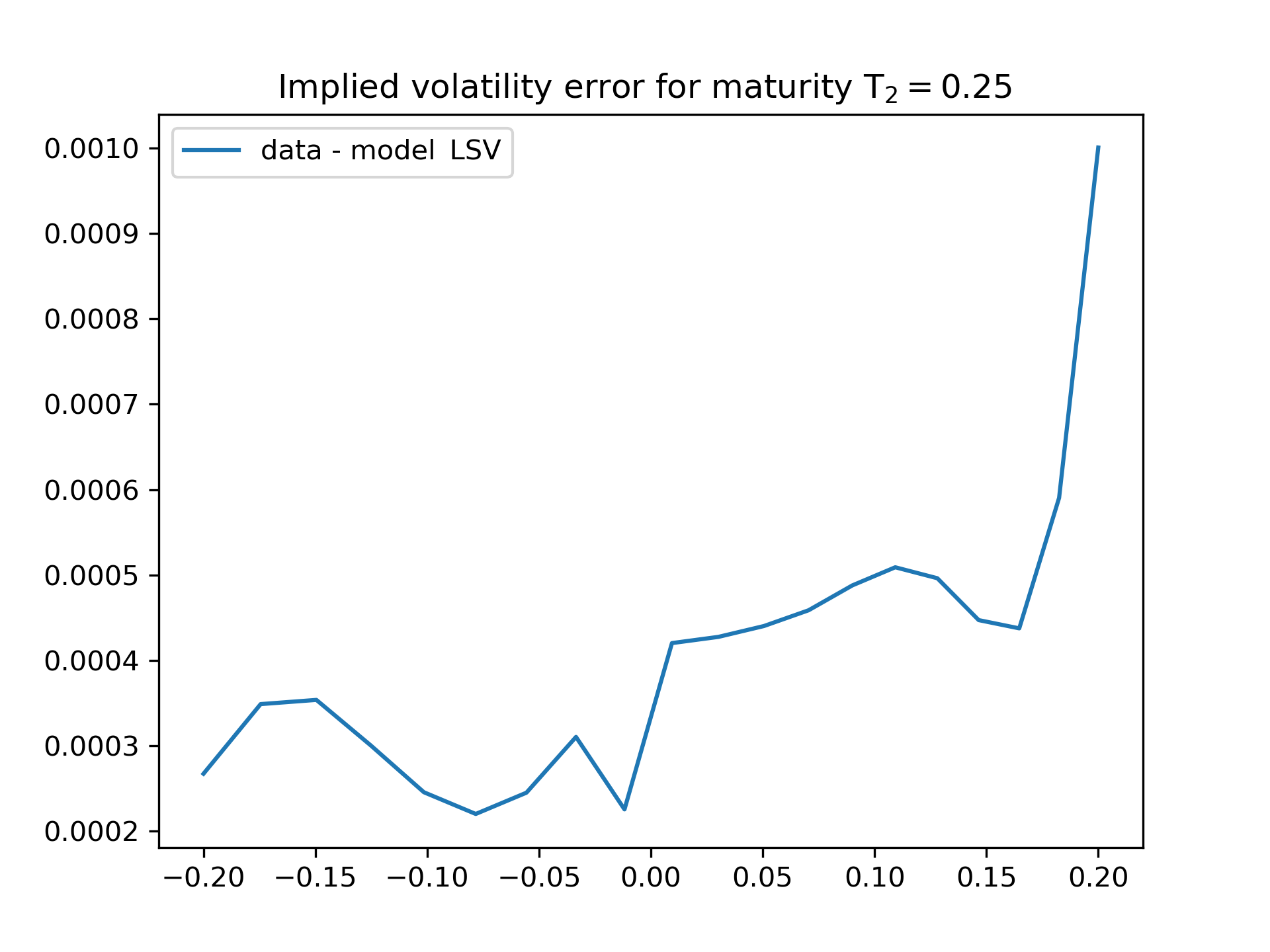}
  \includegraphics[trim=.8cm 0.4cm 1.6cm 0.7cm,
    clip,width=0.44\textwidth]{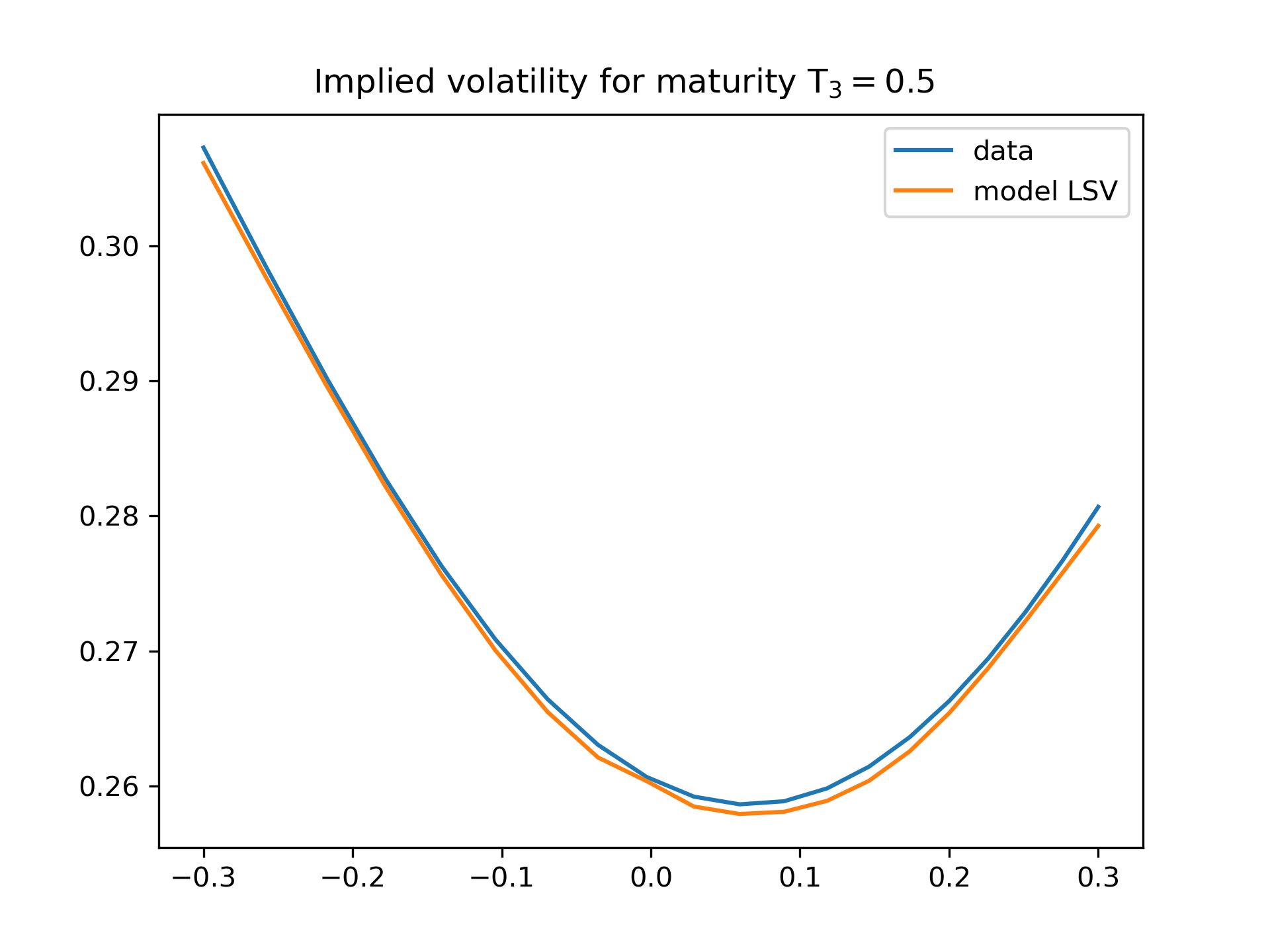} 
  \includegraphics[trim=0.3cm 0.4cm 1.6cm 0.7cm,
    clip,width=0.44\textwidth]{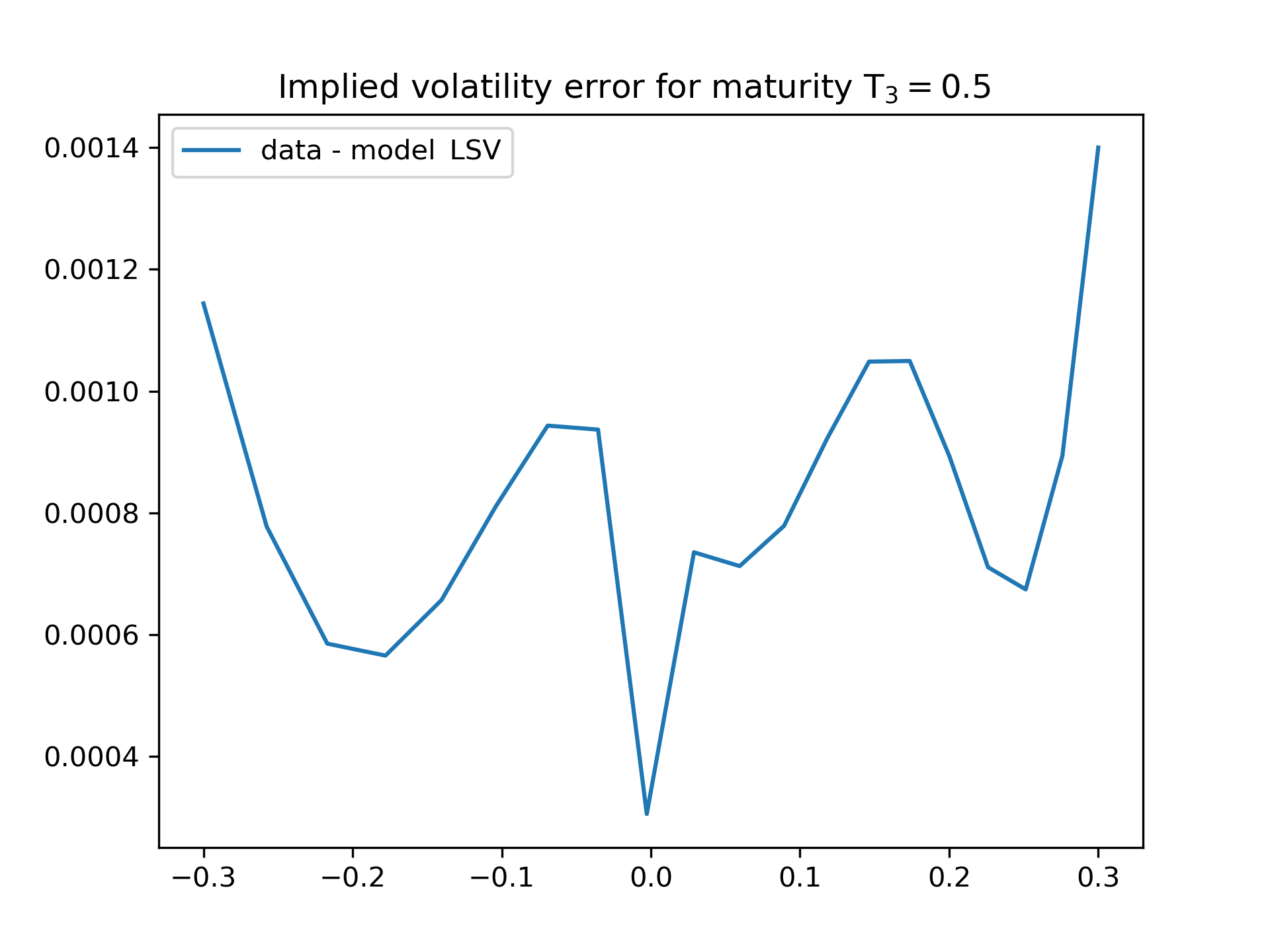}
  \includegraphics[trim=.8cm 0.4cm 1.6cm 0.7cm,
    clip,width=0.44\textwidth]{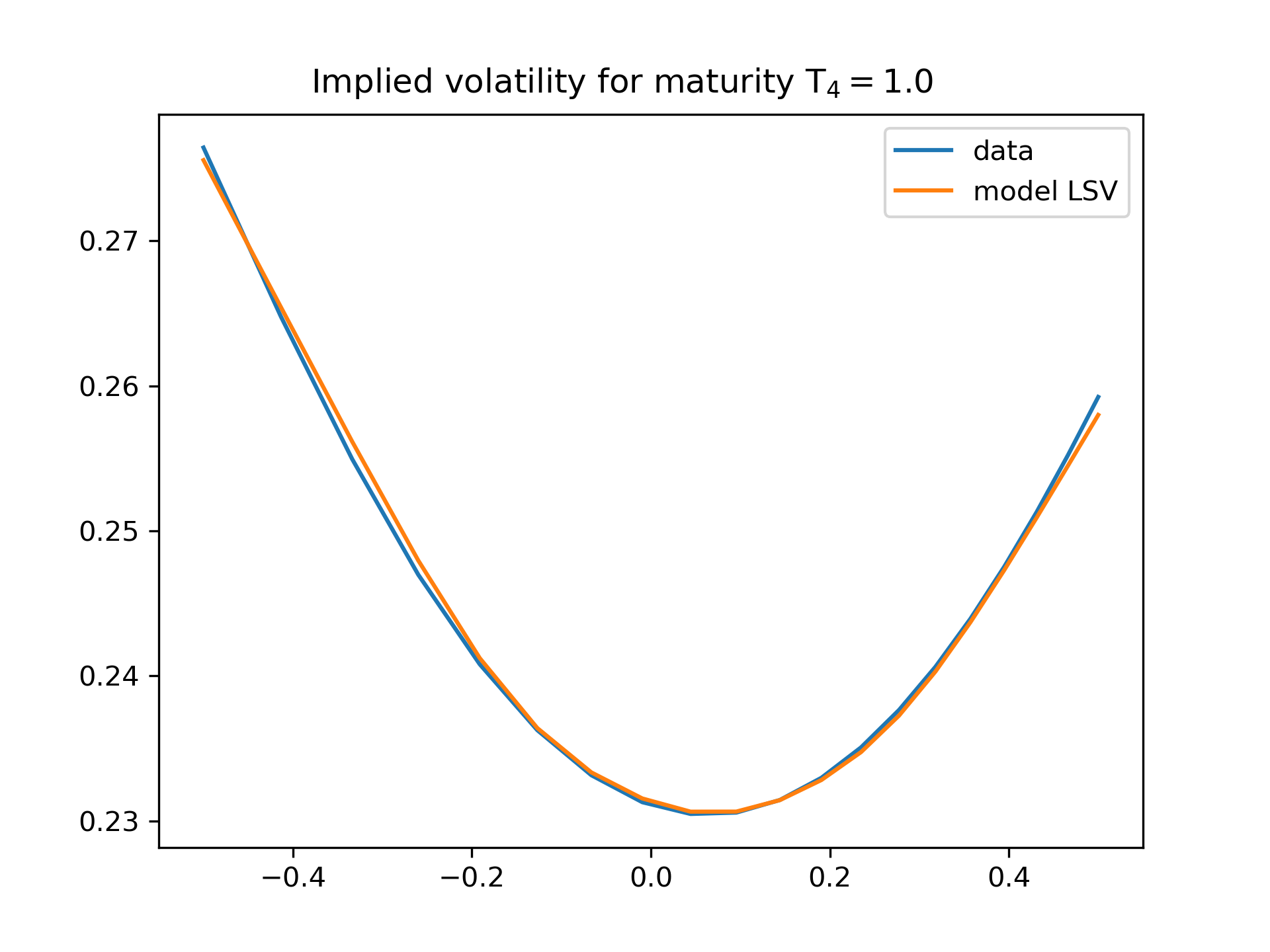} 
  \includegraphics[trim=0.3cm 0.4cm 1.6cm 0.7cm,
    clip,width=0.44\textwidth]{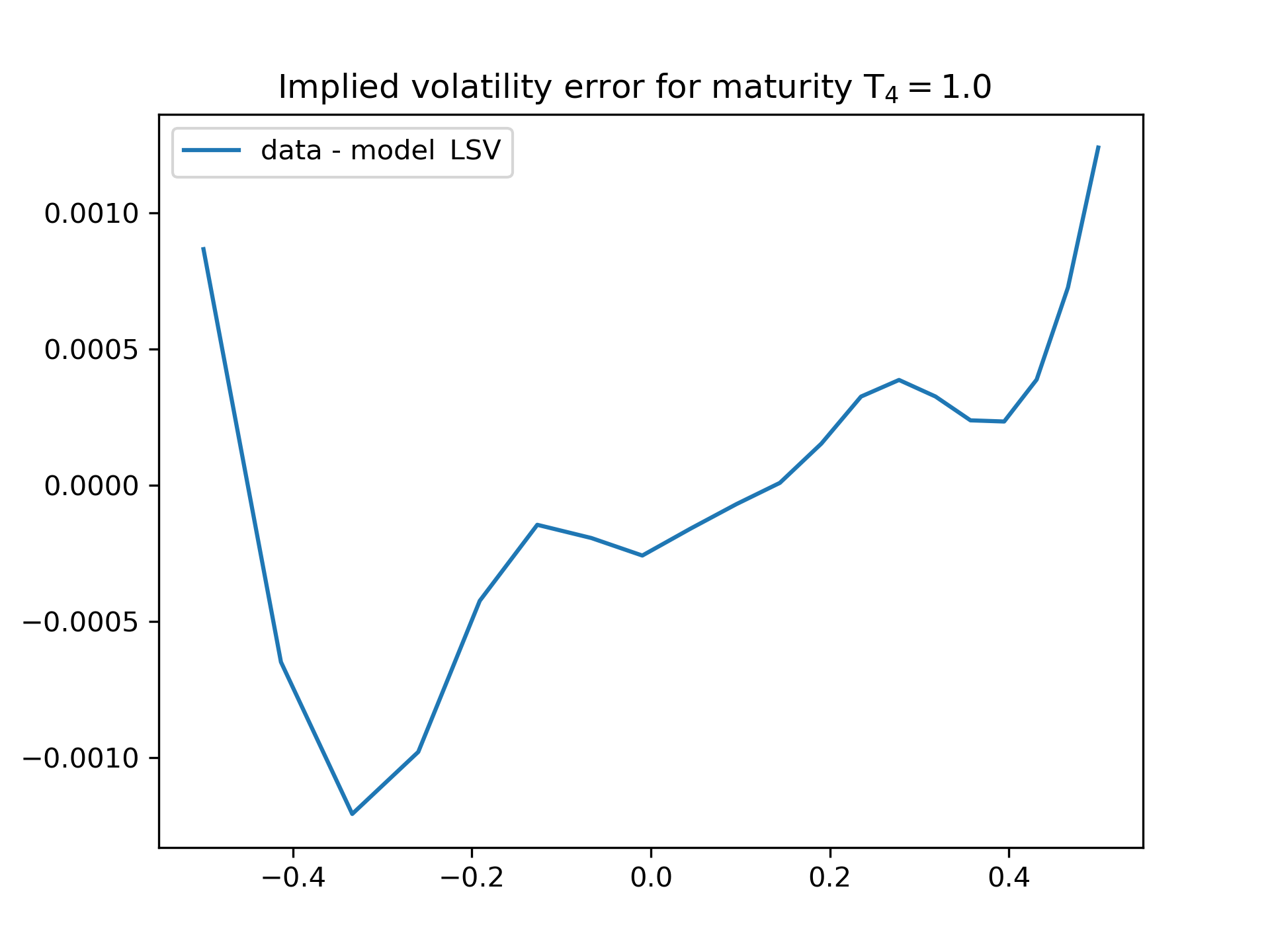}
  \caption{Left column: implied volatilities for the calibrated model together with the data (synthetic market) implied volatilities for a typical example of a synthetic market sample for all available
  maturities. Right column: calibration errors by subtracting model implied volatilities from the data implied volatilities. The $x$-axis is given in log-moneyness $\ln(K/S_0)$.}
  \label{fig:plot1} 
\end{figure}

  \begin{figure}[H]
    \centering
    \includegraphics[width=1\textwidth]{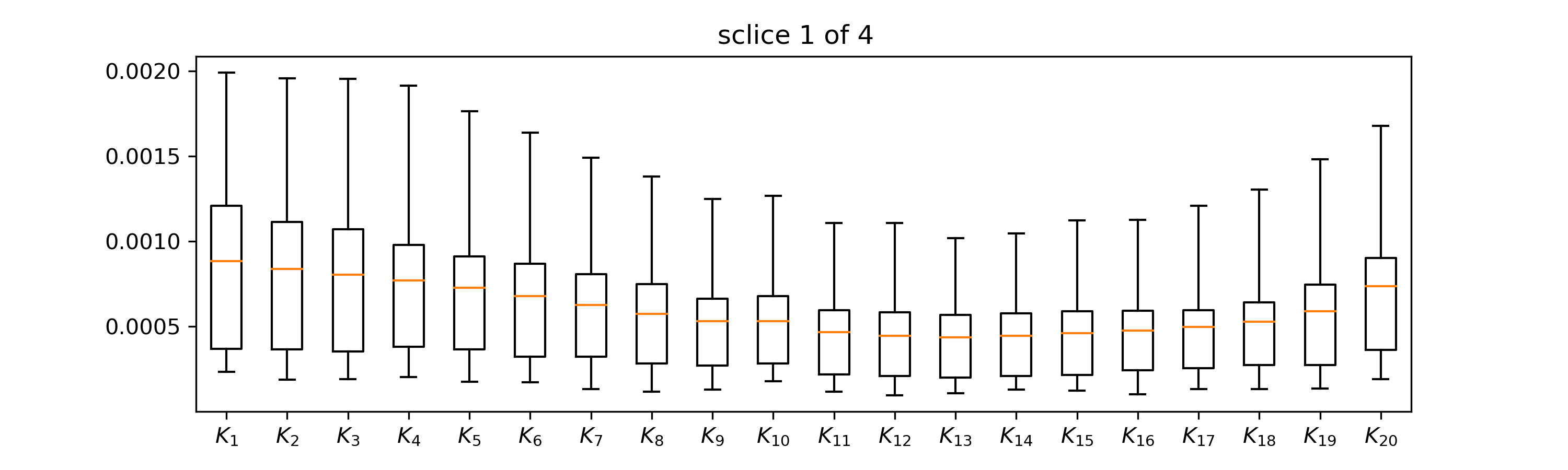} 
    \includegraphics[trim=.0cm 0.0cm 0cm .0cm, clip,
      width=1\textwidth]{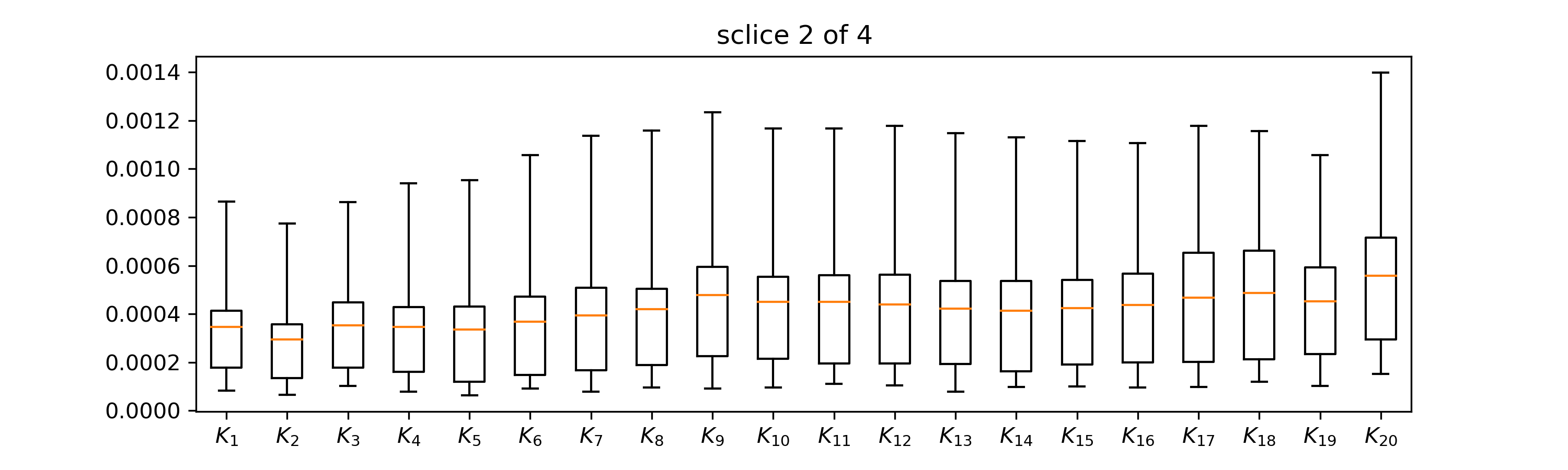}
    \includegraphics[width=1\textwidth]{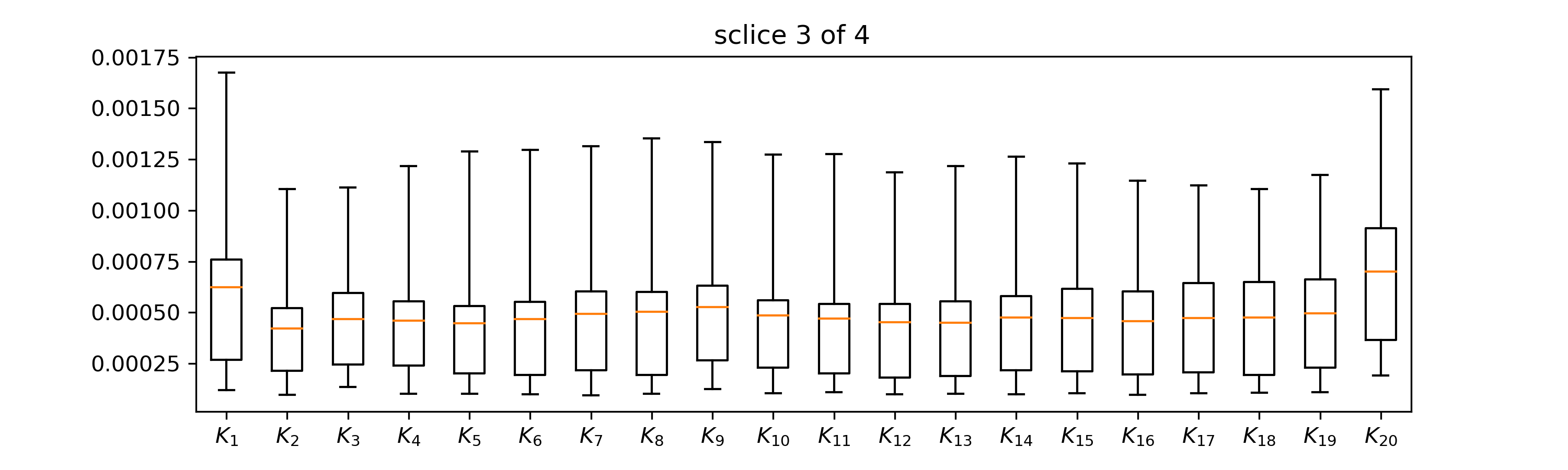}
    \includegraphics[width=1\textwidth]{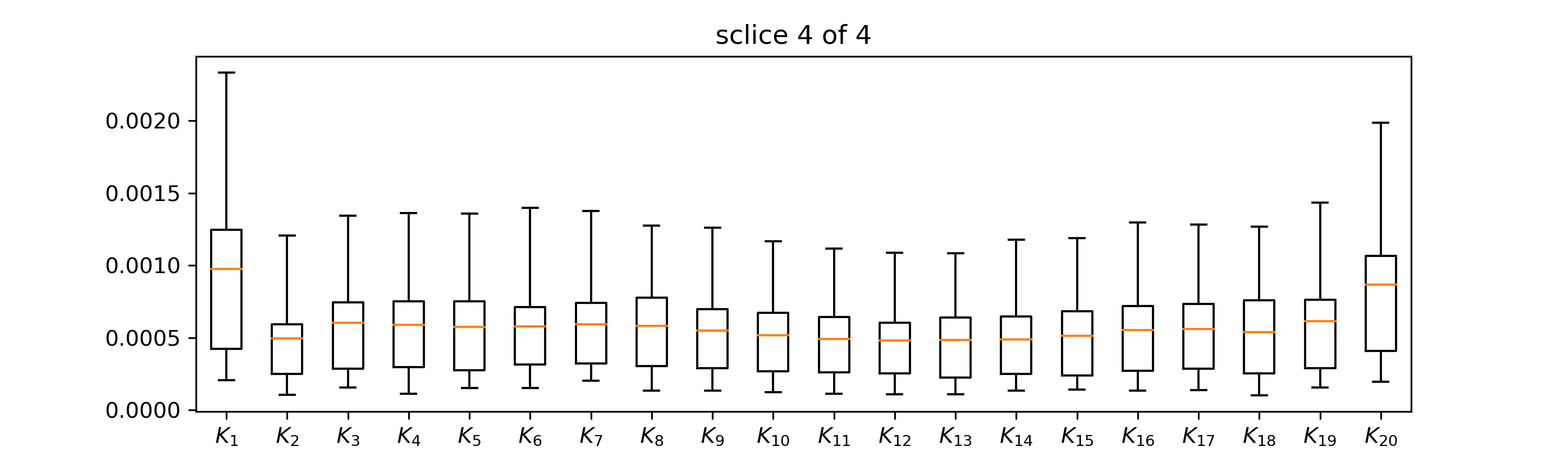}

    \caption{Boxplots of absolute calibration errors of implied volatilities for the statistical test as specified in Section~\ref{sec:performance test}
      for the four synthetic market data slices (maturities). 
      The errors for $K_j$ in the $i$-th row correspond to the calibration error of the synthetic market implied volatility for strike $K_{i,j}$ .
      Depicted are the mean (horizontal line), as well as
      the $0.95,0.70,0.3,0.15$ quantiles for the absolute calibration error per strike. }
    \label{fig: boxplots}
  \end{figure}

\begin{figure}[H]
  \centering
  \includegraphics[trim=.8cm 0.1cm 1.6cm .0cm,
    clip,width=0.46\textwidth]{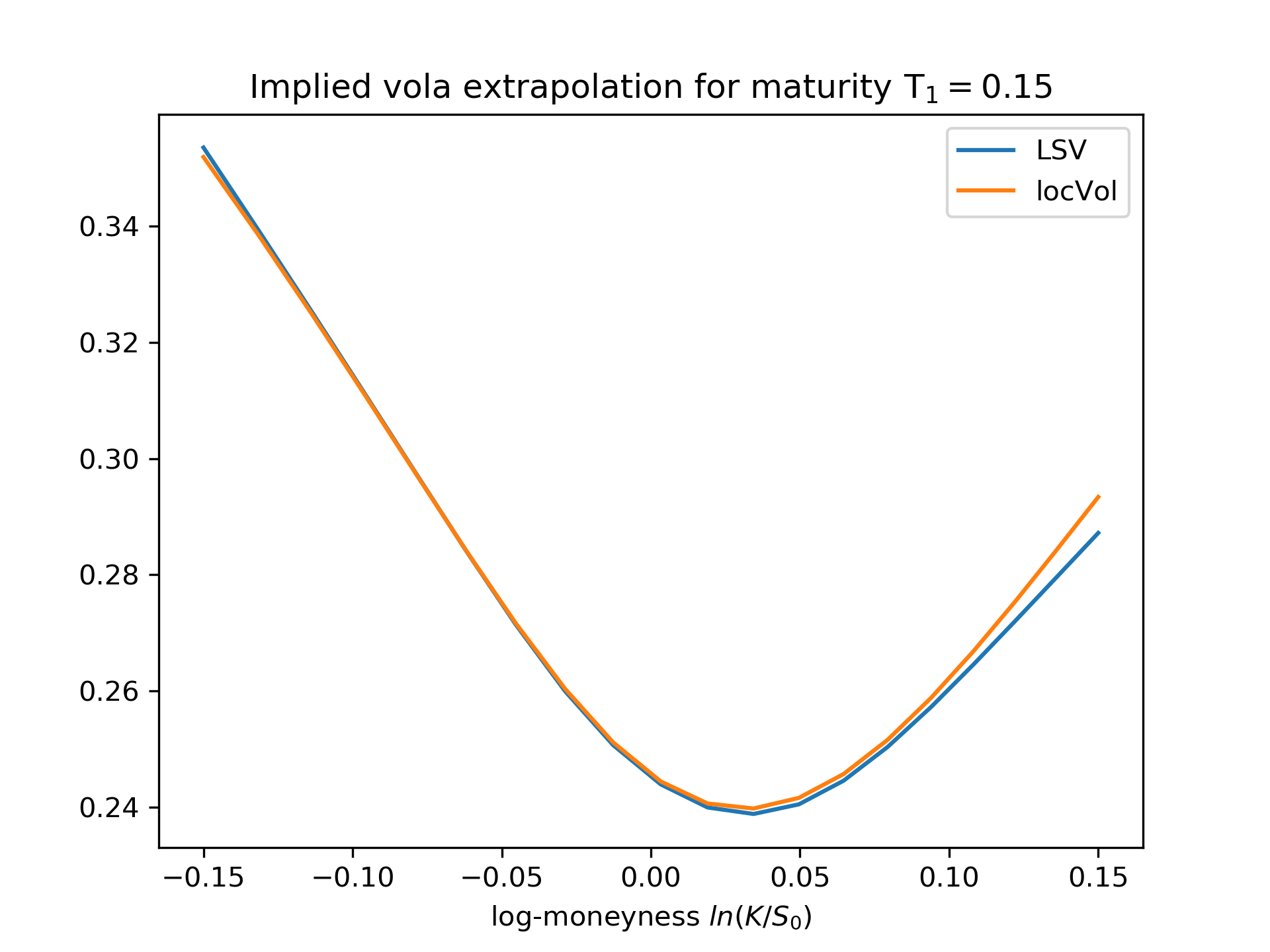}
  \includegraphics[trim=0.8cm 0.1cm 1.6cm .0cm,
    clip,width=0.46\textwidth]{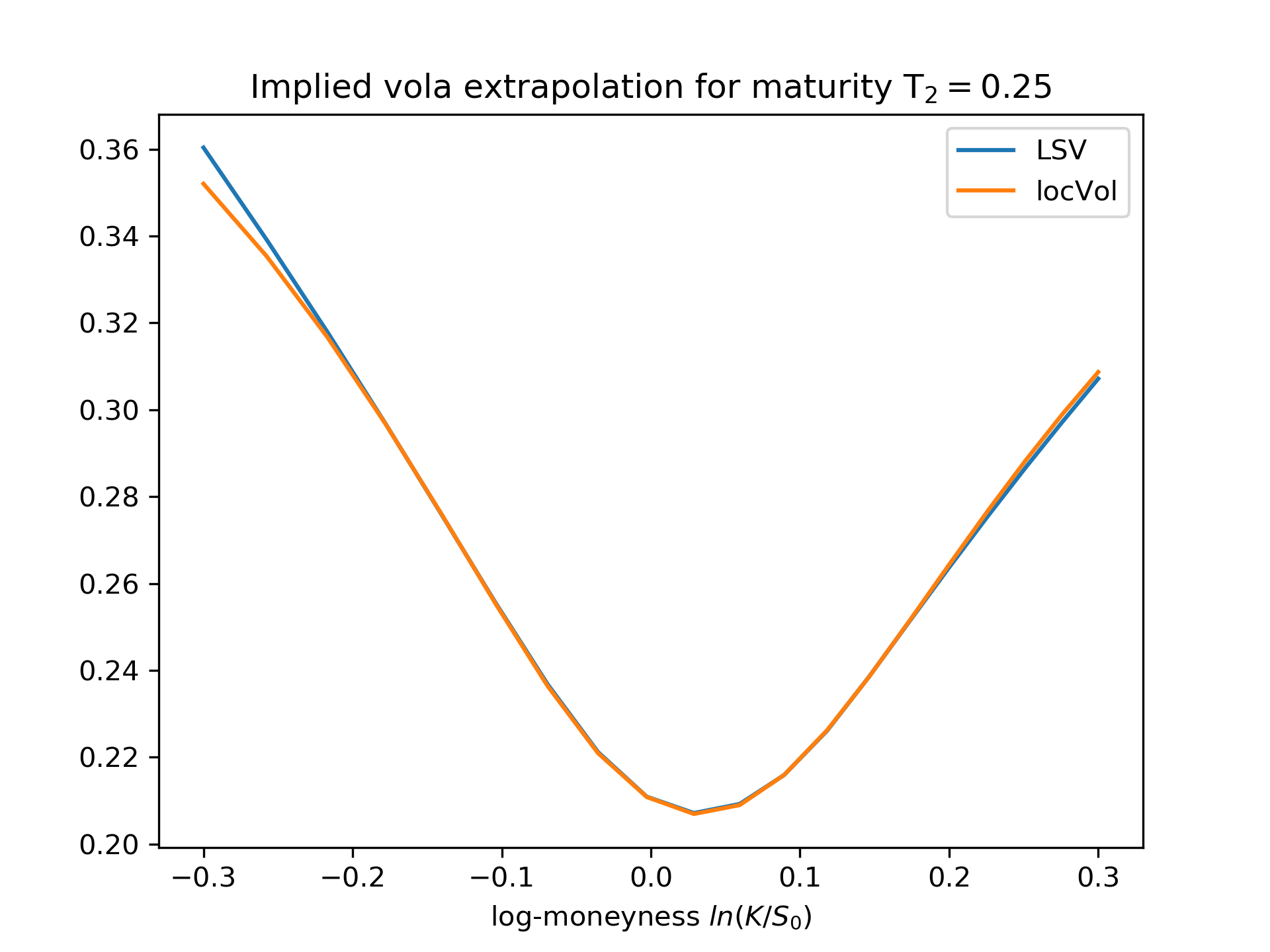}

    \includegraphics[trim=.8cm 0.1cm 1.6cm .0cm,
    clip,width=0.46\textwidth]{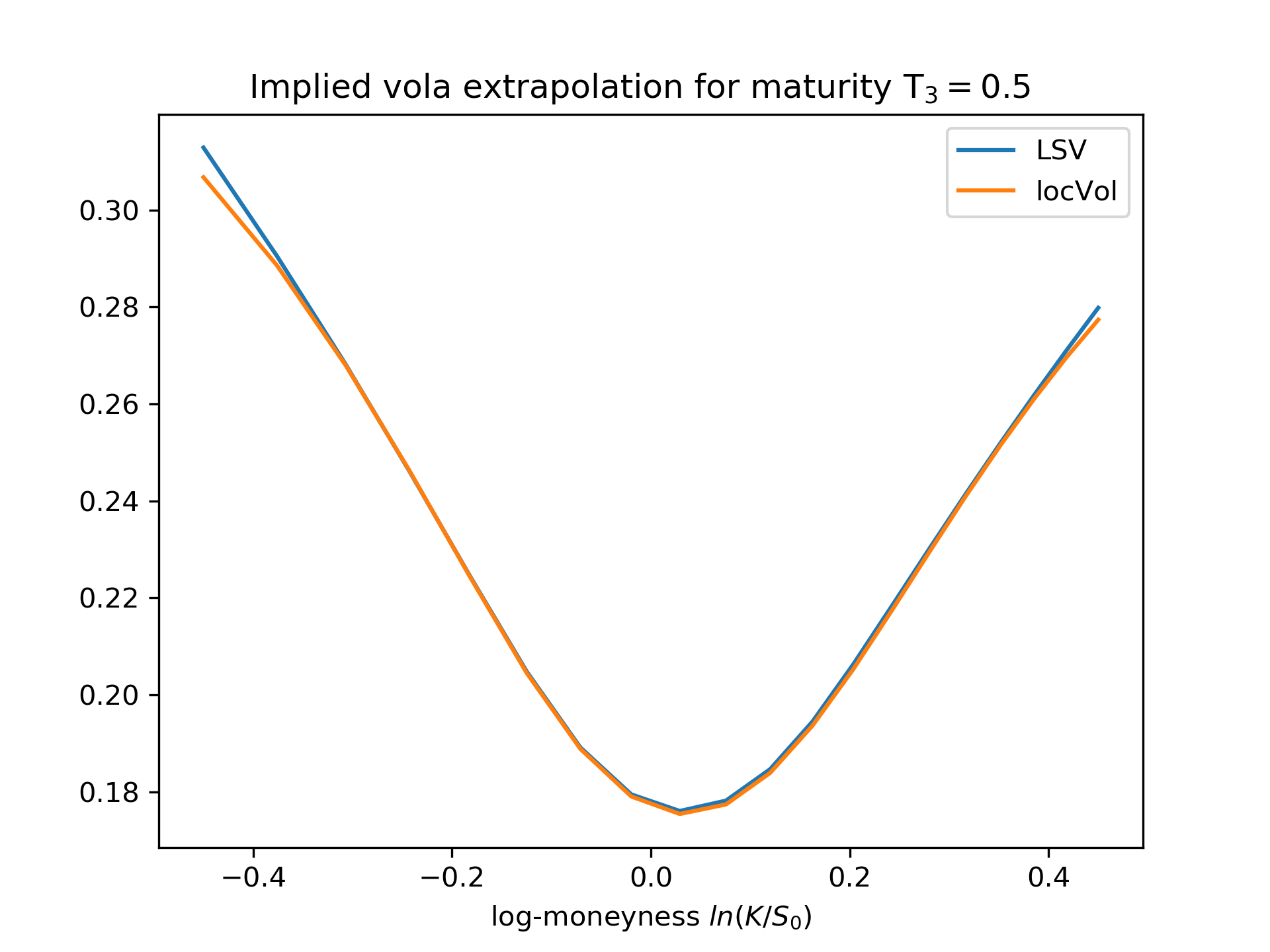} 
  \includegraphics[trim=0.8cm 0.1cm 1.6cm .0cm,
    clip,width=0.46\textwidth]{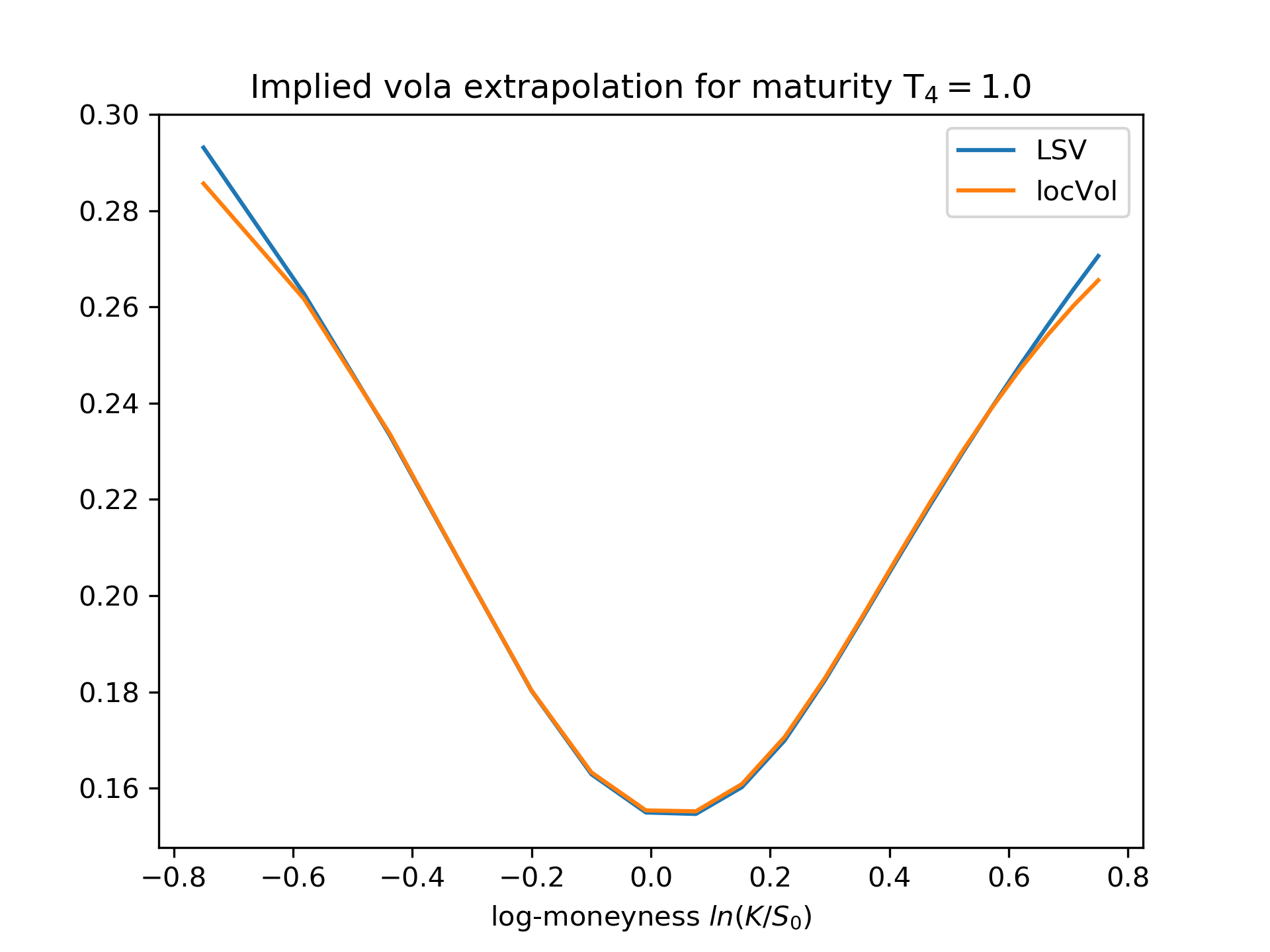}

  \caption{Extra- and interpolation as described in Section \ref{sec:test} between the synthetic prices of the ground truth assumption against the corresponding calibrated SABR-LSV model. Plots are shown for all four considered maturities $\{T_1,\ldots,T_4\}$ as defined in Figure \ref{fig:matsAndStrikes}a. The $x$-axis is given in log-moneyness $\ln(K/S_0)$.}
   \label{fig:extra}

\end{figure}
\unskip
\begin{figure}[H]
  \centering \includegraphics[trim=.8cm 0.2cm 1.6cm .9cm,
    clip,width=0.46\textwidth]{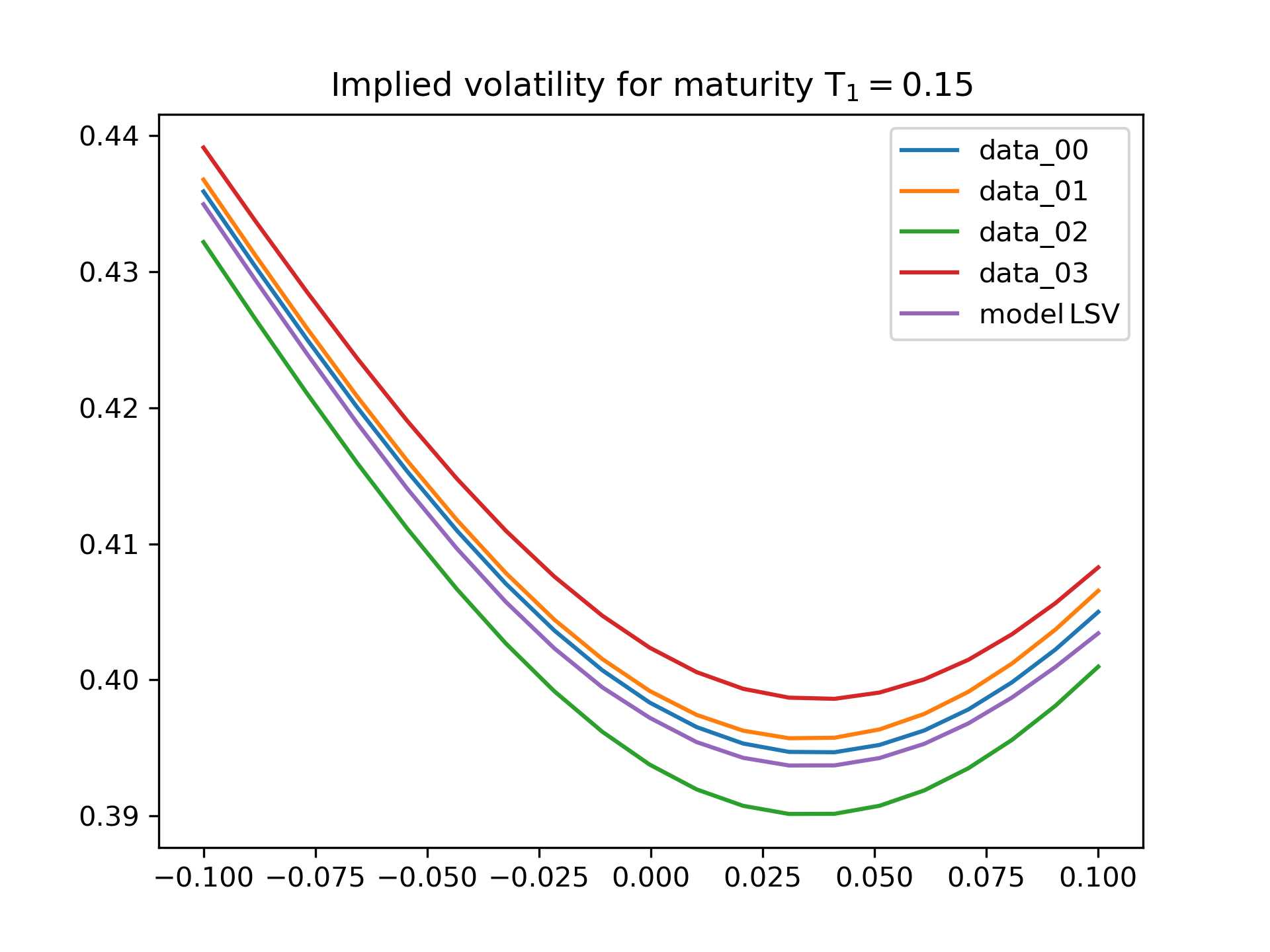}
  \includegraphics[trim=0.3cm 0.2cm 1.6cm .9cm,
    clip,width=0.48\textwidth]{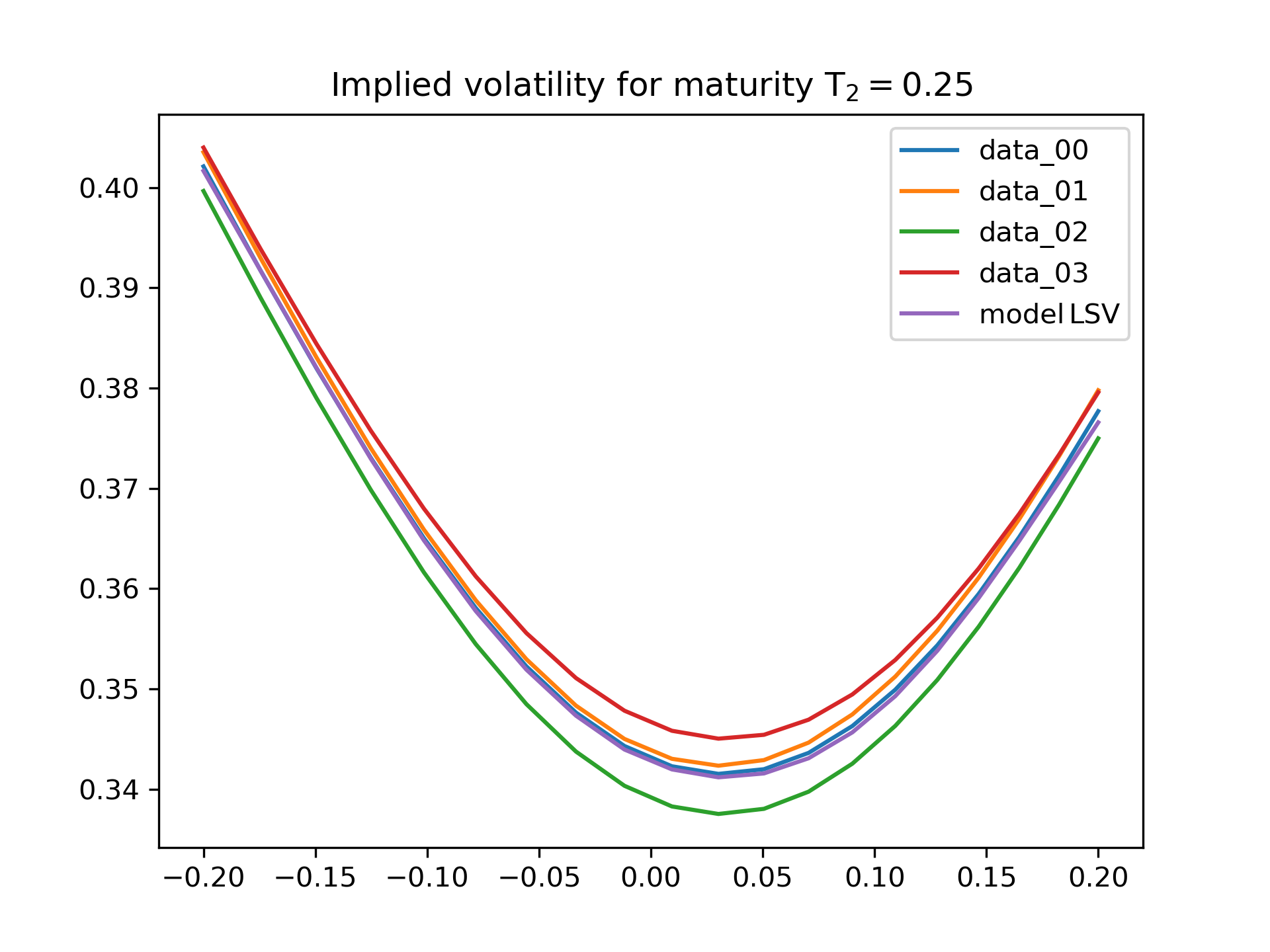}
  \includegraphics[trim=.8cm 0.2cm 1.6cm .9cm,
    clip,width=0.46\textwidth]{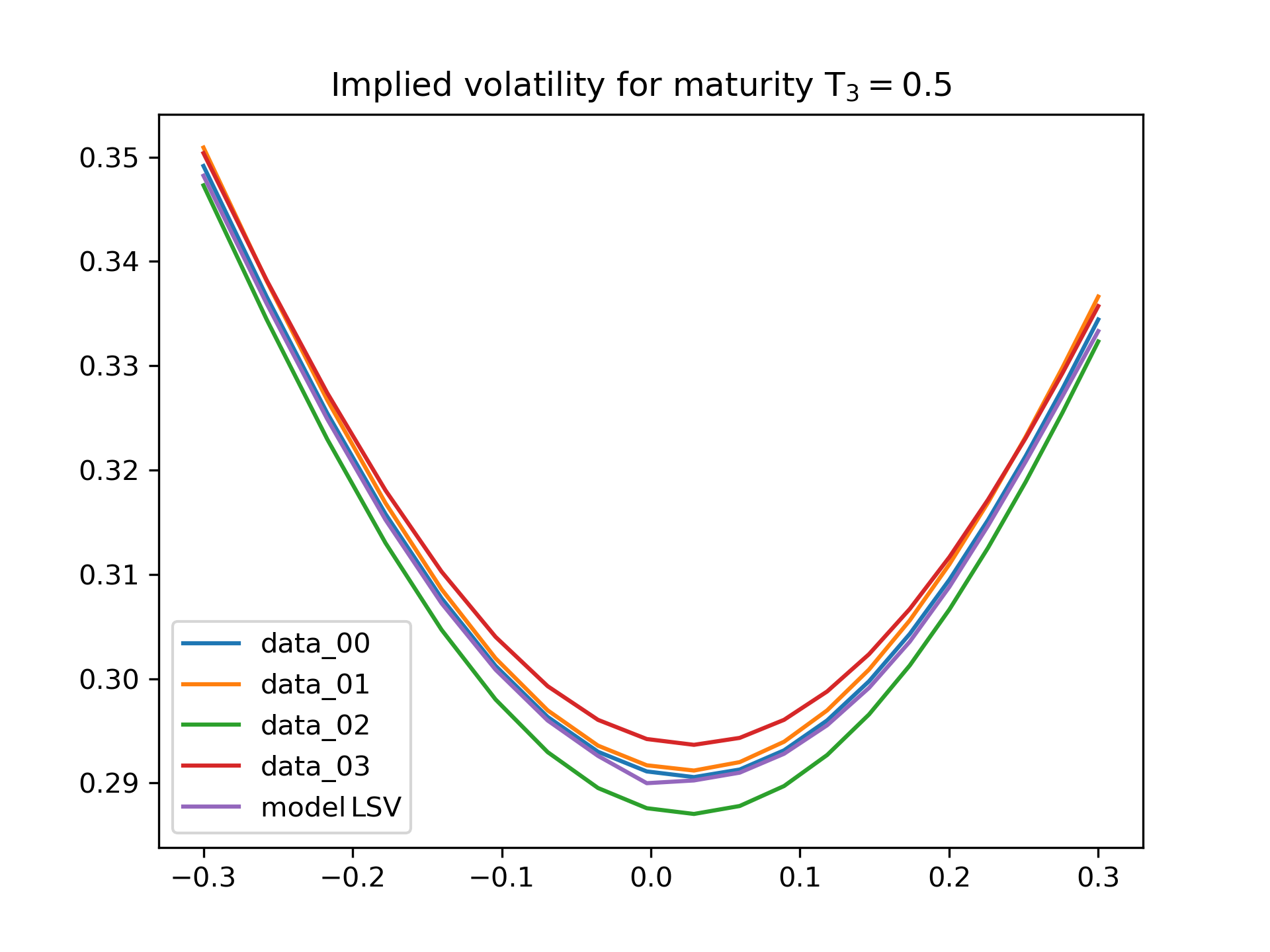}
  \includegraphics[trim=0.3cm 0.2cm 1.6cm .9cm,
    clip,width=0.48\textwidth]{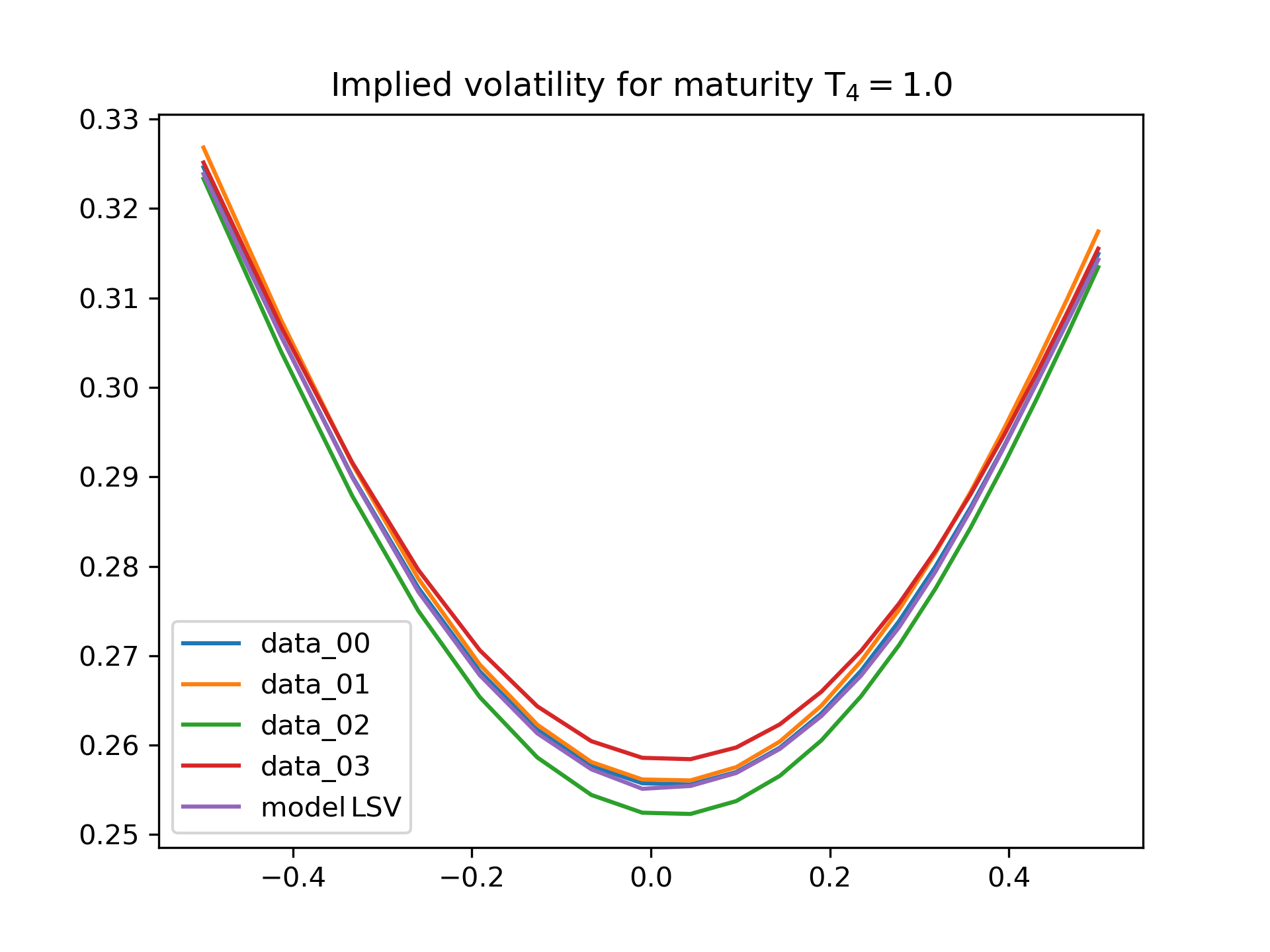}
  \caption{Robust calibration as described in Section \ref{sec:robust} for all four maturities, the $x$-axis is given in log-moneyness $\ln(K/S_0)$.}
  \label{fig:plotrob}
\end{figure}

\section{Plots} \label{chap: plots and histograms}

This section contains the relevant plots for the numerical test
outlined in Section \ref{sec:numerics}. 

\section{Conclusions}

We have demonstrated how the parametrization by means of neural
networks can be used to calibrate local stochastic volatility models
to implied volatility data. We make the following remarks:
\begin{enumerate}

  \item The method we presented does not require any form of
        interpolation for the implied volatility surface since we do not
        calibrate via  Dupire's formula. As the interpolation is usually done ad hoc, this might be a
        desirable feature of our method.
  \item Similar to \cite{GH:11, GH:13}, it is possible to ``plug in'' any stochastic variance process 
        such as rough volatility processes as long as an efficient
        simulation of trajectories is possible.
  \item The multivariate extension is straight forward.
  \item The level of accuracy of the calibration result is of a very
        high degree. The average error in our statistical test is of around 5 to 10 basis points, which is an interesting feature in its own right. We~also observe good  extrapolation and generalization properties of the calibrated leverage~function.
  \item The method can be significantly accelerated by applying
        distributed computation methods in the context of multi-GPU
        computational concepts.
  \item
        The presented algorithm is further able to deal with path-dependent options since all computations are done by means of Monte
        Carlo simulations.
  \item We can also consider the instantaneous variance process of the price process as short end of a forward variance process, which is assumed to follow (under appropriate assumptions) a neural SDE. This setting, as an infinite-dimensional version of the aforementioned ``multivariate'' setting, then qualifies for joint calibration to S\&P and  VIX
   options. This is investigated in a companion~paper.
  \item We stress again the advantages of the generative  adversarial network point of view. 
  We believe that this is a crucial feature  in  the joint calibration of S\&P and  VIX options.
\end{enumerate}

\appendix
\section{Variations of Stochastic Differential Equations}\label{sec:SDE}

We follow here the excellent exposition  of
\cite{pro:90} to understand the dependence of solutions of
stochastic differential equations on parameters, in particular when
we aim to calculate derivatives with respect to parameters of neural
networks.

Let us denote by $ \mathbb{D} $ the set of real-valued, c\`adl\`ag,
adapted processes on a given stochastic basis
$(\Omega,\mathcal{F},\mathbb{Q})$ with a filtration (satisfying
usual conditions). By $\mathbb{D}^n $ we denote the set of
$\mathbb{R}^n$-valued, c\`adl\`ag, adapted processes on the same
basis.

\begin{Definition}
  An operator $ F $ from $\mathbb{D}^n$ to $ \mathbb{D} $ is called
  \emph{functional Lipschitz} if for any $ X,Y \in \mathbb{D}^n $
  \begin{enumerate}
    \item the property
          $ X^{\tau-}=Y^{\tau-} $ implies $ F(X)^{\tau-} = F(Y)^{\tau-} $ for any stopping time $ \tau$,
    \item there exists an increasing process $ {(K_t)}_{t \geq 0}$
          such that for $t \geq 0$
          \[
            \| F(X)_t - F(Y)_t \| \leq K_t \sup_{r \leq t} \| X_r - Y_r \|.
          \]
  \end{enumerate}
\end{Definition}

Functional Lipschitz assumptions are sufficient to obtain existence
and uniqueness for general stochastic differential equations, see
\cite[Theorem V 7]{pro:90}.

\begin{Theorem}\label{e_u:thm}
  Let $ Y=(Y^1,\ldots,Y^d) $ be a vector of semimartingales starting
  at $ Y_0 = 0 $,  $ (J^1,\ldots,J^n)  \in \mathbb{D}^n $
  a vector of processes and let $ F^i_j $,
  $ i =1,\ldots,n$, $j=1,\ldots,d$ be functionally Lipschitz operators. Then there is a unique process
  $ Z \in \mathbb{D}^n$ satisfying
  \[
    Z^i_t = J^i_t + \sum_{j=1}^d \int_0^t F^i_j(Z)_{s-} dY^j_s
  \]
  for $ t \geq 0 $ and $ i =1,\ldots,n $. If $ J $ is a
  semimartingale, then $ Z $ is a semimartingale as well.
\end{Theorem}

With an additional uniformity assumption on a sequence of stochastic
differential equations with converging coefficients and  initial data  we obtain stability, see
\cite[Theorem V 15]{pro:90}.

\begin{Theorem}\label{sta:thm}
  Let $ Y=(Y^{1},\ldots,Y^{d}) $ be vector of semimartingales
  starting at $ Y_0 = 0 $. Consider for $\varepsilon \geq 0$, a vector of processes  $ (J^{\varepsilon,1},\ldots,J^{\varepsilon,n})  \in \mathbb{D}^n$ and functionally
  Lipschitz operators $ F^{\varepsilon,i}_j $ for $ i =1,\ldots,n$,
  $j=1,\ldots,d$. Then, for $\varepsilon \geq 0$, there is a
  unique process $ Z^{\varepsilon} \in \mathbb{D}^n$ satisfying
  \[
    Z^{\varepsilon,i}_t = J^{\varepsilon,i}_t + \sum_{j=1}^d \int_0^t
    F^{\varepsilon,i}_j(Z^{\varepsilon})_{s-} dY^j_s
  \]
  for $ t \geq 0 $ and  $ i =1,\ldots,n $. If $ J^{\varepsilon} \to J^0 $ in ucp,
  $ F^{\varepsilon}(Z^0) \to F^{0}(Z^0) $ in ucp, then
  $ Z^{\varepsilon} \to Z^0 $ in ucp.
\end{Theorem}

\begin{Remark}\label{rem:appl}
  We shall apply these theorems to a local stochastic volatility model of the form
  \[
    d S_t(\theta) = S_t(\theta) L(t,S_t (\theta), \theta) \alpha_t dW_t \, ,
  \]
  where $\theta \in \Theta$,  $ (W,\alpha) $ denotes some  Brownian motion together with an adapted,
  c\`adl\`ag stochastic process $ \alpha $ (all on a given stochastic
  basis) and $S_0 > 0 $ is some real number.

  We assume that for each $\theta \in \Theta$
  \begin{align}\label{ass1:gen}
    (t,s) \mapsto L(t,s, \theta)
  \end{align}
  is bounded, c\`adl\`ag in $t$ (for fixed $ s > 0 $), and globally
  Lipschitz in $ s $ with a Lipschitz constant independent of $ t $ on
  compact intervals . In this case, the map
  \[
    S \mapsto S_{\cdot} L(\cdot , S_{\cdot}, \theta)
  \]
  is functionally Lipschitz and therefore the above equation has a
  unique solution for all times $t$ and any $ \theta $ by Theorem
  \ref{e_u:thm}. If, additionally,
  \begin{align}\label{ass2:gen}
    \lim_{\theta \to \widehat{\theta}}\, \sup_{(t,s)} | L(t,s , \theta) - L(t,s , \widehat{\theta})|=0,
  \end{align}
  where the $\sup$ is taken over some compact set, then
  we also have that the solutions $ S(\theta) $ converge ucp to
  $ S(\widehat{\theta}) $, as $ \theta \to \widehat{\theta} $ by Theorem \ref{sta:thm}.
\end{Remark}

\section{Preliminaries on Deep Learning}\label{sec:NN}

We shall here briefly introduce two core concepts in deep learning,
namely \emph{artificial neural networks} and \emph{stochastic
  gradient descent}. The latter is a widely used optimization method
for solving maximization or minimization problems involving the
first. In standard machine-learning terminology, the optimization
procedure is usually referred to as ``training''. We shall use both
terminologies~interchangeably.

\subsection{Artificial Neural Networks}

We start with the definition of \emph{feed-forward neural networks.} 
These are functions obtained by composing layers consisting of an
affine map and a componentwise nonlinearity.  They serve as
universal approximation class which is stated in Theorem
\ref{th:universalapprox}.  Moreover,
derivatives of these functions can be efficiently expressed
iteratively (see e.g.~\cite{hecht1992theory}), which is a desirable
feature from an optimization point of view.

\begin{Definition}
  Let $M, N_0, N_1, \ldots, N_M \in \mathbb{N}$,
  $\phi: \mathbb{R} \to \mathbb{R}$ and for any
  $m \in \{1, \ldots, M\}$, let
  $w_{m}: \mathbb{R}^{N_{m-1}} \to \mathbb{R}^{N_{m}},\, x \mapsto
    A_{m} x + b_{m}$ be an affine function with
  $A_{m} \in \mathbb{R}^{N_{m} \times N_{m-1}} $ and
  $b_{m} \in \mathbb{R}^{N_{m}}$. A function
  $\mathbb{R}^{N_0} \to \mathbb{R}^{N_{M}}$ defined as
  \[
    F(x) = w_M \circ F_{M-1} \circ \cdots \circ F_{1}, \quad
    \text{with } F_{m}=\phi \circ w_{m} \quad \text{for } m \in \{1,
    \ldots, M-1\}
  \]
  is called a \emph{feed-forward neural network}.  Here the
  \emph{activation function} $\phi$ is applied componentwise. $M-1$
  denotes the number of hidden layers and $N_1, \ldots, N_{M-1}$
  denote the dimensions of the hidden layers and $N_0$ and $N_M$ the
  dimension of the input and output layers.
\end{Definition}

\begin{Remark}
 Unless otherwise stated, the activation functions $\phi$ used in this article are always
  assumed to be smooth, globally bounded with bounded first
  derivative.
\end{Remark}

The following version of the so-called \emph{universal approximation
  theorem} is due to K.~Hornik (\mbox{\citealp{H:91}}). An earlier version was
proved by G.~Cybenko (\mbox{\citealp{C:92}}). To formulate the result, we denote
the set of all feed-forward neural networks with activation function
$\phi$, input dimension $N_0$ and output dimension $N_M$ by
$\mathcal{NN}^{\phi}_{\infty, N_0, N_M}$.

\begin{Theorem}[Hornik (1991)]\label{th:universalapprox}
  Suppose $\phi$ is bounded and nonconstant. Then the following
  statements hold:
  \begin{enumerate}
    \item For any finite measure $\mu$ on
          $(\mathbb{R}^{N_0}, \mathcal{B}(\mathbb{R}^{N_0}))$ and
          $1 \leq p < \infty$, the set
          $\mathcal{NN}^{\phi}_{\infty, N_0, 1}$ is dense in
          $L^p(\mathbb{R}^{N_0},\mathcal{B}(\mathbb{R}^{N_0}), \mu)$.
    \item If in addition $\phi \in C(\mathbb{R}, \mathbb{R})$, then
          $\mathcal{NN}^{\phi}_{\infty, N_0, 1}$ is dense in
          $C(\mathbb{R}^{N_0}, \mathbb{R})$ for the topology of uniform
          convergence on compact sets.
  \end{enumerate}
\end{Theorem}

Since each component of an $\mathbb{R}^{N_M}$-valued neural network is
an $\mathbb{R}$-valued neural network, this~result easily generalizes
to $\mathcal{NN}^{\phi}_{\infty, N_0, N_M}$ with $N_M >1$.

\begin{Notation}\label{not}
  We denote by $\mathcal{NN}_{N_0, N_M}$ the set of all neural
  networks in $\mathcal{NN}^{\phi}_{\infty, N_0, N_M}$ with a
  \emph{fixed architecture,} i.e.,~a fixed number of hidden layers
  $M-1$, fixed input and output dimensions $N_{m}$ for each hidden
  layer $m \in \{1, \ldots, M-1\}$ and a fixed activation function
  $\phi$.  This set can be described by
  \[
    \mathcal{NN}_{ N_0, N_M}=\{F(\cdot,\theta) \, |\, F \text{ feed
      forward neural network and } \theta \in \Theta\},
  \]
  with parameter space $\Theta \in \mathbb{R}^q$ for some
  $q \in \mathbb{N}$ and $\theta \in \Theta$ corresponding to the
  entries of the matrices $A_{m}$ and the vectors $b_{m}$ for
  $m \in \{1, \ldots, M\}$.
\end{Notation}

\subsection{Stochastic Gradient Descent}\label{sec:SDG}

In light of Theorem \ref{th:universalapprox}, it is clear that neural
networks can serve as function approximators. To~implement this, the
entries of the matrices $A_{m}$ and the vectors $b_{m}$ for
$m \in \{1, \ldots, M\}$ are subject to optimization. If the unknown
function can be expressed as the expected value of a stochastic
objective function, one widely applied optimization method is
\emph{stochastic gradient descent,} which we shall review~below.

Indeed, consider the following minimization problem
\begin{align}\label{eq:optimization}
  \min_{\theta \in \Theta} f(\theta) \quad \text{ with } \quad f(\theta)= \mathbb{E}[Q(\theta)]
\end{align}
where $Q$ denotes some stochastic objective function\footnote{We shall
  often omit the dependence on $\omega$.}
$Q: \Omega \times \Theta \to \mathbb{R}$,
$(\omega, \theta) \mapsto Q(\theta)(\omega)$ that depends on
parameters $\theta$ taking values in some space
$\Theta$.

The classical method how to solve generic optimization problems for
some differentiable objective function $f$ (not necessarily of the
expected value form as in \eqref{eq:optimization}) is to apply a
gradient descent algorithm: starting with an initial guess
$\theta^{(0)}$, one iteratively defines
\begin{align}\label{eq:update0}
  \theta^{(k+1)}=\theta^{(k)}- \eta_k \nabla f(\theta^{(k)})
\end{align}
for some learning rate $\eta_k$.  Under suitable assumptions,
$\theta^{(k)}$ converges for $k \to \infty$ to a local minimum of the
function $f$.

In the deep learning context, stochastic gradient
  descent methods, going back to stochastic approximation algorithms
proposed by \cite{RM:51},
are much more efficient. To apply this,
it is crucial that the objective function $f$ is linear in the
sampling probabilities. In other words, $f$ needs to be of the
expected value form as in \eqref{eq:optimization}.  In the simplest
form of stochastic gradient descent, under the assumption that
\[
  \nabla f(\theta) = \mathbb{E}[\nabla Q(\theta)],
\]
the true gradient of $f$ is approximated by a gradient at a single
sample $Q(\theta)(\omega)$ which reduces the computational cost
considerably. In the updating step for the parameters $\theta$ as in
\eqref{eq:update0}, $f$ is then replaced by $ Q (\theta) (\omega_k)$,
hence
\begin{equation}\label{eq:update}
  \theta^{(k+1)}=\theta^{(k)}- \eta_k \nabla Q(\theta^{(k)})(\omega_k).
\end{equation}

The algorithm passes through all samples $\omega_k$ of the so-called
training data set, possibly several times (specified by the number of
epochs), and performs the update until an approximate minimum is~reached.

A compromise between computing the true gradient of $f$ and the
gradient at a single sample $Q(\theta)(\omega)$ is to compute the
gradient of a subsample of size $N_{\text{batch}}$, called
(mini)-batch, so that $Q(\theta^{(k)})(\omega_k)$ used in the update
\eqref{eq:update} is replaced by
\begin{align}\label{eq:stoch_grad}
  Q^{(k)}(\theta)= \frac{1}{N_{\text{batch}}} \sum_{n=1}^{N_{\text{batch}}} Q( \theta) (\omega_{n+kN_{\text{batch}}}), \quad k \in \{0,1,...,\lfloor N/N_{\text{batch}} \rfloor-1\},
\end{align}
where $N$ is the size of the whole training data set.  Any other
unbiased estimators of $\nabla f (\theta) $ can of course also be
applied in \eqref{eq:update}.

\section{Alternative Approaches for Minimizing the Calibration Functional}\label{alt}

We consider here alternative algorithms for minimizing
\eqref{eq:calibration}.

\subsection{Stochastic Compositional Gradient Descent}
One alternative is stochastic compositional gradient descent as
developed e.g.~in~\mbox{\cite{WFH:17}}. Applied to our problem this algorithm
(in its simplest form) works as follows: starting with an initial
guess $\theta^{(0)}$, and $y_j^{(0)}$, $j=1, \ldots, J$ one
iteratively defines
\begin{align*}
  y_j^{(k+1)}    & =(1-\beta_k)y_j^{(k)}+ \beta_k Q_j(\theta^{(k)})(\omega_k)      \quad j=1, \ldots, J,          \\
  \theta^{(k+1)} & = \theta^{(k)}- \eta_k  \sum_{j=1}^J w_j \ell'(y_j^{(k+1)}) \nabla Q_j(\theta^{(k)})(\omega_k)
\end{align*}
for some learning rates $\beta_k,\eta_k \in (0,1]$. Please note that
$y^{(k)}$ is an auxiliary variable to track the quantity
$\mathbb{E}[Q(\theta^{(k)})]$ which has to be plugged in $\ell'$
(other faster converging estimates have also been developed). Of
course $\nabla Q_j(\theta^{(k)})(\omega_k)$ can also be replaced by
other unbiased estimates of the gradient, e.g.~the~gradient of the
(mini)-batches as in \eqref{eq:stoch_grad}. For convergence results in
the case when $ \theta \mapsto \ell(\mathbb{E}[Q_j(\theta)])$ is
convex we refer to \cite[Theorem 5]{WFH:17}. Of course, the same
algorithm can be applied when we replace $Q_j(\theta)$ in
\eqref{eq:calibration} with $X_j(\theta)$ as defined in \eqref{eq:X}
for the variance reduced case.

\subsection{Estimators Compatible with Stochastic Gradient Descent}
Our goal here is to apply at least in special cases of the nonlinear
function $\ell$ (variant \eqref{eq:stoch_grad} of) stochastic gradient
descent to the calibration functional~\eqref{eq:calibration}. This
means that we must cast ~\eqref{eq:calibration} into expected value
form.  We focus on the case when $\ell(x)$ is given by $\ell(x)=x^2$
and write $f(\theta)$ as
\[
  f(\theta)= \sum_{j=1}^{J} w_{j} \EW{ Q_{j}(\theta)
    \widetilde{Q}_j(\theta)}
\]
for some independent copy $\widetilde{Q}_j(\theta)$ of $Q_j(\theta)$,
which is clearly of the expected value form required in~\eqref{eq:optimization}.  A Monte Carlo estimator of $f(\theta)$ is
then constructed by
\begin{align*}
  \widehat{f}(\theta)= \frac{1}{N} \sum_{n=1}^N  \sum_{j=1}^J   w_{j}Q_{j}(\theta)(\omega_n)\widetilde{Q}_{j}(\theta)(\omega_{n}).
\end{align*}
for independent draws $\omega_1, \ldots, \omega_{N}$ (the same $N$
samples can be used for each strike $K_j$). Equivalently~we~have
\begin{align}\label{eq:estimatorQ}
  \widehat{f}(\theta)= \frac{1}{N} \sum_{n=1}^N \sum_{j=1}^J   w_{j}Q_{j}(\theta)(\omega_n)Q_{j}(\theta)(\omega_{n+m}).
\end{align}
for independent draws $\omega_1, \ldots, \omega_{2N}$.  The analog of
\eqref{eq:stoch_grad} is then given by
\begin{align*}
  Q^{(k)}(\theta) & = \frac{1}{N_{\text{batch}}} \sum_{l=1}^{N_{\text{batch}}}
  \sum_{j=1}^J   w_{j}Q_{j}(\theta)(\omega_{l+ 2kN_{\text{batch}}} )Q_{j}(\theta)(\omega_{l+(2k+1)N_{\text{batch}}})
\end{align*}
for $k \in \{0,1,...,\lfloor N/N_{\text{batch}} \rfloor-1\}$.

Clearly we can now modify and improve the estimator by using again
hedge control variates and replace $Q_j(\theta)$ by $X_j(\theta)$ as
defined in \eqref{eq:X}.

\section{Algorithms}\label{app:algo}

In this section, we present the calibration algorithm  discussed above in form of pseudo code given in Algorithm \ref{alg1}. Update rules for parameters in Algorithm \ref{alg1} are provided in Algorithm \ref{spec: SGD}.
We further provide an implementation in form of a github repository, see \url{https://github.com/wahido/neural_locVol}.
\begin{Algorithm}\label{alg1}

  In the subsequent pseudo code, the index $i$ stands for the
  maturities, $N$ for the number of samples used in the variance
  reduced Monte Carlo estimator as of \eqref{eq:calfin} and $k$ for
  the updating step in the gradient descent:

  \vspace{1em}
  \begin{adjustwidth}{2em}{1em}
    \begin{mdframed}[skipabove=11cm,
        backgroundcolor=white,
        linecolor=white,
        innerleftmargin=0.3cm,
        innertopmargin = .5cm,
      ]
      \begin{lstlisting}
# Initialize the network parameters
initialize  $\theta_1,\ldots, \theta_4$
# Define initial number of trajectories and initial step
N, k = 400, 1
# The time discretization for the MC simulations and the
# abort criterion
$\Delta_t$, tol = 0.01, 0.0045

for i = 1,...,4:
  nextslice = False
  # Compute the initial normalized vega weights for this slice:
  $w_j = \tilde w_j/\sum_{l=1}^{20} \tilde w_l$ with  $
  \tilde w_j = 1/v_{ij}$, where $v_{ij}$ $\textcolor{black}{\text{is} }$ the Black-Scholes
  vega $\textcolor{black}{\text{for} }$ strike $K_{ij}$, the corresponding synthetic market implied
  volatility $\textcolor{black}{\text{and} }$ the maturity $T_i$.

while nextslice == False:
  do:
    Simulate $N$ trajectories of the SABR-LSV process up
    to time $T_i$, compute the payoffs.
  do:
    Compute the stochastic integral of the Black-Scholes
    Delta hedge against these trajectories as of $\eqref{eq: BS hedge formula}$
    $\textcolor{black}{\text{for} }$  maturity $T_i$
  do:
    Compute the calibration functional as of $\eqref{eq:calfin}$
    with $\ell(x)=x^2$ $\textcolor{black}{\text{and} }$ weights $w_j$ $\text{ with the modification that we use put}$
    $\text{options instead of call options for strikes larger than the spot.}$
  do:
    Make an optimization step $\textcolor{black}{\text{from} }$ $\theta^{(k-1)}_i$ to $\theta^{(k)}_i$, similarly
    as $\textcolor{black}{\text{in} }$ $\eqref{eq:updateimpl}$ but with the more sophisticated ADAM-
    optimizer with learning rate $10^{-3}$.
  do:
    Update the parameter N, the condition nextslice $\textcolor{black}{\text{and} }$
    compute model prices according to Algorithm $\ref{spec: SGD}$.
  do:
    $k = k+1$
     
    \end{lstlisting}
    \end{mdframed}
  \end{adjustwidth}
\end{Algorithm}

\vspace{1em}

\begin{Algorithm}\label{spec: SGD}
We update the parameters in Algorithm \ref{alg1} according to the following rules:

  \vspace{1em}
 \begin{adjustwidth}{2em}{1em}
    \begin{mdframed}[skipabove=11cm,
        backgroundcolor=white,
        linecolor=white,
        innerleftmargin=0.3cm,
        innertopmargin = .5cm]
      \begin{lstlisting}[gobble=6]
      if k == 500:
        N = 2000
      else if k == 1500:
        N = 10000
      else if k == 4000:
        N = 50000
       
      if k >= 5000 and k mod 1000 == 0:
        do:
          Compute model prices $\pi_\text{model}$ $\textcolor{black}{\text{for slice} }$ $i$ via MC simulation
          using $10^7$ trajectories. Apply the Black-Scholes Delta
          hedge $\textcolor{black}{\text{for} }$ variance reduction.
        do:
          Compute implied volatilities iv$_\text{model}$ $\textcolor{black}{\text{from} }$ the model prices $\pi_\text{model}$.
        do:
          Compute the maximum error of model implied volatilities
          against synthetic market implied volatilities:
          err_cali = $||$ iv_model - iv_market $||_\text{max}$
          if err_cali $\le$ tol or k == 12000:
            nextslice = True
          else:
            Apply the adversarial part: Adjust the weights $w_j$
            according to:

            for j = 1,$\ldots$,20:
              $w_j$ = $w_j$ + $|$ iv_model$_j$ - iv_market$_j$ $|$
            This puts higher weights on the options where the fit
            can still be improved
            Normalize the weights:
            for j = 1,$\ldots$,20:
              $w_j$ = $w_j$ / $\sum_{\ell=1}^{20} w_\ell$
      \end{lstlisting}
    \end{mdframed}
  \end{adjustwidth}
\end{Algorithm}


\begin{thebibliography}{10}
  \bibitem[\protect\citeauthoryear{Abergel and Tachet}{Abergel and
  Tachet}{2010}]{AT:10}
Abergel, Fr{\'e}d{\'e}ric, and R{\'e}mi Tachet. 2010.
\newblock A nonlinear partial integro-differential equation from mathematical
finance.
\newblock {\em Discrete and Continuous Dynamical Systems-Series A\/}~{\em
27\/}: 907--17.

\bibitem[\protect\citeauthoryear{Acciaio and Xu}{Acciaio and Xu}{2020}]{AX:20}
Acciaio, Beatrice, and Tianlin Xu. 2020.
\newblock \emph{Learning Dynamic GANs via Causal Optimal Transport}.
\newblock { Working paper}.

\bibitem[\protect\citeauthoryear{Bayer, Horvath, Muguruza, Stemper, and
  Tomas}{Bayer et~al.}{2019}]{BHMST:19}
Bayer, Christian, Blanka Horvath, Aitor Muguruza, Benjamin Stemper,  and
Mehdi Tomas. 2019.
\newblock On deep calibration of (rough) stochastic volatility models.
\newblock {\em arXiv}. arXiv:1908.08806.

\bibitem[\protect\citeauthoryear{Becker, Cheridito, and Jentzen}{Becker
  et~al.}{2019}]{BCJ:19}
Becker, Sebastian, Patrick Cheridito, and Arnulf Jentzen. 2019.
\newblock Deep optimal stopping.
\newblock {\em Journal of Machine Learning Research\/}~{20} (2019), pp.1--25.



 \bibitem[\protect\citeauthoryear{Buehler, Gonon, Teichmann, and Wood}{B\"uhler
      et~al.}{2019}]{BGTW:19}
  B\"uhler, Hans, Lukas Gonon, Josef  Teichmann, and Ben Wood. 2019.
  \newblock Deep hedging.
  \newblock {\em Quantitative Finance\/}~{19}: 1271--91.

  
\bibitem[\protect\citeauthoryear{Buehler, Horvath, Perez~Arribaz, Lyons, and
  Wood}{B\"uhler et~al.}{2020}]{BHPLW:20}
B\"uhler, Hans, Blanka Horvath, Immanol Perez~Arribaz,  Terry Lyons, and
Ben Wood. 2020.
\newblock A Data-Driven Market Simulator for Small Data Environments.
\newblock {Available online: \url{https://ssrn.com/abstract=3632431} (accessed~on September 22 2020)}.



\bibitem[\protect\citeauthoryear{Carmona and Nadtochiy}{Carmona and
  Nadtochiy}{2009}]{carmona2009local}
Carmona, Ren{\'e}, and Sergey Nadtochiy. 2009.
\newblock Local volatility dynamic models.
\newblock {\em Finance and Stochastics\/}~{13\/}: 1--48.

\bibitem[\protect\citeauthoryear{Carmona, Ekeland, Kohatsu-Higa, Lasry, Lions,
  Pham, and Taflin}{Carmona et~al.}{2007}]{carmonainbook}
Carmona, Rene, Ivar Ekeland, Arturo Kohatsu-Higa, Jean-Michel Lasry, 
Pierre-Louis Lions,   Huyen Pham, and Erik Taflin. 
\newblock {\em HJM: A Unified Approach to Dynamic Models for Fixed Income,
Credit and Equity Markets}.  Berlin/Heidelberg:
Springer, vol. 1919, pp.\  1--50.
\newblock
doi:{\changeurlcolor{black}\href{https://doi.org/10.1007/978-3-540-73327-0_1}{\detokenize{10.1007/978-3-540-73327-0_1}}}.


\bibitem[\protect\citeauthoryear{Cont and Ben Hamida}{Cont and
  Ben Hamida}{2004}]{cont2004recovering}
Cont, Rama, and Sana Ben Hamida. 2004.
\newblock Recovering volatility from option prices by evolutionary
optimization.
\newblock {\em Journal of Computational Finance\/}~{8:} 43--76.

\bibitem[\protect\citeauthoryear{Cozma, Mariapragassam, and Reisinger}{Cozma
  et~al.}{2017}]{CMR:17}
Cozma, Andrei, Matthieu Mariapragassam, and Christoph Reisinger. 2019.
\newblock Calibration of a hybrid local-stochastic volatility stochastic rates
model with a control variate particle method.
\newblock {\em SIAM Journal on Financial Mathematics\/}
~{10:} 181--213.


\bibitem[\protect\citeauthoryear{Cuchiero, Marr, M., Mitoulis, Singh, and
  Teichmann}{Cuchiero et~al.}{2018}]{CMMMST:18}
Cuchiero, Christa, Alexia Marr, Milusi Mavuso, Nicolas Mitoulis, Aditya Singh, and Josef Teichmann.   2018.
\newblock \emph{Calibration of Mixture Interest Rate Models with Neural Networks}.
\newblock {Technical report}.

\bibitem[\protect\citeauthoryear{Cuchiero, Schmocker, and Josef}{Cuchiero
  et~al.}{2020}]{CST:20}
Cuchiero, Christa, Philipp Schmocker, and Teichmann Josef. 2020.
\newblock \emph{Deep Stochastic Portfolio Theory}. 
\newblock {Working paper}.

\bibitem[\protect\citeauthoryear{Cybenko}{Cybenko}{1989}]{C:92}
Cybenko, George. 1992.
\newblock Approximation by superpositions of a sigmoidal function.
\newblock {\em Mathematics Control, Signal and Systems\/}~{2}: 303--14.

\bibitem[\protect\citeauthoryear{Dupire}{Dupire}{1994}]{D:94}
Dupire, Bruno. 1994.
\newblock  Pricing with a smile.
\newblock {\em Risk\/}~{7}: 18--20.

\bibitem[\protect\citeauthoryear{Dupire}{Dupire}{1996}]{D:96}
Dupire, Bruno. 1996.
\newblock A unified theory of volatility.
\newblock In~{\em Derivatives Pricing: The Classic Collection\/}. {London: Risk Books}: 185--96.


\bibitem[\protect\citeauthoryear{Eckstein and Kupper}{Eckstein and
  Kupper}{2019}]{EK:19}
Eckstein, Stephan, and Michael Kupper. 2019.
\newblock Computation of optimal transport and related hedging problems via
penalization and neural networks.
\newblock {\em Applied Mathematics \& Optimization\/} 1--29. doi:10.1007/s00245-019-09558-1.

\bibitem[\protect\citeauthoryear{Gao, Tu, and Xu}{Gao et~al.}{2019}]{GTX:19}
Gao, Xiaojie, Shikui Tu, and Lei Xu. 2019.
\newblock A tree search for portfolio management.
\newblock {\em arXiv}.  arXiv:1901.01855.

\bibitem[\protect\citeauthoryear{Gatheral, Jaisson, and Rosenbaum}{Gatheral
  et~al.}{2018}]{GJR:18}
Gatheral, Jim, Thibault Jaisson,  and  Mathieu Rosenbaum. 2018.
\newblock Volatility is rough.
\newblock {\em Quantitative Finance\/}~{18:} 933--49.
\newblock
doi:{\changeurlcolor{black}\href{https://doi.org/10.1080/14697688.2017.1393551}{\detokenize{10.1080/14697688.2017.1393551}}}.

\bibitem[\protect\citeauthoryear{Gierjatowicz, Sabate, Siska, Szpruch, Zuric}{Gierjatowicz et~al.}{2020}]{GSSS:20}
Gierjatowicz, Patryk, Mark Sabate, David Siska, and Lukasz Szpruch. 2020.
\newblock Robust pricing and hedging via neural SDEs. 
\newblock {Available online: \url{https://papers.ssrn.com/sol3/papers.cfm?abstract_id=3646241} (accessed~on September 22 2020)}.


\bibitem[\protect\citeauthoryear{Goodfellow, Pouget-Abadie, Mirza, Xu,
  Warde-Farley, Ozair, Courville, and Bengio}{Goodfellow
  et~al.}{2014}]{GPMXWOCB:14}
Goodfellow, Ian, Jean Pouget-Abadie, Mehdi Mirza, Bing Xu, David Warde-Farley, 
Sherjil Ozair, Aaron Courville, and Yoshua Bengio. 2014.
\newblock Generative adversarial nets.
\newblock In {\em Advances in Neural Information Processing Systems}.  Cambridge: The MIT Press, pp.~2672--80.


\bibitem[\protect\citeauthoryear{Guyon and Henry-Labordere}{Guyon and
  Henry-Labordere}{2012}]{GH:11}
Guyon, Julien, and Pierre Henry-Labord\`ere. 2012.
\newblock Being particular about calibration.
\newblock {\em Risk\/}~{25}: 92--97.

\bibitem[\protect\citeauthoryear{Guyon and Henry-Labord{\`e}re}{Guyon and
Henry-Labord{\`e}re}{2013}]{GH:13}
Guyon, Julien, and Pierre Henry-Labord{\`e}re. 2013.
\newblock {\em Nonlinear Option Pricing}.
\newblock  {Chapman \& Hall/CRC Financial Mathematics Series}. 

\bibitem[\protect\citeauthoryear{Guyon}{Guyon}{2014}]{G:14}
Guyon, Julien. 2014.
\newblock Local correlation families.
\newblock {\em Risk\/}~{(February)}: 52--58.

\bibitem[\protect\citeauthoryear{Guyon}{Guyon}{2016}]{G:16}
Guyon, Julien. 2016.
\newblock Cross-dependent volatility.
\newblock  {\em Risk\/}~{29}: 61--65.


\bibitem[\protect\citeauthoryear{Han, Jentzen, and Weinan}{Han
  et~al.}{2017}]{HJE:17}
Han, Jiequn, Arnulf  Jentzen, and Weinan E. 2017.
\newblock Overcoming the curse of dimensionality: Solving high-dimensional
partial differential equations using deep learning.
\newblock {\em arXiv}. arXiv:1707.02568.

\bibitem[\protect\citeauthoryear{Hecht-Nielsen}{Hecht-Nielsen}{1992}]{hecht1992theory}
Hecht-Nielsen, Robert. 1989.
\newblock Theory of the backpropagation neural network.
\newblock Paper presented at  International 1989 Joint Conference on Neural Networks,   Washington, DC, USA, vol. 1, pp. 593--605. doi: 10.1109/IJCNN.1989.118638.

\bibitem[\protect\citeauthoryear{Heiss, Teichmann, and Wutte}{Heiss
  et~al.}{2019}]{HTW:19}
Heiss, Jakob, Josef Teichmann, and Hanna Wutte. 2019.
\newblock How implicit regularization of neural networks affects the learned
function--part i.
\newblock {\em arXiv}. arXiv:1911.02903.

\bibitem[\protect\citeauthoryear{Henry-Labordere}{Henry-Labordere}{2019}]{H:19}
Henry-Labord\`ere, Pierre. 2019.
\newblock Generative Models for Financial Data. 
\newblock {Preprint, Available online: \url{https://ssrn.com/abstract=3408007} (accessed~on 22 September 2020)}

\bibitem[\protect\citeauthoryear{Hernandez}{Hernandez}{2017}]{H:16}
Hernandez, Andres. 2017.
\newblock Model calibration with neural networks.
\newblock {\em Risk}. doi:10.2139/ssrn.2812140.


\bibitem[\protect\citeauthoryear{Hornik}{Hornik}{1991}]{H:91}
Hornik, Kurt. 1991.
\newblock Approximation capabilities of multilayer feedforward networks.
\newblock {\em Neural Networks}~{4:} 251--57.

\bibitem[\protect\citeauthoryear{Hur{\'e}, Pham, Bachouch, and
  Langren{\'e}}{Hur{\'e} et~al.}{2018}]{HPBL:18}
Hur{\'e}, C{\^o}me, Huy{\^e}n Pham, Achref Bachouch, and Nicolas Langren{\'e}. 
2018.
\newblock Deep neural networks algorithms for stochastic control problems on
finite horizon, part i: Convergence analysis.
\newblock {\em arXiv}. arXiv:1812.04300.

\bibitem[\protect\citeauthoryear{Hur{\'e}, Pham, and Warin}{Hur{\'e}
  et~al.}{2019}]{HPW:19}
Hur{\'e}, C{\^o}me, Huy{\^e}n Pham, and Xavier Warin. 2019.
\newblock Some machine learning schemes for high-dimensional nonlinear
{P}{D}{E}s.
\newblock {\em arXiv}. arXiv:1902.01599.

\bibitem[\protect\citeauthoryear{Jex}{Jex}{1999}]{Jex:99}
Jex, Mark,  Robert Henderson, and David Wang. 1999.
\newblock Pricing Exotics under the Smile.
\newblock {\em Risk Magazine\/}~{12}: 72--75.

\bibitem[\protect\citeauthoryear{Jourdain and Zhou}{Jourdain and
  Zhou}{2016}]{JZ:16}
Jourdain, Benjamin, and Alexandre Zhou. 2016.
\newblock Existence of a calibrated regime switching local volatility model and
new fake brownian motions.
\newblock {\em arXiv}. arXiv:1607.00077.

\bibitem[\protect\citeauthoryear{Kondratyev and Schwarz}{Kondratyev and
  Schwarz}{2019}]{KS:19}
Kondratyev, Alexei, and Christian Schwarz. 2019.
\newblock The Market Generator.
Available online: \url{https://ssrn.com/abstract=3384948} ({accessed on 22 September 2020}).

\bibitem[\protect\citeauthoryear{Lacker,  Shkolnikov and Zhang}{Lacker et.~al.}{2019}]{Lacker:19}
Lacker, Dan, Misha Shkolnikov, and Jiacheng Zhang. 
2019.
\newblock Inverting the Markovian projection, with an application to local
stochastic volatility models.
\newblock {\em arXiv}. arXiv:1905.06213.


\bibitem[\protect\citeauthoryear{Lipton}{Lipton}{2002}]{L:02}
Lipton, Alexander. 2002.
\newblock The vol smile problem.
\newblock {\em Risk Magazine}~{15}: 61--65.

\bibitem[\protect\citeauthoryear{Liu, Borovykh, Grzelak, and Oosterlee}{Liu
  et~al.}{2019a}]{liu2019neural}
Liu, Shuaiqiang,  Anastasia Borovykh, Lech Grzelak, and  Cornelis Oosterlee. 
2019a.
\newblock A neural network-based framework for financial model calibration.
\newblock {\em Journal of Mathematics in Industry}~9: 9.

\bibitem[\protect\citeauthoryear{Liu, Oosterlee, and Bohte}{Liu
  et~al.}{2019b}]{liu2019pricing}
Liu, Shuaiqiang,  Cornelis Oosterlee, and  Sander Bohte. 2019b.
\newblock Pricing options and computing implied volatilities using neural
networks.
\newblock {\em Risks}~{7}: 16.

\bibitem[\protect\citeauthoryear{Potters, Bouchaud and Sestovic}{Potters et~al.}{2001}]{potters:2001}
Potters, Marc, Jean-Philippe Bouchaud, and Dragan Sestovic. 2001
\newblock Hedged Monte-Carlo: Low Variance Derivative Pricing with Objective Probabilities.
\newblock {\em Physica A: Statistical Mechanics and Its Applications}~ {289}: 517--25.

\bibitem[\protect\citeauthoryear{Protter}{Protter}{1990}]{pro:90}
Protter, Philip. 1990.
\newblock {\em Stochastic Integration and Differential Equations}. Volume~21 of
{\em Applications of Mathematics (New~York)}. A new approach.
\newblock Berlin: Springer.


\bibitem[\protect\citeauthoryear{Ren, Madan, and Qian}{Ren
  et~al.}{2007}]{RMQ:07}
Ren, Yong, Dilip Madan, and Michael Qian Qian. 2007.
\newblock Calibrating and pricing with embedded local volatility models.
\newblock {\em London Risk Magazine Limited-}~{20}: 138.

\bibitem[\protect\citeauthoryear{Robbins and Monro}{Robbins and
  Monro}{1951}]{RM:51}
Robbins, Herbert, and Sutton Monro. 1951.
\newblock A stochastic approximation method.
\newblock {\em The Annals of Mathematical Statistics} 22: 400--407.

\bibitem[\protect\citeauthoryear{Ruf and Wang}{Ruf and Wang}{forthcoming}]{RW:19}
Ruf, Johannes, and Weiguan Wang. {Forthcoming}.
\newblock Neural networks for option pricing and hedging: A literature review.
\newblock {\em {Journal of Computational Finance}}.

\bibitem[\protect\citeauthoryear{Samo and Vervuurt}{Samo and
  Vervuurt}{2016}]{SV:16}
Samo, Yves-Laurent Kom, and Alexander Vervuurt. 2016.
\newblock Stochastic portfolio theory: A machine learning perspective.
\newblock {\em arXiv.}  arXiv:1605.02654.

\bibitem[\protect\citeauthoryear{Saporito, Yang, and Zubelli}{Saporito
  et~al.}{2017}]{SYZ:17}
Saporito, Yuri~F, Xu Yang, and Jorge Zubelli. 2019.
\newblock The calibration of stochastic-local volatility models-an inverse
problem perspective.
\newblock {\em Computers \& Mathematics with Applications\/}~{77}: 3054--67.




\bibitem[\protect\citeauthoryear{Sirignano and Cont}{Sirignano and
  Cont}{2019}]{SC:19}
Sirignano, Justin and Cont, Rama. 2019.
\newblock Universal features of price formation in financial markets:
perspectives from deep learning.
\newblock {\em Quantitative Finance\/}~{19}: 1449--59.

\bibitem[\protect\citeauthoryear{Tian, Zhu, Lee, Klebaner, and Hamza}{Tian
  et~al.}{2015}]{TZLKH:15}
Tian, Yu, Zili Zhu, Geoffrey Lee, Fima Klebaner, and Kais Hamza. 2015.
\newblock Calibrating and pricing with a stochastic-local volatility model.
\newblock {\em Journal of Derivatives}~{22}: 21.

\bibitem[\protect\citeauthoryear{Vidales, Siska, and Szpruch}{Vidales
  et~al.}{2018}]{VSSS:18}
Vidales, Marc-Sabate, David Siska, and Lukasz Szpruch. 2018.
\newblock Unbiased deep solvers for parametric pdes.
\newblock {\em arXiv}.  arXiv:1810.05094.

\bibitem[\protect\citeauthoryear{Wang, Fang, and Liu}{Wang
  et~al.}{2017}]{WFH:17}
Wang, Mengdi, Ethan~X Fang, and Han Liu. 2017.
\newblock Stochastic compositional gradient descent: algorithms for minimizing
compositions of expected-value functions.
\newblock {\em Mathematical Programming\/}~{161\/}: 419--49.

\bibitem[\protect\citeauthoryear{Wiese, Bai, Wood, and Buehler}{Wiese
  et~al.}{2019}]{WBWB:19}
Wiese, Magnus, Lianjun Bai, Ben Wood, and Hans B\"uhler. 2019.
\newblock Deep Hedging: Learning to Simulate Equity Option Markets.
Available online: \url{https://ssrn.com/abstract=3470756} ({accessed on 20 September 2020}).

\end{thebibliography}
\end{document}